\documentclass[a4paper,reqno,11pt]{amsart}
\numberwithin{equation}{section}
\usepackage{amsfonts}
\usepackage{amsmath,amscd,amssymb}
\usepackage{graphicx,epsfig}
\usepackage[enableskew]{youngtab}
\usepackage[notref,notcite,final]{showkeys}
\usepackage{ytableau}
\usepackage{fourier}
\usepackage{tikz-cd}

\usepackage{hyperref}

\newtheorem{theorem}{Theorem}
\theoremstyle{plain}
\newtheorem*{acknowledgement}{Acknowledgement}

\newtheorem{corollary}[theorem]{Corollary}
\newtheorem{definition}[theorem]{Definition}
\newtheorem{example}[theorem]{Example}
\newtheorem{lemma}[theorem]{Lemma}

\newtheorem{proposition}[theorem]{Proposition}
\newtheorem{remark}[theorem]{Remark}

\numberwithin{theorem}{section}
\numberwithin{equation}{section}
\numberwithin{figure}{section}

\newcommand{\mb}[1]{\mathbb{#1}}
\newcommand{\mc}[1]{\mathcal{#1}}

\newcommand{\op}[1]{\operatorname{#1}}

\newcommand{\defeq}{\overset{\text{def}}{=}}

\newcommand{\cF}{\mathcal{F}}
\newcommand{\cH}{\mathcal{H}}
\newcommand{\cP}{P}

\newcommand{\cT}{\mathcal{T}}
\newcommand{\cR}{\mathcal{R}}

\newcommand{\cV}{\mathcal{V}}

\newcommand{\cZ}{\mathcal{Z}}

\newcommand{\h}{\mathfrak{h}}
\newcommand{\wt}{\operatorname{wt}}
\newcommand{\prob}{\Omega}
\newcommand{\End}{\operatorname{End}}
\newcommand{\Cl}{\operatorname{Cl}}

\newcommand{\ch}{\operatorname{ch}}
\newcommand{\0}{\varnothing}

\newcommand{\N}{\mathbb{N}}

\newcommand{\Z}{\mathbb{Z}}
\newcommand{\C}{\mathbb{C}}

\usepackage{hyperref}

\begin{document}
\title[Cylindric Hecke characters and Gromov-Witten invariants]{Cylindric Hecke characters and Gromov-Witten invariants\\ via the asymmetric six-vertex model }
\author{Christian Korff}
\address{School of Mathematics and Statistics, Glasgow G12 8QQ, UK}
\email{christian.korff@glasgow.ac.uk}
\date{11 June 2019}

\begin{abstract}
We construct a family of infinite-dimensional positive sub-coalgebras within the Grothendieck ring of Hecke algebras, when viewed as a Hopf algebra with respect to the induction and restriction functor. These sub-coalgebras have as structure constants the 3-point genus zero Gromov-Witten invariants of Grassmannians and are spanned by what we call cylindric Hecke characters, a particular set of virtual characters for whose computation we give several explicit combinatorial formulae. One of these expressions is a generalisation of Ram's formula for irreducible Hecke characters and uses cylindric broken rim hook tableaux. We show that the latter are in bijection with so-called `ice configurations' on a cylindrical square lattice, which define the asymmetric six-vertex model in statistical mechanics. A key ingredient of our construction is an extension of the boson-fermion correspondence to Hecke algebras and employing the latter we find new expressions for Jing's vertex operators of Hall-Littlewood functions in terms of the six-vertex transfer matrices on the infinite planar lattice.
\end{abstract}

\maketitle

\section{Introduction}\noindent
Positivity phenomena attract a lot of attention within mathematics as they usually point towards links between different areas such as combinatorics, representation theory and geometry, and it has proved very fruitful in the past to investigate these connections. Combinatorial Hopf algebras are one particular example where such connections can be observed. Probably the simplest and best-studied example of a combinatorial Hopf algebra is the ring of symmetric functions $\Lambda=\lim\limits_{\longleftarrow}\C[y_1,\ldots,y_k]^{S_k}$  \cite{zelevinsky1981}, where the structure constants with respect to the multiplication of two Schur functions are the Littlewood-Richardson coefficients, certain non-negative integers. For the latter there exist combinatorial interpretations in terms of polytopes, a representation theoretic interpretation in terms of tensor multiplicites of the general linear group, and a geometric interpretation in terms of intersection numbers of Schubert varieties; see e.g. \cite{fulton1997}. 

In another development, there have been several applications of exactly solved models in statistical mechanics \cite{baxter2016} to problems in combinatorics and enumerative geometry; see e.g. \cite{zinn2009six} and references therein. Two prominent examples are the proof of the alternating sign matrix conjecture via the six-vertex model \cite{kuperberg1996another} and the combinatorics underlying the Razumov-Stroganov conjecture \cite{razumov2004combinatorial}. The latter works revived the connection between statistical lattice models and combinatorics which goes back to the early works of Kasteleyn \cite{kasteleyn1963dimer} and Temperley-Fisher \cite{temperley1961dimer} in the 1960s.

In this article we shall combine these two strands, positivity phenomena and statistical mechanics models, by focussing on certain positive sub-coalgebras within the ring of symmetric functions $\Lambda$. This will be done via Gromov-Witten invariants of Grassmannians, which are a known geometric generalisation of Littlewood-Richardson numbers. They count algebraic curves intersecting Schubert varieties, and we are interested in a combinatorial  algebraic structure, where these invariants are interpreted as structure constants. It turns out that this can be achieved by considering certain linear subspaces of $\Lambda$ which are closed with respect to the coproduct \cite{korff2018positive}: instead of multiplying Schur functions, one considers the separation of variables for so-called cylindric generalisations of Schur functions which form a subcoalgebra in $\Lambda$. In order to describe the underlying combinatorics we identify the cylindric skew Schur function with the statistical mechanics partition function of the asymmetric six-vertex model on the cylinder. This yields a richer algebraic structure, which underlies an isomorphism between the ring of symmetric functions and the Grothendieck ring of Hecke algebras of type A. The long term aim is to generalise our approach to Hecke algebras of other types and, thus, to Gromov-Witten invariants of a larger class of varieties. 

Denote by $\cH_m=\cH_m(t)$ the Iwahori-Hecke algebra of type $A$ defined over $\C(t)$ with $t$ an indeterminate. %
We shall recall the definition of $\cH_m$ in the text, for the moment it suffices to think of the Hecke algebra as a $t$-deformation of the symmetric group algebra $\C[S_m]$. Let $\cR^m(t)$ be the Grothendieck group of finite-dimensional $\cH_m$-modules and $\cR(t)=\bigoplus_{m\ge 0}\cR^m(t)$ the associated Grothendieck ring. Throughout this article we shall  identify each $M\in\cR^m(t)$ with its corresponding character $\chi^M_t\in\op{Hom}_{\C(t)}(\cH_m,\C(t))$. Similar to the case of the symmetric group, the irreducible $\cH_m$-characters $\chi_t^\mu$ are labelled by partitions $\mu\in\cP^+$ with $\mu\vdash m$. Fix some integers $n\ge 2$ and $0\le k\le n$ and denote by $\cP_{k,n}^+\subset\cP^+$ the subset of partitions satisfying the constraints $\lambda_1\le n-k$, $\ell(\lambda)\le k$, where $\ell(\lambda)$ is the number of nonzero parts. 

Using techniques from exactly solvable models in statistical mechanics, we shall construct an infinite set of virtual Hecke characters 
$$\cR_{k,n}(t)=\Bigl\{\chi_t^{\lambda[d]}=\sum_{\mu\in\cP^+}c_\mu\chi_t^\mu~\Bigl |\Bigr.~d\in\Z_{\ge 0},\;\lambda\in\cP_{k,n}^+,\;
c_\mu=0,\pm1\Bigr\}\subset\cR(t)\;,$$ 
which naturally arises in  the computation of the partition function of the asymmetric six-vertex model, describing ice and ferroelectrics on a cylindrical lattice of circumference $n$. The integer $k$ is the number of `down spins'. The symbol $\lambda[d]$ denotes a so-called cylindric loop, an infinite periodic continuation of the outline of the Young diagram of $\lambda$ viewed as a lattice path in $\Z\times\Z$ and shifted $d$ times in the direction $(-k,n-k)$. All these notions will be further explained in the text. The noteworthy property of these virtual characters is that they span a positive infinite-dimensional $\Z$-coalgebra in $\cR(t)$ with respect to the restriction functor $\op{Res}:\cR(t)\to\cR(t)\otimes\cR(t)$. Our main result is the following: 
\begin{theorem}\label{thm:main}
Let $\lambda\in\cP^+_{k,n}$, $d\geq 0$ an integer and set $m=|\lambda|+dn$. Then for any decomposition $m=m'+m''$ we have
\begin{equation}\label{chi2GW}
\op{Res}^{\cH_m}_{\cH_{m'}\otimes\cH_{m''}}\chi_t^{\lambda[d]}=
\sum_{(\mu,d')}\;\,\sum_{(\nu,d'')}C^{\lambda,d-d'-d''}_{\mu\nu}\chi_t^{\mu[d']}\otimes\chi_t^{\nu[d'']}\,,
\end{equation}
where $C^{\lambda,d}_{\mu\nu}\in\Z_{\ge 0}$ are the 3-point genus 0 Gromov-Witten invariants of the Grassmannian $\op{Gr}_k(\C^n)$ of $k$-hyperplanes in $\C^n$ and  the sums run over all 
pairs $(\mu,d'),(\nu,d'')\in\cP^+_{k,n}\times\Z_{\ge 0}$ such that $d'+d''\le d$, $|\mu|=m'-d'n$, $|\nu|=m''-d''n$. 
\end{theorem}
For $d=0$ the expansion \eqref{chi2GW} specialises to the familiar decomposition rule for the irreducible characters $\chi^\lambda_t=\chi^{\lambda[0]}_t$ with $C^{\lambda,0}_{\mu\nu}=c^{\lambda}_{\mu\nu}$ being the Littlewood-Richardson coefficients. In addition to \eqref{chi2GW} we state an explicit combinatorial formula for the direct computation of the virtual characters $\chi^{\lambda[d]}_t$ with $d>0$ in terms of cylindric broken rim hook tableaux that generalises Ram's formula \cite{ram1991} for the irreducible Hecke characters $\chi^{\lambda[0]}_t=\chi^\lambda_t$.  Alternatively, the characters can be obtained by the weighted counting of ice configurations (see Figures \ref{fig:BFcorr2h2o} and \ref{fig:h2o2tab} for examples) using the six-vertex model with a particular choice of Boltzmann weights. Within the physics community working on integrable or exactly solvable
lattice models, the {\em symmetric} six-vertex model is one of the prototypical systems which is related to the
quantum XXZ magnet or Heisenberg spin-chain \cite{baxter2016}. In this article we focus on the lesser studied {\em asymmetric} six-vertex model instead.

To motivate the result \eqref{chi2GW} let us first recall the case of ordinary cohomology $H^*(\op{Gr}_{k}(\C^n))$: one combinatorial method of computing the Littlewood-Richardson coefficient $c^{\lambda}_{\mu\nu}$, the intersection number of three hyperplanes in general position, is to consider (semi-standard) tableaux of shape $\lambda$ and decompose these into sub-tableaux of shape $\mu\subset\lambda$ and skew shape $\lambda/\mu$. Counting the number of skew tableaux that rectify under Sch\"utzenberger's jeu-de-taquin to an arbitrary but fixed tableau of shape $\nu$ then gives $c^{\lambda}_{\mu\nu}$; see e.g. the textbook \cite[\S 5.1,Cor 1 \& Prop 2]{fulton1997} for details. Algebraically, this method of calculating  intersection numbers corresponds to computing the coproduct of Schur functions $s_\lambda$ in the ring of symmetric functions $\Lambda$. There exists a well-known (Hopf algebra) isomorphism which identifies the ring of symmetric functions with the Grothendieck ring of symmetric groups, the so-called characteristic map \cite[Ch.I.7]{macdonald1998}, and a perhaps somewhat lesser known $t$-deformation $\ch_t:\cR(t)\to\Lambda(t)=\Lambda\otimes_\C\C(t)$ of this map for the Grothendieck ring $\cR(t)$ of Hecke algebras \cite{fomin1998E,wan2015}. Under these isomorphisms the computation of the coproduct of Schur functions becomes the computation of the restriction functor acting on an irreducible character in $\cR(t)$. 

Theorem \ref{thm:main} is a generalisation of this algebraic approach to Gromov-Witten invariants, since the cohomology $H^*(\op{Gr}_{k}(\C^n))$ of Grassmannians can be viewed as a finite-dimensional sub-coalgebra of the infinite-dimensional coalgebra \eqref{chi2GW}. 


\subsection{Asymmetric Ice and the Boson-Fermion Correspondence}
For the combinatorial computation of the virtual characters $\chi_t^{\lambda[d]}$ in Theorem \ref{thm:main} we employ a bijection between broken rim hook tableaux and ice configurations on a cylindrical lattice, the so-called asymmetric six-vertex model. We start, however, with the infinite lattice and then project on the cylinder in the final step which will facilitate some of the computations involved.

In the infinite lattice limit, and under suitable boundary conditions, the transfer matrices of the asymmetric six-vertex model are related to the Geck-Rouquier central elements \cite{geck1997} spanning $\cZ(t)=\bigoplus_{m\ge 0}\cZ(\cH_m(t))$, where $\cZ(\cH_m(t))$ is the centre of the Hecke algebra $\cH_m(t)$. The proof exploits the commutative diagram \eqref{tBFcd} of algebra isomorphisms given below, which is a $t$-deformation of an analogous diagram explaining why the character table of the symmetric groups provides a change of basis between Schur functions and power sums in the ring of symmetric functions; see e.g. \cite[Ch.I.7]{macdonald1998}.

\begin{equation}\label{tBFcd}
\begin{tikzcd}
\cR(t) \arrow[r,"\sim"] \arrow[d,"\sim"] \arrow[dr,"\ch_t"] & \cZ(t) \arrow[d,"F_t"]\\
\bigwedge^{\frac{\infty}{2},c}V(t)\arrow[r,"\imath\otimes 1\;\;"] & \Lambda(t)
\end{tikzcd}
\end{equation}
The isomorphism $\cR(t)\overset{\sim}{\to}\cZ(t)$ on the top of the diagram \eqref{tBFcd} is given via the trace map, which sends each $M\in\cR^m(t)$ to its trace function $\chi_t^M=\op{Tr}_M\in\op{Hom}_{\C(t)}(\cH_m,\C(t))$ which in turn defines an element in $\cZ(\cH_m)$ by fixing the coefficients with respect to the Geck-Rouqier basis in $\cZ(\cH_m)$. The map $F_t:\cZ(t)\to\Lambda(t)$ is a $t$-extension of the Frobenius map defined in \cite{wan2015}, which sends the Geck-Rouquier elements to certain symmetric functions which are a $t$-deformation of the power sums, while the quantum characteristic map $\ch_t:\cR(t)\to\Lambda(t)$ sends the irreducibles to the basis of Schur functions in $\Lambda(t)=\Lambda\otimes_\C\C(t)$. Thus, the expansion of Schur functions into the symmetric functions that are the image of the Geck-Rouquier central elements then yields the character table $\{\chi_t^\lambda~|~\lambda\vdash m\}$ of the Hecke algebra $\cH_m(t)$, where the partitions $\lambda\vdash m$ label the irreducible modules in $\cR^m(t)$. 

There is an alternative, `fermionic', version of the upper right triangle of the diagram \eqref{tBFcd}: on the left side of the diagram we identify the partitions $\lambda$ labelling the irreducible Hecke characters in $\cR(t)$ with so-called Maya diagrams $\sigma(\lambda,c):\Z\to\{0,1\}$ of fixed charge $c\in\Z$, a binary sequence satisfying certain boundary conditions at infinity. The linear span of these Maya diagrams is isomorphic to the infinite wedge space $\bigwedge^{\frac{\infty}{2},c}V$, which is called the fermionic Fock space in the physics literature \cite{jimbo1983solitons}.  Each such Maya diagram is mapped to a `spin-configuration' of the six-vertex model: a 1-letter corresponding to a spin pointing down, a 0-letter to a spin pointing up. The transfer matrices of the six-vertex model therefore act naturally on this space forming a commutative subalgebra in $\End \bigwedge^{\frac{\infty}{2},c}V(t)$ with $\bigwedge^{\frac{\infty}{2},c}V(t)=\bigwedge^{\frac{\infty}{2},c}V\otimes_\C\C(t)$. In the next step we then use the boson-fermion correspondence $\imath: \bigwedge^{\frac{\infty}{2},c}V\to\Lambda$ to map each Maya diagram fixed by a partition $\lambda$ to the Schur function $s_\lambda$. Changing base to $\C(t)$ and using the quantum characteristic map $\ch_t:\cR(t)\to\Lambda(t)$ from \eqref{tBFcd} we prove that its image, the Schur function $s_\lambda$, is the partition function of the asymmetric six-vertex model at the so-called free fermion point.

Let $\langle\,\cdot\,,\,\cdot\,\rangle:\bigwedge^{\frac{\infty}{2},c}V(t)\otimes \bigwedge^{\frac{\infty}{2},c}V(t)\to
\bigwedge^{\frac{\infty}{2},c}V(t)$ be the unique bilinear form with respect to which the Maya diagrams are an orthonormal basis.  Fix an   infinite alphabet $X=x_1+x_2+\cdots$ of commuting indeterminates, the so-called `spectral parameters' of the six-vertex model, which are related to the Miwa variables $Y=y_1+y_2+\cdots$ of the ring of symmetric functions $\Lambda$ by the plethystic substitution $Y=(t-1)X$. Then we have the following combinatorial description of the boson-fermion correspondence for Hecke algebras:
 
\begin{proposition}\label{prop:tBFcorr}
The algebra isomorphism $\imath_t=\imath\otimes 1:\bigwedge^{\frac{\infty}{2},c}V\otimes_\C\C(t)\to\Lambda\otimes_\C\C(t)$ takes the explicit form
\begin{equation}\label{tBFcorrespondence}
\imath_t: \sigma(\lambda,c)\mapsto s_{\lambda}[Y]=\langle\sigma(\lambda,c),\prod_{i\ge 1}A(x_i;t) \sigma(\0,c)\rangle=
\sum_{\alpha\vdash m}\chi^{\lambda}_t(\alpha)(t-1)^{\ell(\alpha)}x^\alpha,
\end{equation}
where $A(x_i,t)=\sum_{r\ge 0}x_i^rA_r(t)$ is the  asymmetric six-vertex (row-to-row) transfer matrix\footnote{A precise definition will be given later in the text; see \eqref{HeckeA} and \eqref{BFA}. For experts familiar with the algebraic Bethe ansatz, the operator in question is an infinite lattice version of the familiar $A$-operator from the Yang-Baxter algebra, which plays the role of the transfer matrix under a particular choice of boundary conditions.}, i.e. the matrix element is a weighted sum over ice configurations in a single lattice row.  
\end{proposition}

The proof of the proposition rests on a bijection between six-vertex lattice configurations and broken rim hook tableaux, examples of which are shown in Figures \ref{fig:BFcorr2h2o} and \ref{fig:h2o2tab}. There exists by now a plethora of different combinatorial applications of the six-vertex model, see e.g. \cite{zinn2009six}, but to the best of the author's knowledge the bijection discussed in this article is new. In particular, we prove that the action of the $A$-operator in the fermionic Fock space is described by the same $t$-extension of the Murnaghan-Nakayama rule as the one for Hecke characters derived in \cite{ram1991}. The latter is a generalisation of the Murnaghan-Nakayama rule that allows one to recursively compute the values of irreducible characters $\chi^{\lambda }$ of the symmetric group $S_{m}$: fix a cycle $\sigma \in S_{m}$ of length $\ell $ and let $\pi \in S_{m}$ be of cycle type $\mu \vdash m'=m-\ell $ containing the remaining $m-\ell $ letters on which $S_{m}$ acts by permutations. Then 
\begin{equation}
\chi ^{\lambda }(\pi \sigma )=
\sum_{\nu\vdash m'}\chi^{\lambda/\nu}(\sigma)\chi ^{\nu}(\pi )=
\sum_{\nu\vdash m'}(-1)^{r(\lambda /\nu )-1}\chi ^{\nu}(\pi )\;,  \label{MNrule}
\end{equation}%
where the sum runs over all partitions $\nu \vdash m'$ such that the skew
diagram $\lambda /\nu $ is a connected rim hook or border strip (the precise
definition will be given in the text below) and $r(\lambda /\nu )$ is the
number of rows it occupies. In the case of the Hecke algebra the above rule becomes $t$-deformed and one has to allow for broken, disconnected, rim hooks as well. It is only at the level of the Hecke algebra that one can fully see the connection with the asymmetric six-vertex model, albeit a degenerate version of it, describing cylindric versions of symmetric group characters \cite{korff2018positive}, can be defined in the $t\to1$ limit.

\begin{figure}\label{fig:BFcorr2h2o}
\centering
\includegraphics[width=.85\textwidth]{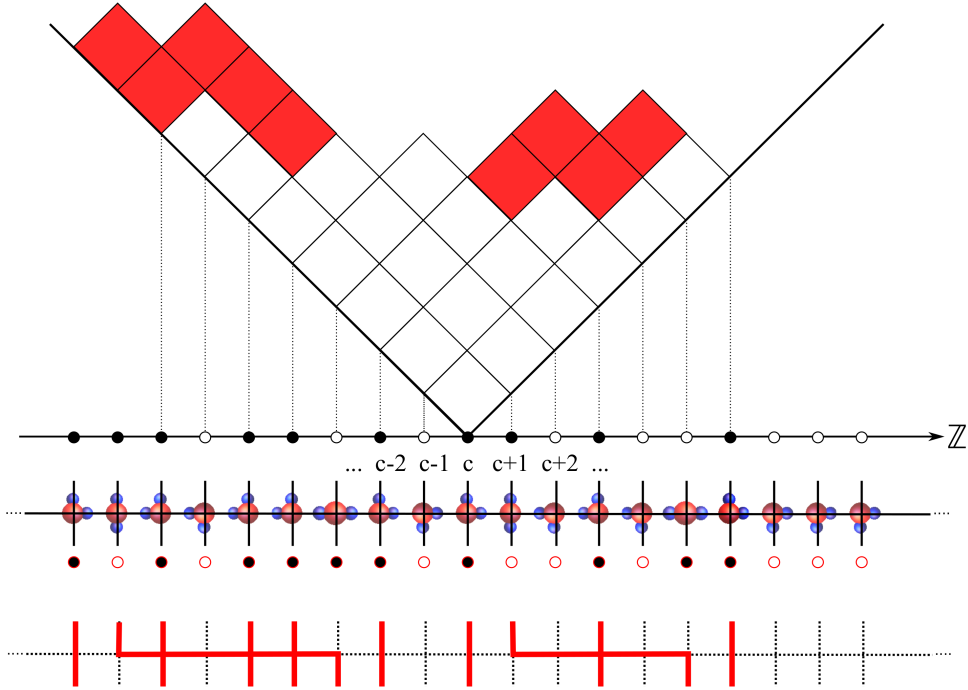} 
\caption{A depiction of the boson-fermion correspondence. Shown is the bijection between Maya diagrams (a fermion or spin configuration on the infinite line where black bullets represent particles and white bullets are `holes') and Young diagrams of partitions. Adding a row configuration of water molecules, `ice', to the Maya diagram, adds a broken rim hook (red squares) to the corresponding partition. At the very bottom the same ice configuration is displayed in terms of coloured edges. }
\end{figure}

As a `by-product' of the proof of Proposition \ref{prop:tBFcorr} we derive novel expressions for Jing's vertex-operators $\Phi^\pm(x;t):\Lambda(t)\to\Lambda(t)$ describing Hall-Littlewood functions at generic $t$ \cite{jing1991} in terms of the six-vertex transfer matrix. Vertex operators play an important role in the representation theory of Kac-Moody algebras and conformal field theory. The precise connection is as follows: define another, $t$-deformed, bilinear form $\langle\,\cdot\,,\,\cdot\,\rangle_t:\bigwedge^{\frac{\infty}{2},c}V(t)\otimes \bigwedge^{\frac{\infty}{2},c}V(t)\to\bigwedge^{\frac{\infty}{2},c}V(t)$ by `pulling back' the scalar product $\langle s_\lambda,s_\mu\rangle_t=\langle s_\lambda[(1-t)Y],s_\mu[Y]\rangle$ from $\Lambda(t)$ to $\bigwedge^{\frac{\infty}{2},c}V(t)$ via the isomorphism $\imath_t$ in \eqref{tBFcorrespondence}. 
\begin{proposition}\label{prop:A2VO}
Let $A(x;t)$ be the transfer matrix of the asymmetric six-vertex model from \eqref{tBFcorrespondence} and denote by $A^{-1}(x;t)$ and $A^*(x;t)$ its inverse and adjoint with respect to the inner product $\langle\,\cdot\,,\,\cdot\,\rangle_t$, respectively. Then we have the identities
\[
\Phi^-(x;t)\circ\iota_t=\iota_t\circ A^{-1}(x;t)\circ A^*(x;t)
\quad\text{and}\quad
\Phi^+(x;t)\circ\iota_t=\iota_t\circ A(x;t)\circ(A^{-1}(x;t))^*\,,
\]
where $\imath_t:\bigwedge^{\frac{\infty}{2},c}V(t)\to\Lambda(t)$ is the isomorphism defined in \eqref{tBFcorrespondence}. 
 \end{proposition}
While connections between the {\em symmetric} six-vertex model and vertex operators are part of the Kyoto School approach \cite{jimbo1994algebraic} to the computation of correlation functions, where one heavily exploits the underlying quantum group symmetry, we stress that this algebraic structure is not available here, since we are dealing with the {\em asymmetric} six-vertex model and consider solutions to the Yang-Baxter equations which are not derived from intertwiners of quantum group representations but instead have a geometric origin, so-called convolution algebras \cite{gorb2020yangbaxter}. 

The vertex operators $\Phi^\pm(x;t)$ we will consider in this article play an important role in several areas of mathematics: their specialisation at $t=0$ yields the vertex operators of the free boson conformal field theory with central charge $c=1$, which form the simplest example of a vertex operator algebra. They also simplify the study of Hall-Littlewood functions and connected geometric representation theory; see e.g. \cite{haiman2002combinatorics} and references therein. Finally, there are connections with classical integrable hierarchies, systems of non-linear PDEs that have soliton solutions \cite{jimbo1983solitons}: at $t=-1$ one obtains the vertex operators connected to Schur's Q-functions. Similar to how Schur functions ($t=0$) are polynomial solutions to the KP hierarchy \cite{date2000solitons}, Schur's Q-functions ($t=-1$) are solutions to the BKP hierarchy \cite{date1982transformation}. Proposition \ref{prop:A2VO} thus establishes a direct link between the asymmetric six-vertex model and these classical integrable hierarchies. This connection we plan to explore  further in future work. In this article, we limit ourselves to deriving new simple (fermionic) formulae for the vertex operators $\Phi^\pm(x;t)$ in terms of the six-vertex transfer matrices as we do not wish the present discussion to distract from the main result, which is Theorem \ref{thm:main}.

\begin{figure}
\centering
\includegraphics[width=.85\textwidth]{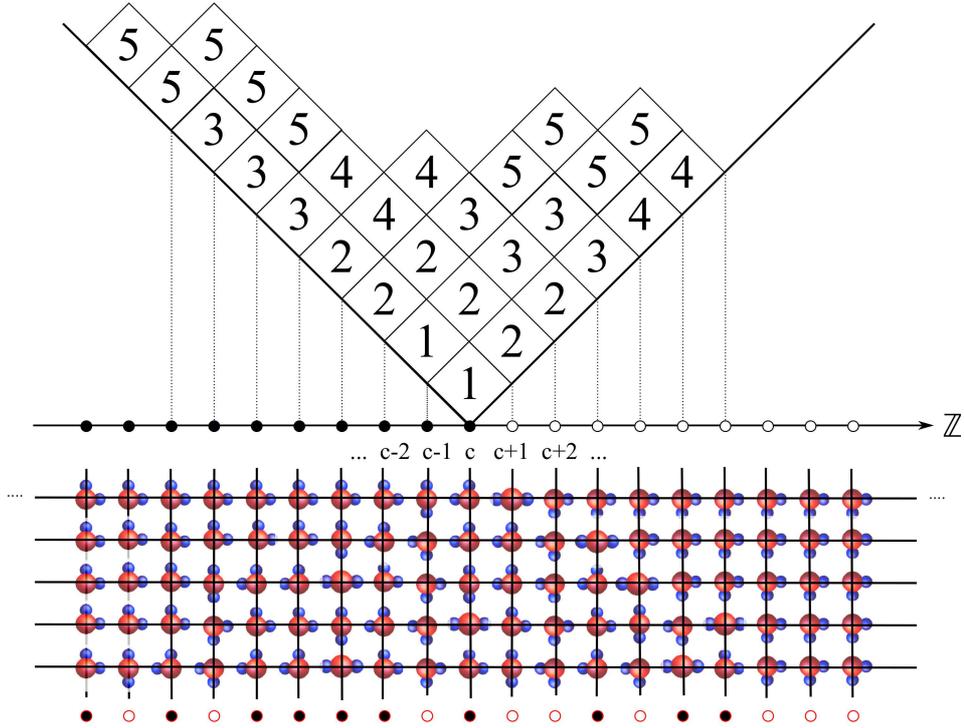}
\caption{A configuration of water molecules, `ice', on a square lattice. As explained in the text each such configuration can be mapped to a unique broken rim hook tableau and the latter are used in the computation of irreducible Hecke characters.}\label{fig:h2o2tab}
\end{figure}

\subsection{Quasi-periodic boundary conditions and small quantum cohomology}
Quantum cohomology had its origin in mathematical physics \cite{gepner1991fusion,intriligator1991fusion,vafa1991topological,witten1993verlinde} in connection with fusion rings of Wess-Zumino-Witten conformal field theories before becoming a subject of mathematical study in its own right. One of the first and best studied examples is the quantum cohomology of Grassmannians. Denote by $qH^*(\op{Gr}_k(\C^n))$ the small quantum cohomology ring of the Grassmannian $\op{Gr}_k(\C^n)$, the variety of $k$-hyperplanes in $\C^n$, whose structure constants are the Gromov-Witten invariants appearing in Theorem \ref{thm:main}. There is a known presentation of this ring as a quotient of the ring of symmetric functions due to Siebert and Tian \cite{siebert1997}. In this article, we exploit the so-called rim hook algorithm \cite{bertram1999} to describe the projection $\pi_{k,n}:\Lambda\twoheadrightarrow qH^*(\op{Gr}_k(\C^n))$. The latter, in conjunction with the boson-fermion correspondence, induces a fermionic analogue of the rim hook algorithm, $\pi'_{k,n}:\bigwedge^{\frac{\infty}{2},k}V\twoheadrightarrow\bigoplus_{d\ge 0}q^d\otimes\bigwedge^k\C^n$, such that the following diagram commutes:
\begin{equation}\label{qBFcd}
\begin{tikzcd}
\bigwedge^{\frac{\infty}{2},k}V\otimes_{\C}\C(t) \arrow[r, "\imath\otimes 1"] \arrow[d,two heads,"\pi'_{k,n}\otimes 1"] & 
\Lambda(t)=\Lambda\otimes_{\C}\C(t)
\arrow[d,two heads,"\pi_{k,n}\otimes 1"]\\
(\bigoplus\limits_{d\ge 0}q^d\otimes\bigwedge^k\C^n)\otimes_\C\C(t)\arrow[r,"\varphi\otimes 1"] 
& qH^*(\op{Gr}_k(\C^n))\otimes_{\C}\C(t)
\end{tikzcd}
\end{equation}
The isomorphism $\varphi:\bigoplus_{d\ge 0}q^d\otimes\bigwedge^k\C^n\to qH^*(\op{Gr}_k(\C^n))$ at the bottom of the diagram is the simplest case of the Satake correspondence which maps finite wedge products to Schubert classes; see \cite{golyshev2011quantum} and references therein. 

In terms of the asymmetric six-vertex model the projection $\pi'_{k,n}$ corresponds to changing from the infinite square lattice to a cylindrical lattice of circumference $n$ and with quasi-periodic boundary conditions, where the quantum parameter $q$ of quantum cohomology is identified with the quasi-periodicity parameter of the lattice model; similar as it has been discussed previously in \cite{korff2014quantum} for certain five-vertex degenerations of the six-vertex model where one Boltzmann weight is set to zero. The dimension $k$ of the hyperplanes is fixed by the number of fermions or down spins, which is left invariant under the action of the transfer matrix. The resulting cylindrical versions of the Murnaghan-Nakayama rules then define -- via the inverse characteristic map -- the virtual character set $\cR_{k,n}(t)$. In fact, we will introduce the following cylindrical analogue of the boson-fermion correspondence \eqref{tBFcorrespondence}, $\jmath_{k,n}:\bigwedge^k\C^n(t)\otimes_\C\C[\!\![q]\!\!]\to\Lambda(t)\otimes_\C\C[\!\![q]\!\!]$,
\begin{equation}
\jmath_{k,n}: v_\lambda\mapsto\langle v_\lambda,\prod_{i\ge 1}H(x_i;t) v_\0\rangle
=\sum_{d\ge 0}q^d\sum_{\alpha}\chi^{\lambda[d]}_t(\alpha)(t-1)^{\ell(\alpha)}x^\alpha,
\end{equation}
where $\{v_\lambda~:~\lambda\in\cP^+_{k,n}\}$ is the pre-image of the basis of Schubert classes under the Satake correspondence $\varphi$ in \eqref{qBFcd} and $H(x;t)$ is the six-vertex transfer matrix for the periodic lattice (up to an important normalisation factor) satisfying, $\pi'_{k,n}\circ A(x;t)=H(x;t)\circ \pi'_{k,n}$. The matrix element has the physical interpretation of being the partition function of the asymmetric six-vertex model on the infinite cylinder and is mathematically a formal power series in the quantum parameter $q$ whose coefficients are symmetric functions in the spectral variables $x_i$.

The expansion formula \eqref{chi2GW} is proved by using the Bethe ansatz, a well-established technique in quantum integrable systems, which allows us to describe the spectrum of the transfer matrices $H$ for the cylindrical lattice. Namely, using that the eigenbasis of the asymmetric six-vertex transfer matrices coincides with the basis of idempotents in $qH^*(\op{Gr}_k(\C^n))$, one shows that the partition functions for the cylindrical lattice with quasi-periodic boundary conditions are $t$-deformations of the cylindric Schur functions considered in \cite{gessel1997,lam2006,mcnamara2006} and \cite{korff2018positive}. We show that the pre-image of these cylindric Schur functions with respect to the quantum characteristic map $\ch_t$ in \eqref{tBFcd} are the virtual characters in $\cR_{k,n}(t)$ of Theorem \ref{thm:main}.

\subsection{Outline of the article}
Section 2 reviews some of the preliminary results and combinatorial notions needed for the discussion, such as the boson-fermion correspondence, the quantum characteristic map and the computation of Hecke characters via broken rim hook tableaux. 

Section 3 generalises the boson-fermion correspondence to Hecke algebras and states explicit `fermionic' expressions of the operators $A,A^{-1}$ in Proposition \ref{prop:A2VO} before identifying them with so-called bosonic `half-vertex operators' acting on the ring of symmetric functions. The latter are the image of the basis dual to the Geck-Rouqier central elements under the Frobenius map in the diagram \eqref{tBFcd}.

Section 4 introduces the asymmetric six-vertex model as a combinatorial tool and shows that the operator $A$ is a six-vertex transfer matrix on the infinite lattice with special boundary conditions. It also discusses the underlying solutions of the Yang-Baxter equation and the resulting Yang-Baxter algebras, which are described in terms of broken rim hook tableaux. The discussion is then extended to quasi-periodic boundary conditions to motivate the definition of cylindric broken rim hook tableaux. The section ends with a discussion of the eigenvalue problem of the transfer matrix for quasi-periodic boundary conditions using the Bethe ansatz. The latter leads to a set of polynomial equations whose quotient ring is the small quantum cohomology of Grassmannians.  The main result states that the six-vertex transfer matrix with quasi-periodic boundary conditions corresponds to multiplication by certain ($t$-deformed) linear combinations of Chern classes of the tautological and quotient bundle in  $qH^*(\op{Gr}_k(\C^n))$. 

Section 5 gives the definition of cylindric Hecke characters in terms of the asymmetric six-vertex model on the cylinder. We show that the latter are virtual characters in $\cR(t)$ and compute their co-product proving Theorem \ref{thm:main}. We also prove that the cylindric Hecke characters are mapped to cylindric Schur functions under the quantum characteristic map.
\begin{acknowledgement}\rm
The author wishes to thank Gwyn Bellamy, Sira Gratz and Greg Stevenson for sharing knowledge and valuable discussions as well as the anonymous reviewer whose detailed comments helped to improve this article.
\end{acknowledgement}
\smallskip
\noindent{\bf Notation}. Throughout this article tensor products $\otimes$ are always understood to be tensor products over the complex numbers, $\otimes_\C$, unless stated otherwise.

\section{Combinatorial Preliminaries}
In order to keep this article self-contained we briefly review the boson-fermion correspondence and some connected combinatorial notions. We then recall known formulae for the computation of Hecke characters and the quantum characteristic map from \eqref{tBFcd}.

\subsection{Maya diagrams and the fermionic Fock space}
A {\em Maya diagram} is an infinite binary string $\sigma:\Z\to\{0,1\}$ such that there exists integers $n_\pm(\sigma)\in\Z$ with $\sigma_i=1$ for all $i\le n_-$ and $\sigma_i=0$ for all $i\ge n_+$. We call
\begin{equation}\label{c}
c(\sigma)=n_-+\sum_{i>n_-}\sigma_i=n_+-\sum_{i\le n_+}(1-\sigma_i) 
\end{equation}
the {\em charge} of the Maya diagram. 
The set of Maya diagrams of fixed charge $c\in\Z$ is in bijection with the set of partitions $\cP^+$: given a partition $\lambda=(\lambda_1,\ldots,\lambda_\ell,0,\ldots)$ define the Maya diagram $\sigma(\lambda)$ by setting
 \begin{equation}\label{Maya}
\sigma_i( \lambda,c)=\left\{
\begin{array}{ll}
1, & \text{ if } i=c+1+\lambda_j-j,\;j\in\N\\
0, & \text{ else}
\end{array}
\right.
 \end{equation}
In particular, the empty partition $\varnothing=(0,0,\ldots)$ corresponds to the Maya diagram $\sigma_i=1$ for $i\le c$ and $\sigma_i=0$ for $i>c$. We adopt the common notations $|\sigma(\lambda)|\defeq|\lambda |$ for the weight of a partition, i.e. the sum of its parts, and $\ell(\sigma)\defeq\ell (\lambda )$ for its length, i.e. the number of nonzero parts. Note that $\ell(\lambda)=\sum_{i>n_-}\sigma_i(\lambda)$  and $\lambda_1=\sum_{i<n_+}(1-\sigma_i(\lambda))$. Conversely, fix $\lambda$ and $c\in\Z$, then we have that $n_-(\sigma(\lambda,c))=c-\ell(\lambda)$ and $n_+(\sigma(\lambda,c))=c+1+\lambda_1$. A graphical depiction of the bijection \eqref{Maya} is shown in Figure \ref{fig:BFcorr2h2o}. 

The {\em fermionic Fock space} is defined as the direct sum $\cF=\bigoplus_{c\in\Z}\cF_c$, where $\cF_c$ is the formal $\C$-linear span of Maya diagrams $\sigma:\Z\to\{0,1\}$ of charge $c$ (not their pointwise addition). The space $\cF$ is naturally endowed with an action of the Clifford algebra $\Cl$ with generators $\{\psi_i^{\pm}~:~i\in\Z\}$ and relations
\begin{equation}\label{Cl}
\psi^{\pm}_i\psi^{\pm}_j+\psi^{\pm}_j\psi^{\pm}_i=0,\qquad
\psi^-_i\psi^+_j+\psi_j^+\psi^-_i=\delta_{ij}\;.
\end{equation}
Namely, define maps $\psi^\pm_i:\cF_c\to\cF_{c\pm1}$ by setting
\begin{equation}\label{Cl2F}
\psi_i^+\sigma=
\left\{\begin{array}{ll}
(\sigma+\epsilon_i)\prod_{j>i}(-1)^{\sigma_j}, & \sigma_i=0\\
0, & \sigma_i=1
\end{array}
\right.\,,
\quad
\psi_i^-\sigma=
\left\{\begin{array}{ll}
(\sigma-\epsilon_i)\prod_{j>i}(-1)^{\sigma_j}, & \sigma_i=1\\
0, & \sigma_i=0
\end{array}
\right.
\end{equation}
where $\epsilon_i:\Z\to\{0,1\}$ is the map $j\mapsto \delta_{ij}$ and the maps $\sigma\pm\epsilon_i:\Z\to\{0,1\}$ are defined via pointwise summation, $(\sigma\pm\epsilon_i)(j)=\sigma_j\pm\delta_{ij}$. In particular, both are well-defined Maya diagrams. In words, modulo a sign factor,  acting with $\psi_i^-$ on a Maya diagram changes a one-letter at position $i$ into a zero-letter or, if there is none, gives the null vector. Similarly, $\psi^+_i$ changes a zero-letter at position $i$ into a one-letter. 
\begin{remark}\rm
Instead of using Maya diagrams it is often customary to identify the basis elements in the fermionic Fock space $\cF_c$ with `semi-infinite wedge products'. Namely, let $\bigwedge^{\frac{\infty}{2},c}V$ with $V=\bigoplus_{i\in\Z}\C v_i$ be the $\C$-linear span of wedge products of the form
\[
v_{i_1}\wedge v_{i_2}\wedge \cdots\wedge v_{i_r}\wedge v_{c-r-1}\wedge v_{c-r-2}\wedge \cdots
\]
for some $r\in\Z_{\ge 0}$ and fixed charge $c\in\Z$. This is the notation used in the introduction. These wedge products should be understood as formal symbols which are antisymmetric under the exchange of the basis vectors $v_i$ of $V$. Given a Maya diagram $\sigma\in\cF_c$ we define the map $\sigma\mapsto v_{i_1(\sigma)}\wedge v_{i_2(\sigma)}\wedge \cdots$, where the set of integers $i_1(\sigma)>i_2(\sigma)>\ldots$ is given by the positions of 1-letters in the Maya diagram, i.e. $\sigma(i_j)=1$ for all $j\in\N$. This bijection between Maya diagrams and wedge products induces a vector space isomorphism $\cF_c\cong\bigwedge^{\frac{\infty}{2},c}V$. The action of the Clifford algebra \eqref{Cl2F} is now more easily described by the familiar actions 
$$\psi^+_i v_{i_1}\wedge v_{i_2}\wedge\cdots=v_i\wedge v_{i_1}\wedge v_{i_2}\wedge\cdots$$ and 
$$\psi^-_i v_{i_1}\wedge v_{i_2}\wedge\cdots =\sum_{r\ge 1}\delta_{ii_r}(-1)^{r-1}v_{i_1}\wedge v_{i_2}\wedge \cdots v_{i_{r-1}}
\wedge v_{i_{r+1}}\wedge\cdots,$$
where the $r$th factor in the wedge product on the right hand side has been omitted. In this article we have used the language of Maya diagrams instead in order to elucidate the connection with the six-vertex model.
\end{remark}

\subsection{The Boson-Fermion Correspondence}
Denote by $\h$ the Heisenberg algebra with generators $\{p_r,p_{-r}\}_{r\in\N}$ and relations
\begin{equation}\label{Heisenberg}
[p_r,p_s]=-r\delta_{r+s,0}\;.
\end{equation}
As usual, we will identify the commutative subalgebra $\h^+=\C[p_1,p_2,\ldots]\subset\h$ with the ring of symmetric functions $\Lambda=\lim\limits_{\longleftarrow}\Lambda_k$, where $\Lambda_k=\C[y_1,y_2,\ldots,y_k]^{S_k}$, by mapping its generators $p_r$ to the power sums,  $p_r\mapsto p_r[Y]\overset{\text{def}}{=}\sum_{i\ge 0}y_i^r$. Here $Y=y_1+y_2+\cdots$ is some infinite auxiliary alphabet of commuting indeterminates, the so-called Miwa variables. Recall that the set $\{p_\lambda=p_{\lambda_1}p_{\lambda_2}\ldots~|~\lambda\in\cP^+\}$ forms a $\mathbb{Q}$-basis of $\Lambda$ \cite{macdonald1998}. The following %
bilinear form $\Lambda\otimes\Lambda\to\C$, 
\begin{equation}\label{Hall}
\langle p_\lambda,p_\mu\rangle=z_\lambda\delta_{\lambda\mu}\;,
\qquad z_\lambda\overset{\text{def}}{=}\prod_{i\ge 1}i^{m_i(\lambda)}m_i(\lambda)!\;,
\end{equation}
is known as the {\em Hall inner product}. 
We let $\h$ act on $\Lambda$ by identifying for $r>0$ the $p_r$ as multiplication operators and $p_{-r}\mapsto r\partial/\partial p_r$ as differential operators. %
Following the literature we call $\h^+\cong\Lambda$ the {\em bosonic Fock space}.

It is well-known that the following operators on the fermionic Fock space $\cF_c$ of fixed charge $c\in\Z$,
\begin{equation}\label{H2Cl}
P_r=\sum_{i\in\Z}\psi_i^-\psi^+_{i+r}\;,\qquad r\in\Z\backslash\{0\}
\end{equation}
also define a representation $\rho_c:\h\to\End\cF_c$ of the Heisenberg algebra by mapping $p_r\mapsto P_r$. Note that $P_{-r}\sigma(\varnothing,c)=0$ for all $r\in\N,c\in\Z$ and, hence, the representation is highest weight. In fact, all the $P_r$ are locally nilpotent, i.e. for each $\sigma\in\cF$ there exists $N=N_\sigma\in\N$ such that $P_r^N\sigma=0$. 

 Fix an inner product on $\cF$ by setting 
\begin{equation}\label{fermiHall}
\langle\sigma,\sigma'\rangle=\prod_{i\in\Z}\delta_{\sigma_i\sigma'_i}
\end{equation} 
for any two Maya diagrams $\sigma,\sigma'$. Then $(\psi_i^+)^*=\psi_i^-$ and $P_r^*=P_{-r}$. Let $z$ be some indeterminate. An essential part of the boson-fermion correspondence is the following statement:
\begin{theorem}
The linear map $\imath:\cF\to\C[z,z^{-1}]\otimes\C[p_1,p_2,\ldots]$ defined via
\begin{equation}\label{BFiso}
v\mapsto \sum_{c\in\Z}z^c\otimes\langle v,e^{H[Y]}\sigma(\0,c)\rangle\;,
\end{equation}
where $H[Y]=\sum_{r>0}(p_r[Y]/r) P_{r}$ is called the `Hamiltonian', is an isomorphism of $\h$-modules. That is, we have the identity $\imath\circ P_r=p_r\circ\imath$ for all $r\neq 0$. Moreover, $\imath$ is an isometry with respect to the inner products \eqref{Hall} and \eqref{fermiHall}.
\end{theorem}
In light of the isomorphism \eqref{BFiso} it is convenient to enlarge the Heisenberg algebra $\h$ by the central element $p_0\overset{\text{def}}{=}z(\partial/\partial z)$ and introduce on $\cF$ the `charge operator' $P_0$ satisfying
\[
\iota\circ P_0=p_0\circ\iota,\qquad P_0\overset{\text{def}}{=}\sum_{i>0}\psi_i^+\psi_i^--\sum_{i\leq 0}\psi^-_i\psi^+_i\;.
\]
\begin{remark}\rm
The variables $t_r=p_r/r$ are interpreted as generalised time parameters in the context of the KP hierarchy, which explains why one calls $H$ the Hamiltonian. Moreover, the image of the Maya diagram $\sigma(\lambda,c)$ under the boson-fermion correspondence is the Schur function $z^c\otimes s_\lambda[Y]$ which is a known polynomial solution of the KP equation \cite{date2000solitons}. In fact, all polynomial solutions are linear combinations of Schur functions. Expressing each Schur polynomial $s_\lambda$ in terms of the power sums gives the irreducible characters $\chi^\lambda$ of the symmetric group $S_m$ with $m=|\lambda|$. The map $\chi^\lambda\mapsto s_\lambda$ constitutes a ring isomorphism $\cR\to\Lambda$ known as the {\em characteristic map}, where $\cR=\bigoplus_m\cR^m$ is the Grothendieck ring finite-dimensional modules of the symmetric groups; see \cite[Ch.I.7]{macdonald1998}.
\end{remark}

\subsection{The quantum characteristic map}
Following the exposition in \cite{wan2015} the following is a brief summary of the connection between the centres of Hecke algebras and the ring of symmetric functions. Let $t$ be an indeterminate, the `deformation parameter', then the Hecke algebra $\cH_m=\cH_m(t)$ is the $\C(t)$-algebra generated by $\{T_1,\ldots,T_{m-1}\}$ subject to the relations
\begin{equation}\label{Hecke}
T_i^2=(t-1)T_i+t,\quad T_i T_{i+1} T_i = T_{i+1} T_i T_{i+1},\quad T_i T_j=T_j T_i\;\text{ if }|i-j|>1\;.
\end{equation}
Because of the latter relations the algebra $\cH_m$ can be viewed as a $t$-deformation of the symmetric group algebra $\C[S_m]$.  Given a permutation $w\in S_m$, let $w=s_{i_1}\ldots s_{i_r}$ be a reduced expression in terms of the elementary transpositions $s_i\in S_m$, then we set $T_w=T_{i_1}\ldots T_{i_r}$. The relations \eqref{Hecke} ensure that the element $T_w$ is independent of the choice of the reduced expression for $w$. Moreover, the elements $\{T_w\}_{w\in S_m}$ form a basis of $\cH_m$. 

It is well-known that the finite-dimensional representations of the Hecke algebras $\cH_m$ carry the structure of a Hopf algebra. Namely, let $\cR(t)=\bigoplus_{m\ge 0}\C(t)\otimes_{\Z}\cR^m(t)$, where $\cR^m(t)$ is the Grothendieck group of finite-dimensional $\cH_m$-modules $M$ and we set $\cR^0(t)=\C(t)$. Define a graded algebra structure on $\cR(t)$ by introducing a product $\cR^m(t)\times\cR^n(t)\to\cR^{m+n}(t)$  via the induction functor
\begin{equation}\label{Ind}
(M,N)\mapsto\op{Ind}_{\cH_{m}\otimes\cH_{n}}^{\cH_{m+n}}M\otimes N\,,
\end{equation}
where $M$ is a $\cH_m$-module, $N$ a $\cH_n$-module and one uses the natural embedding $S_{m}\times S_{n}\hookrightarrow S_{m+n}$ to identify $\cH_{m}\otimes\cH_{n}$ as a subalgebra in $\cH_{m+n}$.\footnote{Namely, define a map $\cH_m\otimes\cH_n\hookrightarrow\cH_{m+n}$ by setting $T_i\otimes 1\mapsto T_i$ for $i=1,\ldots,m-1$ and $1\otimes T_{i}\mapsto T_{i+m}$ for $i=1,\ldots,n-1$.} %
In fact, $\cR(t)$ can be turned into a graded Hopf algebra by defining a  co-product $\Delta:\cR^m(t)\to\bigoplus_{m'+m''=m}\cR^{m'}(t)\otimes\cR^{m''}(t)$ via the restriction functor, 
\begin{equation}\label{Res}
\Delta(M)=\bigoplus_{m'+m''=m}\op{Res}^{\cH_m}_{\cH_{m'}\otimes\cH_{m''}}M\,.
\end{equation}
In what follows, we will identify elements in $M\in\cR^m(t)$ with their trace functions $\chi^M_t\in\op{Hom}_{\C(t)}(\cH_m,\C(t))$. 
Set $T^{\vee}_w=t^{-\ell(w)}T_{w^{-1}}$ with $\ell(w)$ denoting the length of the permutation $w\in S_m$.  Then given a trace function $\chi_t\in\cR^m(t)$ we assign to it the element $\sum_{w\in S_m}\chi_t(T_w)T^{\vee}_w$ in the centre $\cZ(\cH_m)$. This map defines a bijection $\cR^m(t)\to\cZ(\cH_m)$ and its extension $\cR(t)\to\cZ(t)$ with $\cZ(t)=\bigoplus_{m\geq 0}\cZ(\cH_m(t))$ then induces a (Hopf) algebra structure on $\cZ(t)$ by demanding that the latter is a (Hopf) algebra isomorphism. Here we have set $\cZ(\cH_0)=\C(t)$.

In order to formulate an analogue of the Frobenius map for $\cH_m$, one first has to fix a basis in each of the centres $\cZ(\cH_m)$.  The latter is a $t$-deformation of the basis of class sums in the centre $\cZ(\C S_m)$ of the symmetric group algebra. Recall that the conjugacy classes $C_\alpha\subset S_m$ consist of all permutations of fixed cycle type $\alpha=(\alpha_1,\ldots,\alpha_\ell)$ where $\alpha$ is a partition of $m$. Define the following minimal length element in $C_\alpha$,
\begin{equation}\label{minC}
w_\alpha=(1,\ldots ,\alpha _{1})(1+\alpha _{1},\ldots ,\alpha_{1}+\alpha_{2})\cdots (1+\sum_{i<k }\alpha_{i},\ldots ,m)\;.
\end{equation}
For any $w\in S_m$ there exist $f_{\alpha}(w)\in\C[t,t^{-1}]$ such that $T_w\equiv\sum_{\alpha\vdash m}f_{\alpha}(w)T_{w_\alpha}\mod[\cH_m,\cH_m]$ \cite[Thm 5.1]{ram1991}. The following  elements, known as {\em Geck-Rouquier central elements},
\[
c_\alpha=\sum_{w\in C_\alpha}f_{\alpha}(w)T^{\vee}_w\,.
\]
are well-defined and can be shown to specialise to the class sums $\sum_{w\in C_\alpha}w\in\cZ(\C S_m)$ in the limit $t\to 1$.
\begin{theorem}[Geck-Rouquier]
The set $\{c_\alpha~|~\alpha\vdash m\}$ forms a basis of $\cZ(\cH_m)$.
\end{theorem}
Other choices of bases and the transition matrices between them can be found in \cite{lascoux2006, wan2015} and in the references cited therein. 

Denote by $\Lambda^m(t)=\Lambda^m\otimes\C(t)$, where $\Lambda^m$ is the set of homogeneous symmetric functions of degree $m$. Following \cite{wan2015} define the quantum Frobenius map $F_t:\cZ(t)\to\Lambda(t)$ with $\Lambda(t)=\bigoplus_{m\ge0}\Lambda^m(t)$ as
\begin{equation}\label{tFrobenius}
F_t:c_\alpha\mapsto (t-1)^{\ell(\alpha)}m_{\alpha}\left[\frac{Y}{t-1}\right]\,,
\end{equation}
where (using {\em plethystic notation}\footnote{The plethystic substitution $X\to X/(t-1)$ in terms of power sums is the unique ring isomorphism determined by $p_r[Y]\to p_r[Y/(t-1)]=p_{r}(y_1,y_2,\ldots)/(t^r-1)$.}) $m_\alpha[Y/(t-1)]$ denotes the monomial symmetric function with the alphabet $Y=y_1+y_2+\cdots$, replaced by the alphabet $X=Y/(t-1)$; see e.g. \cite{macdonald1998}. 
\begin{theorem}[Wan-Wang]
The quantum Frobenius map $F_t:\cZ(t)\to\Lambda(t)$ is an algebra isomorphism.
\end{theorem} 
Note that the Frobenius map is `degree preserving', that is $F_t(\cZ(\cH_m))=\Lambda^m(t)$. The {\em quantum characterictic map} $\ch_t:\cR(t)\to\Lambda(t)$ is then defined as the unique map such that the following diagram commutes for all $m\ge 0$,
\begin{equation}\label{chart}
\begin{tikzcd}[ arrows={-stealth}]
  \cR^m(t) \rar["\sim"] \drar[shift right, swap, "\ch_t"] & \cZ(\cH_m) \dar[shift right, "F_t"] \\
  & \Lambda^m(t)
\end{tikzcd}\;.
\end{equation}

\subsection{Hecke characters and rim hook tableaux}
In order to describe the quantum characteristic map $\ch_t$ explicitly, we now recall the following $t$-extension of the Frobenius formula, which is due to Ram \cite{ram1991}.

Given two partitions $\lambda ,\mu \in\cP^{+}$ such that $\mu \subset
\lambda $ we say that the skew Young diagram $\lambda /\mu $ is a
(connected) \emph{rim hook} $h$ if it consists of a sequence of squares
along the SE corner of the Young diagram of $\lambda $ such that two
consecutive squares $s,s^{\prime }\in h$ share exactly one edge and $\lambda
/\mu $ does not contain a $2\times 2$ block of squares. The
length $|h|$ of a rim hook is simply the number of squares it contains.

A \emph{broken rim hook} $b$\ is a finite sequence of rim hooks such that $%
h,h^{\prime }\in b$ share at most a common corner but not an edge. Denote by $|b|=\sum_{h\in b}|h|$ the length of
a broken rim hook.

Define a (skew rim hook) tableau $\mathcal{T}$ of shape $\lambda/\mu$ and content $\alpha$ to be a sequence of
partitions $(\lambda ^{(0)},\lambda ^{(1)},\ldots ,\lambda ^{(\ell )})$ such
that $\lambda^{(0)}=\mu$, $\lambda^{(\ell)}=\lambda$ and each $\lambda ^{(i)}/\lambda ^{(i-1)}$ is a broken rim hook $b_{i}$ of length $\alpha_i$. We can interpret such a sequence as a map $\mathcal{T}:\lambda /\mu \rightarrow \{1,\ldots
,\ell \}$ by filling the squares in $\lambda ^{(i)}/\lambda ^{(i-1)}$ with the integer $i$.

Given a broken rim hook tableau $\cT$ define the following weight function via %
\begin{equation}
\wt(\mathcal{T})\overset{\text{def}}{=}\prod_{i=1}^\ell\wt%
(\lambda^{(i)}/\lambda^{(i-1)})\;,  \label{tabweight}
\end{equation}%
where for each broken rim hook $b=\lambda^{(i)}/\lambda^{(i-1)}$ in $\cT$ we set
\begin{equation}
\wt(b)=( t-1)^{\#b-1}\prod_{h\in b}(-1)^{r(h)-1}t^{c(h)-1}  \label{Hecke_wt}
\end{equation}%
with $r(h)$ denoting the number of rows a rim hook $h\in b$ occupies, by $c(h)$ the number of columns and 
by $\#b$ the number of (connected) rim hooks $h$ contained in $b$. 

\begin{example}
Set $\lambda=(6,6,5,3,2,2,2,2,1)$ and $\mu=(6,4,3,3,2,1,1)$. The skew diagram $\lambda/\mu$ is a broken rim hook which has two connected components $h_1,h_2$, and thus $\#b=2$, of length 4 and 5; see Figure \ref{fig:BFcorr2h2o}. We read off from the diagram in Figure \ref{fig:BFcorr2h2o} that $r(h_1)=2,c(h_1)=4$ and $r(h_2)=4,c(h_2)=2$.
\end{example}

%

\begin{theorem}[Ram]
(i) Any trace function $\chi_t\in\cR^m(t) $ is completely determined
by the values $\chi_t(\alpha)\defeq\chi_t(T_{w_\alpha})$, where $T_{w_\alpha}\in\cH_m$ is the element fixed by \eqref{minC}. %
(ii) One has the following $t$-extension of the Frobenius formula,
\begin{equation}\label{tFrobenius}
\frac{h_\mu[(t-1)Y]}{(t-1)^{\ell(\mu)}}=\sum_{\lambda\vdash m}\chi^\lambda_t(\mu)s_\lambda[Y]\,,
\end{equation}
where $s_\lambda\in\Lambda$ denotes the Schur function and $\chi^\lambda_t$ is the character obtained from the irreducible $\cH_m$-module fixed by the partition $\lambda\vdash m$. %
(iii) One has the following combinatorial sum formula for the irreducible characters,%
\begin{equation}\label{chi2ribbon}
\chi _{t}^{\lambda }(\alpha)=\sum_{|\cT|=\lambda }\wt(\cT),
\end{equation}%
where the sum runs over all broken rim hook tableaux $\cT$ of shape $\lambda $
and content $\alpha$.
\end{theorem}
The `inverse' of the characteristic map is described by the following dual version of the Frobenius formula \cite[Prop 5.3]{wan2015}:
\begin{proposition}[Wan-Wang]
\begin{equation}\label{dualtFrobenius}
s_\lambda[Y]=\sum_{\mu\vdash m}\chi^\lambda_t(T_{w_\mu})(t-1)^{\ell(\mu)}m_{\mu}\left[\frac{Y}{t-1}\right]\,.
\end{equation}
\end{proposition}
In particular, it follows from \eqref{tFrobenius} and \eqref{dualtFrobenius} that $\ch_t(\chi_t^\lambda)=s_\lambda$ which completely fixes the quantum characteristic map \eqref{chart}.

\subsection{Skew Schur functions and the restriction functor}
We now recall the Hopf algebra structure on $\Lambda(t)$. Using the basis of Schur functions $\{s_\lambda:\lambda\in\cP^+\}$ define the following product, coproduct, unit, co-unit and antipode on $\Lambda(t)$,
\begin{gather}
s_\lambda s_\mu=\sum_{\nu}c_{\lambda\mu}^{\nu}s_{\nu},\qquad \Delta(s_\lambda)=\sum_{\mu,\nu}c^{\lambda}_{\mu\nu}s_\mu\otimes s_\nu\label{Hopf1}\\
f(t)\mapsto f(t) s_{\varnothing},\quad s_\lambda\mapsto \delta_{\lambda,\varnothing},\quad s_\lambda\mapsto (-1)^{|\lambda|}s_{\lambda'},\label{Hopf2}
\end{gather}
where  $f(t)\in\C(t)$, $\lambda'$ denotes the conjugate partition of $\lambda$, and $c^{\lambda}_{\mu\nu}\in\Z_{\ge 0}$ are the {\em Littlewood-Richardson coefficients} which are fixed via
\begin{equation}\label{Ind2LR}
\op{Ind}^{\cH_{m+n}}_{\cH_{m}\otimes\cH_{n}}\chi_t^{\mu}\otimes\chi_t^{\nu}=\sum_{\lambda}c^{\lambda}_{\mu\nu}\chi_t^\lambda
\end{equation}
or, alternatively,
\begin{equation}\label{Rcop}
\op{Res}^{\cH_m}_{\cH_{m'}\otimes\cH_{m''}}\chi_t^{\lambda}=\sum_{\mu\vdash m'}\sum_{\nu\vdash m''}c^{\lambda}_{\mu\nu}\chi_t^\mu\otimes\chi_t^\nu
\end{equation}
By construction, one then obtains:
\begin{proposition}
(i) The maps \eqref{Hopf1}, \eqref{Hopf2} turn $\Lambda(t)$ into a graded Hopf algebra. (ii) 
The quantum characteristic map $\ch_t:\cR(t)\to\Lambda(t)$ defined in \eqref{chart} extends to a graded Hopf algebra isomorphism.  
\end{proposition}
\begin{proof}
The first part is due to  Zelevinsky \cite{zelevinsky1981} and involves a straightforward checking of the Hopf algebra axioms. The second follows from \eqref{dualtFrobenius} and that the monomial symmetric functions $\{m_\lambda~|~\lambda\in\cP^+\}$ form a basis \cite[Ch.I]{macdonald1998}.
\end{proof}

Our main interest in this article will be the following generalisation of the Frobenius formula \eqref{tFrobenius} and its dual version \eqref{dualtFrobenius} to skew diagrams: define the (virtual) skew character 
\begin{equation}\label{skewchi}
\chi_t^{\lambda/\mu}=\sum_{\nu}c^{\lambda}_{\mu\nu}\chi^\nu_t\,,
\end{equation}
then we have the following:
\begin{corollary}
\begin{equation}\label{skewtFrobenius}
\frac{h_\mu[(t-1)Y]}{(t-1)^{\ell(\mu)}}\,s_\nu[Y]=\sum_{\lambda}\chi^{\lambda/\nu}_t(\mu)s_\lambda[Y]
\end{equation}
and
\begin{equation}\label{skewdualtFrobenius}
s_{\lambda/\mu}[Y]=\sum_{\nu}\chi^{\lambda/\mu}_t(\nu)(t-1)^{\ell(\nu)}m_{\nu}\left[\frac{Y}{t-1}\right]\,,
\end{equation}
where the sums run respectively over all partitions $\lambda$ and $\nu$ such that $|\lambda|=|\mu|+|\nu|$. 
\end{corollary}
\begin{proof} Set $\overline{h}_\mu[Y]=h_\mu[(t-1)Y]/(t-1)^{\ell(\mu)}$ and $\underline{m}_\mu=(t-1)^{\ell(\mu)}m_\mu[Y/(t-1)]$. Both sets of symmetric functions form dual bases with respect to the Hall inner product, $\langle\overline{h}_\lambda,\underline{m}_\mu\rangle=\delta_{\lambda\mu}$ \cite{wan2015}. Thus,
\begin{eqnarray*}
\langle\overline{h}_\alpha,s_{\lambda/\mu}\rangle &=& \sum_{\nu}c^\lambda_{\mu\nu}\langle\overline{h}_\alpha,s_{\nu}\rangle
=\chi^{\lambda/\mu}_t(\alpha)\\
&=& \langle\overline{h}_\alpha,s_\mu^*s_{\lambda}\rangle=\langle s_{\mu}\overline{h}_\alpha,s_{\lambda}\rangle\;.
\end{eqnarray*}
Here we have used in the first line the known expansion of skew Schur functions into Schur functions and in the second line that $s_\mu^*s_{\lambda}=s_{\lambda/\mu}$, where $s_\mu^*$ is the adjoint of the operator which multiplies with the Schur function $s_\mu$. The latter identity follows from $\langle \Delta(f),g\otimes h\rangle=\langle f, gh\rangle$ which holds for all $f,g,h\in\Lambda(t)$; see e.g. \cite{macdonald1998}.
\end{proof}
The skew character \eqref{skewchi}  can be obtained combinatorially by summing over all broken rim hook tableaux of skew shape $\lambda/\nu$ and naturally appears in connection with the Hopf-algebra structure on $\cR(t)$ via  \eqref{Rcop}, $\op{Res}^{\cH_m}_{\cH_{m'}\otimes\cH_{m''}}\chi_t^{\lambda}=\sum_{\mu\subset\lambda}\chi_t^\mu\otimes\chi_t^{\lambda/\mu}$. %
Note that $\chi_t^{\lambda/\mu}$ is identically zero unless $\mu\subset\lambda$, i.e. the Young diagram of $\lambda$ contains the Young diagram of $\mu$. It follows at once that for fixed integers $0\le k\le n$ the corresponding  subset %
$$\mc{R}'_{k,n}(t)=\{\chi_t^\lambda~:~\lambda_1\leq n-k,\lambda'_1\le k\}$$ 
of irreducible characters  spans a finite-dimensional sub-coalgebra of $\cR(t)$. Recall that the cohomology ring $H^*(\op{Gr}_k(\C^n))$ of the Grassmannian has a natural Frobenius algebra structure (see e.g. \cite{abrams2000quantum}) and, thus, a coproduct.
\begin{corollary}
The $\Z$ sub-coalgebra $\subset \cR(t)$ spanned by the elements in $\mc{R}'_{k,n}(t)$ is isomorphic to  $H^*(\op{Gr}_k(\C^n))
$ when viewed as a $\Z$-coalgebra. 
\end{corollary}
The main result \eqref{chi2GW} of this article is a generalisation of the last statement to the quantum cohomology of Grassmannians.

\section{A Hecke version of the boson-fermion correspondence}
In this section we extend the boson-fermion correspondence \eqref{BFiso} to Hecke algebras. That is, we perform the plethystic variable substitution  $Y\to X=(t-1)Y$ in the Hamiltonian $H=H[Y]$ and then show that the matrix element on the right hand side in \eqref{BFiso} can be rewritten in terms of certain fermionic operators $A\in\End(\cF_c\otimes\C(t))$ whose matrix elements give Hecke instead of symmetric group characters. We identify these operators via the boson-fermion correspondence with so-called `half-vertex operators' acting on the bosonic Fock space $\Lambda(t)$ and relate them to known (bosonic) vertex operators in the literature \cite{jing1991}. In a subsequent section we then show that these fermionic operators are the transfer matrices of the asymmetric six-vertex model on the infinite lattice and, thus, that the matrix element in \eqref{BFiso} can be identified with its partition function under the plethystic variable transformation $Y\to X=(t-1)Y$. 

\subsection{A $t$-deformation of the Clifford algebra}
Set $\cF(t)=\cF\otimes\C(t)$ and define the following $t$-deformation of the fermion fields \eqref{Cl2F} in $\End\cF(t)$,
\begin{equation}\label{Cl2Ft+}
\psi_i^+(t)\sigma=
\left\{\begin{array}{ll}
(\sigma+\epsilon_i)\prod_{j>i}(-t)^{\sigma_j}, & \sigma_i=0\\
0, & \sigma_i=1
\end{array}
\right.
\end{equation}
and
\begin{equation}\label{Cl2Ft-}
\psi_i^-(t)\sigma=
\left\{\begin{array}{ll}
(\sigma-\epsilon_i)\prod_{j>i}(-t)^{-\sigma_j}, & \sigma_i=1\\
0, & \sigma_i=0
\end{array}
\right.
\end{equation}
Via a straightforward computation one obtains the following:
\begin{lemma}
We have the commutation relations
\[
i<j:\qquad\psi^\pm_i(t)\psi^\pm_j(t)=-t \psi^\pm_j(t)\psi^\pm_i(t)\quad\text{ and }\quad
\psi^\mp_i(t)\psi^\pm_j(t)=-t^{-1} \psi^\pm_j(t)\psi^\mp_i(t)
\]
as well as
\[
\psi_i^-(t)\psi^+_i(t)+\psi^+_i(t)\psi^-_i(t)=1\;,
\]
where $i,j\in\Z$.
\end{lemma}
Note that $\psi_i^\pm(1)=\psi^\pm_i$, so the above operators are a deformation of the Clifford algebra representation \eqref{Cl2F}. Set $E_{ij}(t)=(1-t)\psi_i^-(t)\psi^+_j(t)$ and define the following `half-vertex' operator 
\begin{equation}\label{HeckeA}
A(x;t)=1+\sum_{r>0}\;\sum_{i_1<j_1<\cdots<i_r<j_r}(xt)^{j_1-i_1+\cdots +j_r-i_r}E_{i_1j_1}(t)\cdots E_{i_rj_r}(t)\;,
\end{equation}
where $x$ is a formal variable, the `spectral parameter' of the six-vertex model discussed in one of the subsequent sections. The operator $A(x;t)$ should be understood as a formal power series in the variable $x$,
\[
A(x;t)=\sum_{r\ge 0}x^rA_r(t)
\]
with coefficients $A_r(t)\in\End\cF(t)$. Our motivation for introducing the above operator is Proposition \ref{prop:tBFcorr} from the introduction.
\begin{proposition}\label{prop:BFHecke}
Consider the following $t$-extension 
\begin{equation}\label{BFHecke}
\imath_t=\imath\otimes 1:\cF\otimes\C(t)\to\bigoplus_{c\in\Z}z^c\otimes(\Lambda\otimes\C(t))
\end{equation}
of the boson-fermion correspondence \eqref{BFiso}. Then for any $f\in\C(t)$, $v\in\cF$,
\begin{equation}\label{tBFiso}
v\otimes f\mapsto\imath_t(v\otimes f)=\sum_{c\in\Z}z^c\otimes\langle v\otimes f,A(x_1;t)A(x_2;t)\cdots\sigma(\0,c)\rangle,
\end{equation}
where $\langle v\otimes f, v'\otimes g\rangle=\langle v,v'\rangle f\bar g$ is the extension of the bilinear form \eqref{fermiHall} and the variables $X=x_1+x_2+\cdots$ are related to the variables $Y=y_1+y_2+\cdots$ in $\Lambda$ via the plethystic substitution $X=Y/(t-1)$. 
\end{proposition}

We prove the proposition via a couple of lemmata.
\begin{lemma}\label{lem:E}
The operator $E_{ij}(t)=(1-t)\psi_i^-(t)\psi^+_j(t)$ with $i<j$ acts on Maya diagrams via
\[
E_{ij}(t)\sigma(\mu,c)=\left\{
\begin{array}{ll}
(t^{-1}-1)(-t)^{r(\lambda/\mu)-1}\sigma(\lambda,c),& \lambda/\mu\text{ rim hook of length }j-i\\
0,& else
\end{array}\right.
\]
where the (connected) rim hook $h=\lambda/\mu$ starts at position $i$ and ends at position $j$ under the bijection \eqref{Maya}.
\end{lemma}

\begin{proof}
Recall that $p_rs_\mu$ is a sum of Schur polynomials $s_\lambda$ for which $\lambda/\mu$ is a rim hook of length $r$; see e.g. \cite[Ch.I]{macdonald1998}. Since $\lim_{t\to 1}E_{ij}(t)/(1-t)=\psi^-_i\psi^+_j$ and $P_r=\sum_{j-i=r}\psi^-_i\psi^+_j$, we infer immediately from the boson-fermion correspondence that $E_{ij}(t)\sigma(\mu,c)$ must be a Maya diagram $\sigma(\lambda,c)$ with $\lambda/\mu$ being a rim hook of length $r=j-i$ provided that $\sigma_i(\mu,c)=1$ and $\sigma_j(\mu,c)=0$. For each 1-letter between positions $i$ and $j$, i.e. $\sigma_l=1$ with $i<l<j$,  there is a power of $-t^{-1}$. Under the bijection \eqref{Maya} each such 1-letter corresponds to a row in $\lambda$ intersecting the rim hook $\lambda/\mu$. If either $\sigma_i(\mu,c)=0$ or $\sigma_j(\mu,c)=1$ than such a rim hook cannot be added and $E_{ij}(t)\sigma(\mu,c)=0$ by the definitions \eqref{Cl2Ft+}, \eqref{Cl2Ft-}.
\end{proof}

The second result we require to prove Proposition \ref{prop:BFHecke} is the following description of the action of the coefficients of the half-vertex operator \eqref{HeckeA}.
\begin{lemma}\label{lem:A}
Under the $t$-extension \eqref{BFHecke} of the boson-fermion correspondence the $A$-operator \eqref{HeckeA} is given by multiplication with the function $h_r[(t-1)Y]$. That is, we have the identity
\begin{equation}\label{A2h}
\imath_t\circ A_r(t)=h_r[(t-1)Y]\circ\imath_t
\end{equation}
for all $r\in\Z_{\ge 0}$. In particular, we have the following expression in terms of the Heisenberg algebra: 
\begin{equation}\label{HeckeA2boson}
\imath_t\circ A(x;t)=\exp\left(\sum_{r>0}\frac{t^r-1}{r}p_r[Y] x^r\right)\circ\imath_t\;.
\end{equation}
\end{lemma}
Note that it follows from \eqref{A2h} that the operators $A_r(t)\in\End\cF(t)$ commute with each other.
\begin{proof}
The case $r=0$ is trivial. Let $r>0$ then we infer from the definition \eqref{HeckeA} that
\[
A_r(t)=t^r\sum_{s=1}^r\sum_{\substack{i_1<j_1<\ldots<i_s<j_s\\
j_1-i_1+\cdots+ j_s-i_s=r}}E_{i_1j_1}(t)\cdots E_{i_sj_s}(t)\;.
\]
Thus, by the previous lemma we must have that $A_r\sigma(\mu,c)$ is a linear combination of Maya diagrams $\sigma(\lambda,c)$ with $b=\lambda/\mu$ being a broken rim hook of length $r$ and with $0<\#b\le s$. According to Lemma \ref{lem:E} each rim hook $h\in b$ gives rise to a factor $$(t^{-1}-1)(-t)^{-r(h)-1}=t^{-|h|}(t-1)t^{c(h)-1}(-1)^{r(h)-1}\;.$$
But since $\sum_{h\in b}|h|=r$ we arrive at
\[
A_r(t)\sigma(\mu,c)=\sum_{\lambda}\chi_t^{\lambda/\mu}(r)(t-1)\sigma(\lambda,c),
\]
where the sum runs over all partitions $\lambda$ such that $\lambda/\mu$ is a broken rim hook of length $r$. The assertion now follows from \eqref{skewtFrobenius} and $\imath_t(\sigma(\lambda,c))=s_\lambda$.
\end{proof}

\begin{proof}[Proof of Proposition \ref{prop:BFHecke}]
Repeated application of Lemma \ref{lem:A} together with formulae \eqref{skewtFrobenius} and \eqref{skewdualtFrobenius} now yields the following: for any partition $\mu$ and charge $c\in\Z$ one has
\begin{equation}\label{HeckeSchur}
A(x_1;t)A(x_2;t)\cdots\sigma(\mu,c)=\sum_{\lambda}s_{\lambda/\mu}[(t-1)X]\sigma(\lambda,c)\;,
\end{equation}
where the last sum runs over all partitions $\lambda$ containing $\mu$. To see this, first observe that the right hand side of  \eqref{HeckeA2boson} can be re-written as 
\[
\imath_t\circ A(x;t)=\left(\prod_{i>0}\frac{1-xy_i}{1-txy_i}\right)\circ\imath_t\;.
\]
Employing the generalised Cauchy identity \cite{macdonald1998}
\[
\prod_{i,j>0}\frac{1-x_iy_j}{1-tx_iy_j}=\sum_{\nu}s_\nu[(t-1)X]s_{\nu}[Y]\;
\]
together with the Littlewood-Richardson expansion of products of Schur functions,
\[
\sum_{\nu}s_\lambda[(t-1)X]s_{\nu}[Y]s_{\mu}[Y]=\sum_{\lambda,\nu}c^{\lambda}_{\nu\mu}s_\nu[(t-1)X]s_{\lambda}[Y]
=\sum_{\lambda}s_{\lambda/\mu}[(t-1)X]s_{\lambda}[Y]\;,
\]
we arrive at \eqref{HeckeSchur} using that $\imath_t(\sigma(\lambda,c))=z^c\otimes s_\lambda[Y]$ for any $c\in\Z$. Upon setting $\mu=\0$ and $X=Y/(t-1)$ we obtain the identity \eqref{tBFiso} for $v=\sigma(\lambda,c)$. Since the Maya diagrams span $\cF=\bigoplus_{c\in\Z}\cF_c$ the assertion then follows.
\end{proof}

We call Proposition \ref{prop:BFHecke} a Hecke version of the boson-fermion correspondence \eqref{BFiso} because the half-vertex operator \eqref{HeckeA} generates a basis which is related to the Schur-basis via the character table of Hecke algebras.
\begin{corollary}
Given any partition $\mu$ and $c\in\Z$ set 
\[
\eta_{\mu,c}=(t-1)^{-\ell(\mu)}A_{\mu_1}A_{\mu_2}\cdots A_{\mu_\ell}\sigma(\varnothing,c)\;.
\] 
Then  $\{\eta_{\mu,c}~|~\mu\in\cP^+,c\in\Z\}$ forms a basis and the transformation matrix to the basis $\{\sigma(\lambda,c)~|~\lambda\in\cP^+,c\in\Z\}$ is given by the character table of Hecke algebras, i.e. $
\eta_{\mu,c}=\sum_{\lambda}\chi^{\lambda}_t(\mu)\sigma(\lambda,c)
$. More generally, we have
\begin{equation}\label{A2chi}
\langle\sigma(\lambda,c),A_{\nu_1}(t)\cdots A_{\nu_\ell}(t)\sigma(\mu,c')\rangle=\delta_{cc'}(t-1)^{\ell}\chi^{\lambda/\mu}_t(\nu)\;.
\end{equation}
\end{corollary}
\begin{proof}
This is a direct consequence of the identity \eqref{tFrobenius}, $\imath_t(\sigma(\lambda,c))=z^c\otimes s_\lambda[Y]$ and \eqref{A2h}.
\end{proof}
In order to make contact with vertex operators previously defined in the literature, we also compute the inverse of the operator \eqref{HeckeA}.
\begin{corollary}
The operator \eqref{HeckeA} is invertible. That is, we have the identities
\begin{equation}\label{Ainv}
A(x t;t^{-1})\circ A(x;t)=A(x;t)\circ A(x t;t^{-1})=\op{Id}_{\cF(t)}
\end{equation}
and, thus, $A^{-1}(x;t)=A(xt;t^{-1})$. Explicitly,
\begin{equation}\label{HLA}
A^{-1}(x;t)=1+\sum_{r>0}\;\sum_{i_1<j_1<\cdots<i_r<j_r}x^{j_1-i_1+\cdots +j_r-i_r}E_{i_1j_1}(t^{-1})\cdots E_{i_rj_r}(t^{-1})\;.
\end{equation}
Furthermore,
\begin{equation}\label{HLSchur}
A^{-1}(x_1;t)A^{-1}(x_2;t)\cdots\sigma(\mu,c)=\sum_{\lambda}s_{\lambda/\mu}[(1-t)X]\sigma(\lambda,c)
\end{equation}
and
\begin{equation}\label{Ainv2chi}
\langle\sigma(\lambda,c),A^{-1}_{\nu_1}(t)\cdots A^{-1}_{\nu_\ell}(t)\sigma(\mu,c')\rangle=\delta_{cc'}
(-1)^{|\lambda|-|\mu|}(t-1)^{\ell}\chi^{\lambda'/\mu'}_t(\nu)\;,
\end{equation}
where $A^{-1}(x;t)=\sum_{r\ge 0}x^rA^{-1}_r(t)$ and $\lambda'$, $\mu'$ denote the conjugate partitions of $\lambda$ and $\mu$, respectively.
\end{corollary}

\begin{proof}
The first assertion \eqref{Ainv} can be easily deduced from \eqref{HeckeA2boson}. The rest of the proof then follows along similar steps as the proof of Proposition \ref{prop:BFHecke} using Lemma \ref{lem:A}. We therefore omit the details.
\end{proof}
 
\begin{remark}\label{rmk:poly}\rm 
Note that while \eqref{Cl2Ft+}, \eqref{Cl2Ft-} are defined in $\End\cF(t)$ the matrix elements of the operators \eqref{HeckeA} and \eqref{HLA} are both polynomial in $t$. This can be seen from \eqref{HeckeSchur} and \eqref{HLSchur}, respectively. The skew Schur functions in both expressions acquire only a polynomial dependence on $t$ under the plethystic substitutions $X\to(t-1)X$ and $X\to (1-t)X$. In particular, we can view the operator coefficients $A_r(t)$ and $A^{-1}_r(t)$ as elements in $\End\cF[t]$ with $\cF[t]=\cF\otimes\C[t]$. Only when making the plethystic substitution $X=Y/(t-1)$ in \eqref{tBFiso} do we need to work over the rational functions in $t$.
\end{remark}

\subsection{Vertex operators and Hall-Littlewood functions}
We now relate the operator \eqref{HeckeA} and its inverse \eqref{HLA} to Jing's vertex operators for Hall-Littlewood functions \cite{jing1991}: define on the bosonic Fock space $\Lambda(t)$ the bilinear form
\begin{equation}\label{tHall}
\langle f,g\rangle_t\defeq\left\langle f[Y/(1-t)],g[Y]\right\rangle
\end{equation}
and introduce the vertex operators
\begin{equation}
\Phi^-(x;t)=\sum_{r\in\Z}x^{-r}\Phi^-_r(t)
\defeq\left(e^{\;\sum\limits_{r>0}\frac{1-t^r}{r}\,p_r[Y] x^r} \right)\circ\left(e^{\;\sum\limits_{r>0}\frac{t^r-1}{r}\,p^*_r[Y] x^{-r}}\right)
\end{equation}
and
\begin{equation}
\Phi^+(x;t)=\sum_{r\in\Z}x^r\Phi^+_r(t)
\defeq\left(e^{\;\sum\limits_{r>0}\frac{t^r-1}{r}\,p_r[Y] x^r} \right)\circ\left(e^{\;\sum\limits_{r>0}\frac{1-t^r}{r}\,p^*_r[Y] x^{-r}}\right)\;,
\end{equation}
where $\Phi^+$ is the adjoint of $\Phi^-$ with respect to the inner product \eqref{tHall}. The latter are known to obey the commutation relations
\begin{equation*}
\Phi^-(x_1;t)\Phi^-(x_2;t)(x_1-tx_2)=\Phi^-(x_2;t)\Phi^-(x_1;t)(tx_1-x_2)\;.
\end{equation*}
Equivalently, the $\Phi_i^-(t)$ obey the relations
\begin{equation}
\Phi^-_{i-1}\Phi^-_{j}-t\Phi^-_{i}\Phi^-_{j-1}=t\Phi^-_{j}\Phi^-_{i-1}-\Phi^-_{j-1}\Phi^-_{i}\;.
\end{equation}
In order to make the connection with the the operators \eqref{HeckeA} and \eqref{HLA} we need first to introduce the analogue of the basis $\{s_\lambda[(1-t)Y]\}$ appearing in \eqref{tHall} in the fermionic Fock space $\cF(t)$. 
\begin{lemma}
Define $\sigma_t(\lambda,c)=\det(A^{-1}(t)_{\lambda_i-i+j})_{1\le i,j\le \ell(\lambda)}\,\sigma(\0,c)$ for any $c\in\Z$. Then 
\[
\imath_t(\sigma_t(\lambda,c))=z^c\otimes s_{\lambda}[(1-t)Y]
\]
and, thus, introducing the bilinear form $\cF(t)\otimes\cF(t)\to\cF(t)$ fixed via $$\langle\sigma_t(\lambda,c),\sigma(\mu,c')\rangle_t=\delta_{cc'}\delta_{\lambda\mu}$$ turns the isomorphism $\imath_t:\cF(t)\to\bigoplus_{c\in\Z}z^c\otimes\Lambda(t)$ into an isometry with respect to the inner product \eqref{tHall}.
\end{lemma}
\begin{proof}
It follows from \eqref{A2h} and \eqref{Ainv} that $A^{-1}_r(t)$ corresponds to multiplying with $h_r[(1-t)Y]$ in $\Lambda(t)$. The assertions are now  immediate from the known Jacobi-Trudi identity \cite[Ch.I]{macdonald1998} $s_\lambda=\det(h_{\lambda_i-j+1})_{1\leq 1,j,\leq \ell(\lambda)}$.
\end{proof}
The following corollary, Proposition \ref{prop:A2VO} in the introduction, now states the precise relationship between \eqref{HeckeA} and the above vertex operators:
\begin{corollary}
Set $x^*=x^{-1}$. Then we have the following identities under the boson-fermion correspondence:
\[
\Phi^-(x;t)\circ\iota_t=\iota_t\circ A^{-1}(x;t)\circ A^*(x;t)
\quad\text{and}\quad
\Phi^+(x;t)\circ\iota_t=\iota_t\circ A(x;t)\circ(A^{-1}(x;t))^*\,.
\]
Moreover, introduce the operators $\Psi_m=\sum_{r-s=m}A_r^{-1}A^*_s$. Then,
\[
\Psi_{r-1}\Psi_s-t\Psi_{r}\Psi_{s-1}=t\Psi_s\Psi_{r-1}-\Psi_{s-1}\Psi_r
\] 
and given any partition $\lambda$ of length $\ell$, the image of the vectors
\[
\Psi_{\lambda_1}\Psi_{\lambda_2}\cdots\Psi_{\lambda_\ell}\sigma(\varnothing,c)
\]
under the boson-fermion correspondence are the Hall-Littlewood functions $Q_\lambda(y;t)$.
\end{corollary}

\begin{proof}
We have done all the computational work in previous lemmata: the identity \eqref{HeckeA2boson} and \eqref{Ainv} allows us to identify \eqref{HeckeA} and \eqref{HLA} as the half-vertex operators $\exp\left(\pm\sum_{r>0}\frac{t^r-1}{r}p_r[Y]x^r\right)$ in the defintion of $\Phi^\pm(x;t)$. Using the inner product from the previous lemma then uniquely defines their adjoint operators in $\End\cF(t)$. The assertion regarding Hall-Littlewood functions is a direct consequence of the known result \cite{jing1991} for the bosonic vertex operators $\Phi^\pm$ and exploiting the boson-fermion correspondence \eqref{tBFiso}.
\end{proof}

\section{The asymmetric six-vertex model}
\begin{figure}
\centering
\includegraphics[width=.9\textwidth]{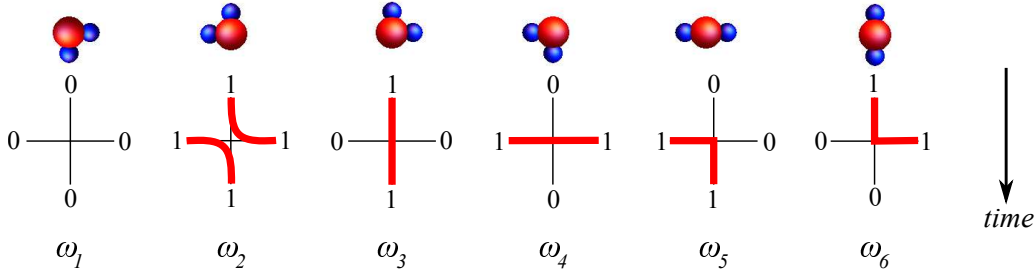} 
\caption{Shown are the six allowed vertex configurations of the lattice model. On top are the corresponding orientations of the water molecule $H_2O$ and below a depiction of the vertex configurations in terms of non-intersecting path segments.}
\label{fig:6vweights}
\end{figure}

In this section we develop a `graphical calculus' for computing the matrix elements of the half-vertex operators \eqref{HeckeA} and \eqref{HLA}, the Hecke characters, using a known exactly solvable lattice model from statistical mechanics \cite{baxter2016}. While the model is known, its combinatorial description and connection to Hecke characters presented here are new. 

\subsection{Definition of the six-vertex model}
Consider a square lattice $\Gamma\subset\Z\times\Z$. We call each point $\mathrm{v}=(i,j)\in\Gamma$ a {\em vertex} and two neighbouring vertices $\mathrm{v}=(i,j)$, $\mathrm{v'}=(i',j')$ an edge if either $i'=i$ and $j'=j\pm1$ or $i'=i\pm 1$ and $j'=j$. Denote by $\mathbb{E}$ the set of lattice edges and given a vertex $\mathrm{v}=(i,j)$, we label the four edges $\{((i,j-1),(i,j)),((i-1,j),(i,j)),((i,j),(i+1,j)),((i,j),(i+1,j))\}$ by W(est), N(orth), E(ast), S(outh). A {\em lattice configuration} $\mathcal{C}$ is a map $\mathbb{E}\rightarrow \{0,1\}$. We assign each
configuration $\mathcal{C}$ a `Boltzmann weight'\footnote{For a proper probability the values must lie in the interval $[0,1]$, but here we loosely borrow the term from statistical mechanics.} by setting%
\begin{equation*}
\prob(\mathcal{C})=\prod_{\mathrm{v}\in \Gamma}\prob(\mathrm{v}),
\end{equation*}%
where the values $\prob(\mathrm{v})=\prob(\mathrm{W,N,E,S})=\omega_i$, $i=1,\ldots,6$ for six individual vertex configurations are displayed in Figure \ref{fig:6vweights}. All other remaining vertex configurations are forbidden, i.e. they
have probability or weight zero. If $\omega_1=\omega_2$, $\omega_3=\omega_4$ and $\omega_5=\omega_6$ the the model is called {\em symmetric}, otherwise {\em asymmetric}. While the set of allowed vertex configurations will remain
the same across the square lattice, we will vary below the weights $\{\omega_1,\ldots,\omega_6\}$ depending on
the lattice row number $i$, that is consider different weights $\omega _{r}=\omega _{r}(i)$. In addition, 
we will impose certain conditions on some lattice edges at the boundary of $\Gamma$, i.e. fix their values.

In order to compute the actual probability $\op{Prob}(\mc{C})$ for a certain lattice configuration $\mc{C}$ to occur we must divide by the {\em partition function}, i.e. $\op{Prob}(\mc{C})=\prob(\mc{C})/Z$, where $Z$ is the weighted sum over all lattice configurations,%
\begin{equation}\label{Z}
Z(\omega _{1},\ldots ,\omega _{6})=\sum_{\mathcal{C}}\prob(\mathcal{C}%
)=\sum_{\mathcal{C}}\prod_{\mathrm{v}\in \Gamma}\prob(\mathrm{v}%
)\;.
\end{equation}%
Here we have tacitly assumed that all of the above expressions are finite. We will show this for the infinite lattice below by imposing suitable boundary conditions. One of our results will be to identify the matrix element in \eqref{tBFiso} for $v=\sigma(\lambda,c)$ as the partition function \eqref{Z} for the infinite lattice and a special choice of Boltzmann weights.

\subsection{Transfer matrices and Yang-Baxter algebras}

We start by looking at the partition functions for single lattice rows, which yield so-called {\em transfer matrices}. 
\begin{definition}\label{def:transfer}
Given $\alpha,\beta\in\{0,1\}$ fixed and any two binary strings $\sigma,\sigma'\in\{0,1\}^{\times n}$, denote by $\langle\sigma|T_{\beta\alpha}|\sigma'\rangle$ the weighted sums over all possible configurations of a single lattice row with $\alpha,\beta,\sigma,\sigma'$ fixing the values of the outer right horizontal edge, left outer horizontal edge and the vertical edges on bottom and top, respectively. Then we call the matrix $T_{\beta\alpha}=(\langle\sigma|T_{\beta\alpha}|\sigma'\rangle)_{\sigma,\sigma'}$ a {\em row-to-row transfer matrix}.
\end{definition}
The partition function \eqref{Z} for the full lattice can be expressed in terms of products of these transfer matrices; see equation \eqref{Z2mom} below. For this reason, we will consider the eigenvalue problem of the transfer matrices in a later section to compute its matrix products more easily. The choice of vertex configurations in Figure \ref{fig:6vweights} allows one to express the transfer matrices in terms of solutions to the Yang-Baxter equation which greatly simplifies solving their eigenvalue problem using familiar techniques from exactly solvable systems \cite{baxter2016}.

Introduce the complex vector space $V=\mathbb{C}v_{0}\oplus \mathbb{C}v_{1}\cong \mathbb{C}^{2}$
and identify the values $\sigma =0,1$ assigned to each
lattice edge under a configuration $\mathcal{C}:\mathbb{E}\to\{0,1\}$ with the basis vectors $%
v_{0},v_{1}$ in $V$. Let $V^{\ast }$ be the dual space and denote by $%
v^{\sigma }$ the dual basis vectors, i.e. $\langle v^{\sigma },v_{\sigma
^{\prime }}\rangle =\delta _{\sigma \sigma ^{\prime }}$. Identify the
Boltzmann weights $\omega _{r}$ as matrix elements of an operator $%
R:\C[\omega]\otimes V\otimes V\rightarrow \C[\omega]\otimes V\otimes V$,%
\begin{equation}
R(\omega )v_{\sigma }\otimes v_{\sigma ^{\prime }}=\sum_{\rho ,\rho ^{\prime
}=0,1}\prob(\sigma ,\sigma ^{\prime },\rho ,\rho ^{\prime })~v_{\rho
}\otimes v_{\rho ^{\prime }}\;.
\end{equation}%
Here $\rho,\rho',\sigma,\sigma'\in\{0,1\}$ are respectively the value of the W, N, E, S edge of a vertex. Explicitly,
\begin{gather}
\prob(0,0,0,0)=\omega_1,\quad\prob(1,1,1,1)=\omega_2,\quad\prob(0,1,0,1)=\omega_3\\
\prob(1,0,1,0)=\omega_4,\quad\prob(1,0,0,1)=\omega_5,\quad\prob(0,1,1,0)=\omega_6
\end{gather}
and $\prob(\rho,\rho',\sigma,\sigma')=0$ for any other 4-tuple of edge values; see Figure \ref{fig:6vweights}. If the map $R=R(\omega)$ satisfies the Yang-Baxter equation then the lattice model is called 
\emph{exactly solvable}, meaning that we can make exact statements regarding
the nature and properties of the partition function \cite{baxter2016}.

Define $e_{01}=\left( 
\begin{smallmatrix}
0 & 1 \\ 
0 & 0%
\end{smallmatrix}%
\right) ,\;e_{10}=\left( 
\begin{smallmatrix}
0 & 0 \\ 
1 & 0%
\end{smallmatrix}%
\right) ,\;e_{00}=\left( 
\begin{smallmatrix}
1 & 0 \\ 
0 & 0%
\end{smallmatrix}%
\right) ,\;e_{11}=\left( 
\begin{smallmatrix}
0 & 0 \\ 
0 & 1%
\end{smallmatrix}%
\right) $ to be the $2\times 2$ unit matrices acting on $V\cong \mathbb{C}%
^{2}$ via $e_{01}v_{1}=v_{0}$, $e_{10}v_{0}=v_{1}$ and $e_{\sigma \sigma
}v_{\sigma ^{\prime }}=\delta _{\sigma \sigma ^{\prime }}v_{\sigma }$. Then
the Boltzmann weights $\omega =(\omega _{1},\ldots ,\omega _{6})$ of the
asymmetric six-vertex model define the $R$-matrix%
\begin{equation}
R(\omega )=\left( 
\begin{array}{cc}
\omega _{1}e_{00}+\omega _{3}e_{11} & \omega _{5}e_{10} \\ 
\omega _{6}e_{01} & \omega _{4}e_{00}+\omega _{2}e_{11}%
\end{array}%
\right) \;.
\end{equation}%
For the moment we are treating the weights $\omega $ as indeterminates
assuming that $\omega _{1},\omega _{3}$ are invertible. Consider the following two ratios,%
\begin{equation}
\Delta(\omega) =\frac{\omega _{1}\omega _{2}+\omega _{3}\omega _{4}-\omega
_{5}\omega _{6}}{2\omega _{1}\omega _{3}}\qquad \text{and}\qquad \Gamma(\omega) =%
\frac{\omega _{2}\omega _{4}}{\omega _{1}\omega _{3}}\;.  \label{quadric}
\end{equation}%
Let $\omega$, $\omega'$, $\omega''$ be three choices of vertex weights 
and set $\mb{V}=\C[\omega]\otimes V$, $\mb{V}'=\C[\omega']\otimes V$ and $\mb{V}''=\C[\omega'']\otimes V$. Denote by $\mb{k}$ the ring $\C[\omega,\omega',\omega'']$ localised at $\omega_1,\omega_3,\ldots$ and modulo the relations 
\begin{equation}\label{quadric=}
\Delta(\omega)=\Delta(\omega')=\Delta(\omega'')\quad\text{ and }\quad\Gamma(\omega)=\Gamma(\omega')=\Gamma(\omega'')\;.
\end{equation}
Note that we allow for $\omega=\omega'=\omega''$ or $\omega=\omega'$ etc. We identify elements in $\mb{V}\otimes\mb{V}'\otimes\mb{V}''$ with elements in $W=\mb{k}\otimes V\otimes V\otimes V$ in the natural way, $f(\omega)u\otimes g(\omega')v\otimes h(\omega'')w\mapsto f(\omega)g(\omega')h(\omega'') u\otimes v\otimes w$.

\begin{proposition}[Baxter]\label{prop:Baxter}
The three $R$-matrices associated with the weights $\omega$, $\omega'$, $\omega''$ obey the Yang-Baxter equation%
\begin{equation}
R_{12}(\omega )R_{13}(\omega ^{\prime })R_{23}(\omega ^{\prime \prime
})=R_{23}(\omega ^{\prime \prime })R_{13}(\omega ^{\prime })R_{12}(\omega
)\; \label{ybe}
\end{equation}
in $\End(W)$, where $R_{ij}$ acts non-trivially only in the $i$th and $j$th factor of the tensor product $V^{\otimes 3}$.
\end{proposition}

\begin{proof}
A straightforward but lengthy computation; see \cite{baxter2016} for details.  
\end{proof}

Consider the tensor product $\mathcal{V}=\mb{k}
\otimes V^{\otimes n}$, the so-called \emph{quantum space}, and
define the \emph{row monodromy matrix} $V\otimes \mathcal{V}\rightarrow
V\otimes \mathcal{V}$ as%
\begin{equation}
T(\omega )=R_{0n}(\omega )\cdots 
R_{01}(\omega )=
e_{00}\otimes A+e_{01}\otimes B+e_{10}\otimes C+e_{11}\otimes D=
\left( 
\begin{array}{cc}
A & B \\ 
C & D%
\end{array}%
\right) \ ,  \label{mom}
\end{equation}%
where we identify the matrix elements $A=A(\omega),B=B(\omega),C=C(\omega),D=D(\omega)$ as maps $\mathcal{V}%
\rightarrow \mathcal{V}$. The operators $A,B,C,D$ generate a representation of the \emph{%
Yang-Baxter algebra} in $\End\mc{V}$ whose commutation relations are given by the
Yang-Baxter equation or $RTT$-relation.

\begin{corollary}
The monodromy matrices satisfy the Yang-Baxter equation%
\begin{equation}
R_{12}(\omega )T_{1}(\omega ^{\prime })T_{2}(\omega ^{\prime \prime
})=T_{2}(\omega ^{\prime \prime })T_{1}(\omega ^{\prime })R_{12}(\omega )
\label{RTT}
\end{equation}%
The latter equation implies the following commutation relations between
the monodromy matrix elements,%
\begin{eqnarray}
\omega _{1}A^{\prime }A^{\prime \prime } &=&\omega _{1}A^{\prime \prime
}A^{\prime },\quad\omega _{2}D^{\prime }D^{\prime \prime }=\omega _{2}D^{\prime
\prime }D^{\prime },\;  \notag \\
\omega _{1}B^{\prime }B^{\prime \prime } &=&\omega _{2}B^{\prime \prime
}B^{\prime },\quad\omega _{2}C^{\prime }C^{\prime \prime }=\omega _{1}C^{\prime
\prime }C^{\prime },\;
\end{eqnarray}%
and%
\begin{eqnarray}
\omega _{1}B^{\prime }A^{\prime \prime } &=&\omega _{4}A^{\prime \prime
}B^{\prime }+\omega _{5}B^{\prime \prime }A^{\prime }  \notag \\
\omega _{2}C^{\prime }D^{\prime \prime } &=&\omega _{3}D^{\prime \prime
}C^{\prime }+\omega _{6}C^{\prime \prime }D^{\prime }
\end{eqnarray}%
\begin{eqnarray}
\omega _{3}A^{\prime }C^{\prime \prime }+\omega _{5}C^{\prime }A^{\prime
\prime } &=&\omega _{1}C^{\prime \prime }A^{\prime }  \notag \\
\omega _{4}D^{\prime }B^{\prime \prime }+\omega _{6}B^{\prime }D^{\prime
\prime } &=&\omega _{2}B^{\prime \prime }D^{\prime }
\end{eqnarray}%
as well as%
\begin{eqnarray}
\omega _{4}(D^{\prime }A^{\prime \prime }-A^{\prime \prime }D^{\prime })
&=&\omega _{5}B^{\prime \prime }C^{\prime }-\omega _{6}B^{\prime }C^{\prime
\prime }  \notag \\
\omega _{3}B^{\prime }C^{\prime \prime }-\omega _{4}C^{\prime \prime
}B^{\prime } &=&\omega _{5}(D^{\prime \prime }A^{\prime }-D^{\prime
}A^{\prime \prime })\;.
\end{eqnarray}
Here $A'=A(\omega')$, $A''=A(\omega'')$ etc.
\end{corollary}

\begin{figure}\label{fig:mom}
\centering
\includegraphics[width=.80\textwidth]{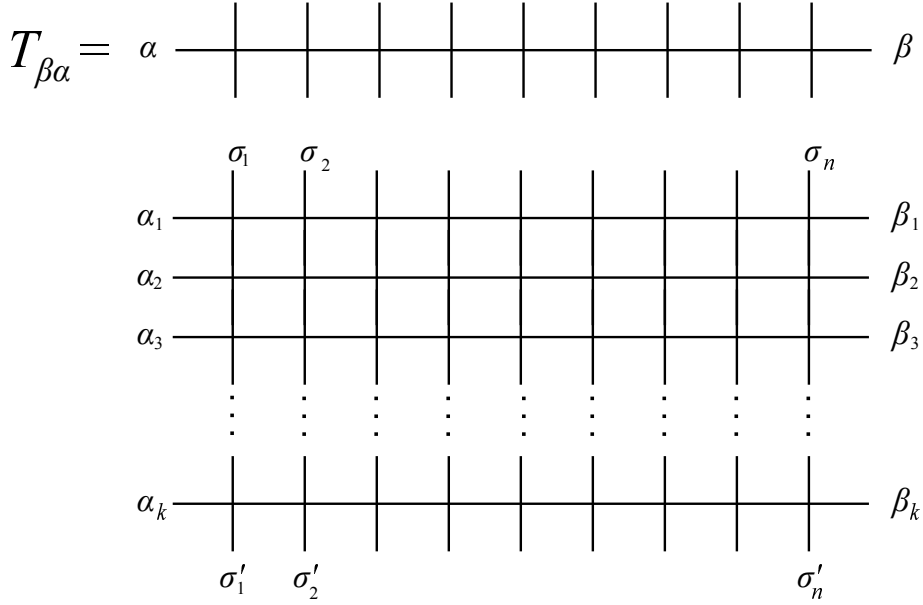} 
\caption{Diagrammtic depiction of the monodromy matrix $T$ and the partition function $Z$ for fixed boundary conditions, where the binary strings $\alpha,\beta$ and $\sigma,\sigma'$ fix the values of the outer edges.}
\end{figure}

Let the square lattice $\Gamma\subset\Z\times\Z$ have $n$ columns and $k$ rows, i.e. $\Gamma=\{(i,j):1\le i\le k,\;1\le j\le n\}$. Fix the values of the top vertical lattice edges to be $\sigma_1,\ldots,\sigma_n$, the ones at the bottom to be $\sigma'_1,\ldots,\sigma'_n$ and the values of the horizontal boundary edges to be $\alpha_1,\ldots,\alpha_k$ on the left and $\beta_1,\ldots,\beta_k$ on the right; see Figure \ref{fig:mom}.
\begin{lemma}
The partition function \eqref{Z} of the asymmetric six-vertex model with the boundary conditions as described above is given by the following matrix element,
\begin{equation}\label{Z2mom}
Z(\omega_1,\ldots,\omega_6)=\langle v^{\sigma'_1}\otimes\cdots\otimes v^{\sigma'_n},T_{\beta_k\alpha_k}(\omega)\cdots
T_{\beta_1\alpha_1}(\omega)v_{\sigma_1}\otimes\cdots\otimes v_{\sigma_n}\rangle\;,
\end{equation}
where $T_{\beta\alpha}(\omega)$ are the matrix elements of the monodromy matrix \eqref{mom}. That is, $T_{00}=A$, $T_{01}=B$ etc.
\end{lemma}
\begin{proof}
A straightforward computation which is standard in integrable lattice models and follows from the definition of the monodromy matrix.
\end{proof}
In the following section we give a combinatorial realisation of the Yang-Baxter algebra in terms of broken rim hooks.

\subsection{A combinatorial description of the Yang-Baxter algebra}
Given a basis vector $v_{\sigma _{1}}\otimes \cdots \otimes v_{\sigma _{n}}\in V^{\otimes n}$ with $V=\C^2$, the corresponding binary string $\sigma =(\sigma_1,\ldots,\sigma_n)$ of length $n$ can be associated with a partition $\lambda(\sigma)$  that is an $n$-core via the map , $\sigma\mapsto\lambda(\sigma)=(i_k-k,i_{k-1}-k-1,\ldots,i_1-1)$, where $i_1<i_2<\cdots<i_k$ are the positions of 1-letters in $\sigma$. This is a finite size analogue of the bijection \eqref{Maya} discussed earlier with $c=k$. Consider the following decomposition $V^{\otimes n}\cong \bigoplus_{k=0}^{n}V_{k }$, where $V_{k}$ is spanned by 
$$\{v_{\lambda(\sigma)}=v_{\sigma _{1}}\otimes \cdots \otimes v_{\sigma _{n}}~ :~\lambda\in\cP^+_{k,n}\}\,.$$
Denote by $\{v^\lambda~:~\lambda\in\cP^+_{k,n}\}$ the dual basis. %
In order to describe the action of the Yang-Baxter algebra, we first note that by construction the matrix elements 
$$\langle\lambda|X|\mu\rangle\defeq\langle v^\lambda,X(\omega)v_\mu\rangle\in\C[\omega_1,\ldots,\omega_6]$$ with $X=A,B,C,D$ are polynomial in the Boltzmann weights. Then each generator $X=A,B,C,D$ of the Yang-Baxter algebra can be decomposed into a sum  $X=\sum_{r\geq 0}X_{r}$ such that the non-vanishing matrix
elements $\langle\lambda|X_{r}|\mu\rangle\in\C[\omega_1,\ldots,\omega_6]$ are all of degree $r$ in the variables $\omega_2,\omega_4,\omega_6$. We ignore the degrees in the variables $\omega_1,\omega_3,\omega_5$. In what follows we describe the action of each of the $X_r$ rather than the action of the corresponding generator $X$.

Define the following six-vertex Boltzmann weight (probability) of a broken rim hook $b=\lambda/\mu$,%
\begin{equation}
\prob(b;\omega )\overset{\text{def}}{=}\omega _{1}^{\bar{c}(b)}\omega _{3}^{\bar{r}%
(b)}\left( \omega _{5}\omega _{6}\right) ^{\#b}\prod_{h\in b}\omega
_{2}^{r(h)-1}\omega _{4}^{c(h)-1}  \label{rimhookweight}
\end{equation}%
and set $\overline{\prob}(b;\omega _{1},\ldots ,\omega _{6})=%
\prob(b;\omega _{4},\omega _{3},\omega _{2},\omega _{1},\omega
_{6},\omega _{5})$. Here $\bar r(b)$ and $\bar c(b)$ denote the number of rows and columns in $\lambda$ which do not intersect $b$. The definition of $r(b),c(b),\#b$ is the same as in \eqref{Hecke_wt}. The following proposition links broken rim hooks to the asymmetric six-vertex model.

\begin{figure}\label{fig:ADop}
\centering
\includegraphics[width=.85\textwidth]{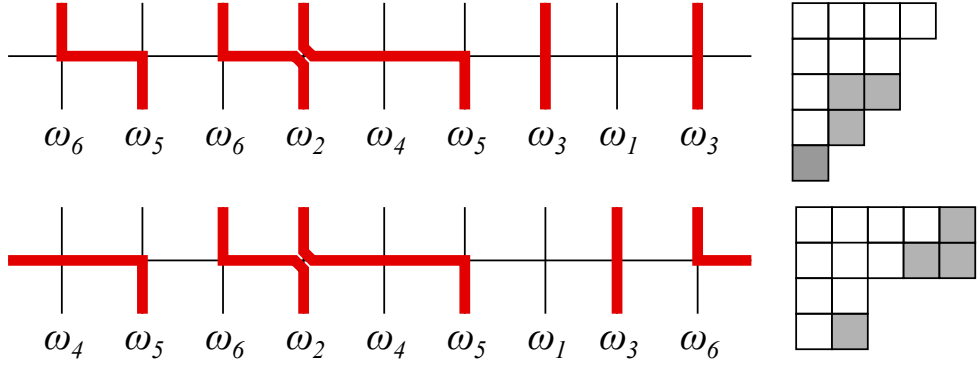} 
\caption{Examples of a row
configuration for the $A$ (top) and $D$-operator (bottom) of the six-vertex
model. Displayed on the right are the broken rim hooks associated with each
row configuration. }
\end{figure}

\begin{proposition}\label{prop:AD6v}
Let $A_{r},D_{r}$ be the degree $r=1,\ldots,n-1$ components of the diagonal matrix
elements of the monodromy matrix (\ref{mom}) with respect to the Boltzmann weights $\omega_2,\omega_4,\omega_6$. Then for any 
$n$-core $\mu\in\cP^+_{\ell,n}$ of length $\ell(\mu)=\ell$ we have that %
\begin{equation}
A_{r}v_\mu
=\sum_{\substack{ \lambda \in\cP_{\ell,n}^{+}  \\ \lambda
/\mu =b,\;|b|=r}}\prob(b)v_\lambda
\text{\quad\ and\quad\ }%
D_{r}v_\mu
 =\sum_{\substack{ \lambda \in\cP_{\ell,n }^{+}  \\ \mu
/\lambda =b,\;|b|=n-r}}\overline{\prob}(b)v_\lambda
\;,
\label{AD6v}
\end{equation}%
where the sum runs over all $\lambda \in\cP_{\ell,n }^{+}$ such that
respectively $\lambda /\mu $ and $\mu /\lambda $ are broken rim hooks $b$ of
length $r$ and $n-r$. For $A_0$ and $D_n$ we include the case of the empty broken rim hook, i.e. $\lambda=\mu$.
\end{proposition}

\begin{proof}
Recall that the $A$-operator is the sum over six-vertex lattice configurations where the left and right outer horizontal lattice edge always have value 0. One then easily deduces from Figure \ref{fig:6vweights} the following rules:
\begin{itemize}
\item[Rule 1.] The vertex at the left boundary of the lattice row must either have weight $\omega_1,\omega_3$ or $\omega_6$. The vertex at the right boundary must either have weight $\omega_1,\omega_3$ or $\omega_5$.
\item[Rule 2.] Each vertex with weight $\omega_5$, say in lattice column $j$, must be preceded by a vertex with weight $\omega_6$ in some column $i$ with $i<j$ such that the vertices in columns $i<k<j$ must have either weight $\omega_2$ or $\omega_4$.
\item[Rule 3.] Each vertex between a vertex of weight $\omega_5$ in column $i$ that precedes a vertex of weight $\omega_6$ in column $j$ (and not being of either type) must have weight $\omega_1$ or $\omega_3$. The same applies to vertices that lie between the left boundary and a vertex of weight $\omega_6$ or vertices between a vertex of weight $\omega_5$ and the right boundary. If there are no vertices with weight $\omega_5$ or $\omega_6$ then each vertex has either weight $\omega_1$ or $\omega_3$.
\end{itemize}
Consider the matrix element $\langle \lambda |A|\mu \rangle $ for $\lambda
,\mu \in\cP_{\ell,n }^{+}$, then we need to show that $\langle \lambda |A|\mu
\rangle =\prob(b)$ if $\lambda /\mu $ is a broken rim hook and
vanishes otherwise. 

First we note that it follows from Rule 2 that the vertices with weight $\omega_5$ and $\omega_6$ always occur in pairs. This implies that the corresponding binary strings $\sigma(\lambda) $, $\sigma(\mu )$ must have the same number of 1-letters. Hence, $\langle
\lambda |A|\mu \rangle =0$ if $\lambda \in\cP_{\ell',n}^{+}$ and $\mu
\in\cP_{\ell,n }^{+}$ with $\ell \neq \ell ^{\prime }$. 
Furthermore, all paths segments in Figure \ref{fig:6vweights} propagate either downward or to
the right. Therefore, each of the 1-letter positions in $\sigma(\lambda )$ must be greater or
equal than the ones in $\sigma(\mu )$ (when numbering 1-letters from left to right). Employing the correspondence between
01-words and partitions we can conclude that $\mu \subset \lambda$. 

We claim that each set of $r$ consecutive vertices starting with a vertex of weight $\omega_6$ in column $i$ and ending with one of weight $\omega_5$ in column $j$ corresponds to adding a ribbon or rim hook $h$ of length $r=j-i$ to $\mu$. According to Rule 2, there is a 1-letter in $\sigma(\mu)$ and a 0-letter in $\sigma(\lambda)$ at position $i$, while there is 0-letter in $\sigma(\mu)$ and a 1-letter in $\sigma(\lambda)$ at position $j$. Strictly in between positions $i$ and $j$ the binary strings $\sigma(\mu)$ and $\sigma(\lambda)$ must be identical. According to the bijection \eqref{Maya} with $c=\ell$ this corresponds to adding a box with diagonal or content $k$ for each $i\le k<j$. One then easily verifies that the second and fourth vertex configurations in Figure \ref{fig:6vweights}
correspond to having two consecutive squares in a column and row of $\lambda/\mu $, respectively. The fifth and sixth vertex configuration signal the start and the end of the rim hook $h$.

Similarly, it follows from Rule 3 that the substrings of $\sigma(\mu)$ and $\sigma(\lambda)$ in between the end of one rim hook and the start of another must be identical. Thus, the first (third) vertex in lattice column (row) $j$ corresponds to having no square $s=(x,y)$ in $\lambda /\mu $ with diagonal $y-x=j-n-1$. If there are no rim hooks added then we must have trivially that $\sigma(\mu)=\sigma(\lambda)$ and $\lambda=\mu$.

Since the map \eqref{Maya} is a bijection, we deduce that for each given broken rim hook $b$ each of the above statements can be reversed. For each rim hook $h\in b$ starting at position $i$ and ending at position $j$ we must have that the 01-substrings of $\sigma(\mu),\sigma(\lambda)$ in between these positions are identical. Hence, Rule 2 follows. Because we only consider partitions with $\lambda_1\le n-\ell$ and $\ell(\lambda)\le \ell$ one obtains Rule 1: either there is a rim hook starting (ending) in the first (last) column, or there is not. Similarly, Rule 3 is obtained by considering columns in $\lambda$ which do not intersect $\lambda/\mu$, with the last case (no vertices of type $\omega_5$ or $\omega_6$) corresponding to $\lambda=\mu$, the empty skew diagram.

For the $D$-operator the argument is similar, but now the bijection between
a row configuration and a skew diagram $\mu /\lambda $ which is a broken rim
hook is different. Namely, given a row configuration where the values of the
top vertical edges are fixed by $\sigma(\mu )$ and the ones on the bottom by $%
\sigma(\lambda )$, remove a square from the Young diagram $Y$ of $\mu $ for each
unoccupied horizontal lattice edge (value $0$) between two 1-letters
starting at the bottom. For instance, suppose there are $k$ unoccupied
horizontal lattice edges between the leftmost lattice site and the first
1-letter position in $\sigma(\mu )$, then we would remove $k$ squares from the
bottom row of $Y$. If there are $k^{\prime }$ unoccupied edges between the
first and second 1-letter position in $\sigma(\mu )$, then we remove $k^{\prime }$
squares from the second row from the the bottom, etc. We omit the further steps in proving the bijection as they are analogous to the ones used in the case of the $A$-operator.
\end{proof}

\begin{example}\rm
Consider the example shown in Figure \ref{fig:ADop}. For the $A$-operator we
infer from the skew tableau on the right that there are two rim hooks, so $%
\#b=2$, and $\bar{r}(b)=2$, $\bar{c}(b)=1$. Moreover, since one rim hook
consists just of a single square we have $r(h_{1})=c(h_{1})=1$, while for
the other rim hook we find $r(h_{2})=c(h_{2})=2$. Hence the associated
weight is $\prob(b)=\omega _{1}\omega _{2}\omega _{3}^{2}\omega
_{4}(\omega _{5}\omega _{6})^{2}$ which matches the product over the vertex
weights displayed in Figure \ref{fig:ADop}. For the $D$-operator
configuration shown in Figure \ref{fig:ADop} we also have $\#b=2$, but now 
$\bar{r}(b)=1$, $\bar{c}(b)=2$. We obtain similarly as before that $%
r(h_{1})=c(h_{1})=1$ and $r(h_{2})=c(h_{2})=2$, whence $\overline{\prob}(b)
=\omega _{1}\omega _{2}\omega _{3}\omega _{4}^{2}(\omega _{5}\omega
_{6})^{2}$. 
\end{example}

We now consider the off-diagonal elements of the monodromy matrix (\ref{mom}%
) of the asymmetric six-vertex model. As before we decompose $B=\sum_{r\geq 0}B_{r}$ and $C=\sum_{r\geq 0}C_{r}$ where the
non-vanishing matrix elements of $B_{r},C_{r}$ are all of degree $r$ in the variables $\omega_2,\omega_4,\omega_6$.
\begin{proposition}
 Let $\mu\in\cP_{\ell,n}^+$. Then%
\begin{equation}\label{BCaction}
B_{r}v_\mu
=\sum_{\substack{ \lambda \in\cP_{\ell +1,n}^{+}  \\ \lambda
^{+}/\mu =b,\;|b|=r+1}}\frac{\prob(b)}{\omega _{6}}~v_\lambda
\text{\quad\ and\quad\ }C_{r}v_{\mu}
=\sum_{\substack{ \lambda \in
\cP_{\ell-1,n }^{+}  \\ \mu ^{+}/\lambda =b,\;|b|=n+1-r}}\frac{\overline{\prob%
}(b)}{\omega _{5}}v_\lambda\;,
\end{equation}%
where the first sum runs over all partitions $\lambda \in\cP_{\ell +1,n}^{+}$
such that $\lambda ^{+}/\mu $ is a broken rim hook $b$ of length $r$ with $%
\lambda ^{+}=(\lambda _{1}+1,\ldots ,\lambda _{\ell +1}+1)$ and the second
over all $\lambda \in\cP_{\ell -1,n}^{+}$ such that $\mu ^{+}/\lambda $ is a
broken rim hook $b$ of length $n+1-r$.
\end{proposition}

Note that the coefficients in \eqref{BCaction} are still polynomial in the $\omega_i$, $i=1,\ldots,6$ because each $B$-configuration contains at least one vertex with weight $\omega_6$ and each $C$-configuration at least one vertex with weight $\omega_5$. Thus, $\prob(b)$ and $\overline{\prob}(b)$ each contain at least one factor $\omega_6$ and $\omega_5$, respectively.
\begin{proof}
Note that for the $B$-operator we can use the previous bijection proved for the $A$-operator by prepending a vertex configuration of weight $\omega_6$; see Figure \ref{fig:BCop}. This results in a rim hook starting in the first column of $\lambda^+$ and increasing the number of 1-letters by one. Because of adding this extra vertex we then need to divide the resulting weight $\prob(b)$ of the broken rim hook $b=\lambda^+/\mu$ by $\omega_6$.

The argument for the $C$-operator is analogous: prepend a vertex of weight $\omega_5$ to the lattice row configuration to obtain a row configuration of the $D$-operator and use the previous bijection from Proposition \ref{prop:AD6v} to deduce the result.
\end{proof}

\begin{example}\rm
See the top configuration in Figure \ref{fig:BCop}. For the $B$-operator we
use the same bijection as in the case of the $A$-operator to arrive at the
broken rim hook $b$ displayed on the right. One has $\bar{c}(b)=1$, $\bar{r}%
(b)=2$, $\#b=2$ and $r(h_{1,2})=2,2,c(h_{1,2})=2,3$. Hence, $\prob%
(b)/\omega _{6}=\omega _{1}\omega _{2}^{2}\omega _{3}^{2}\omega
_{4}^{3}\omega _{5}(\omega _{5}\omega _{6})$.

The bottom configuration in Figure \ref{fig:BCop} displayes a possible row
configuration for the $C$-operator. We now use the previous bijection from
the $D$-operator to verify the last proposition also in this case.

\begin{figure}\label{fig:BCop}
\centering
\includegraphics[width=.9\textwidth]{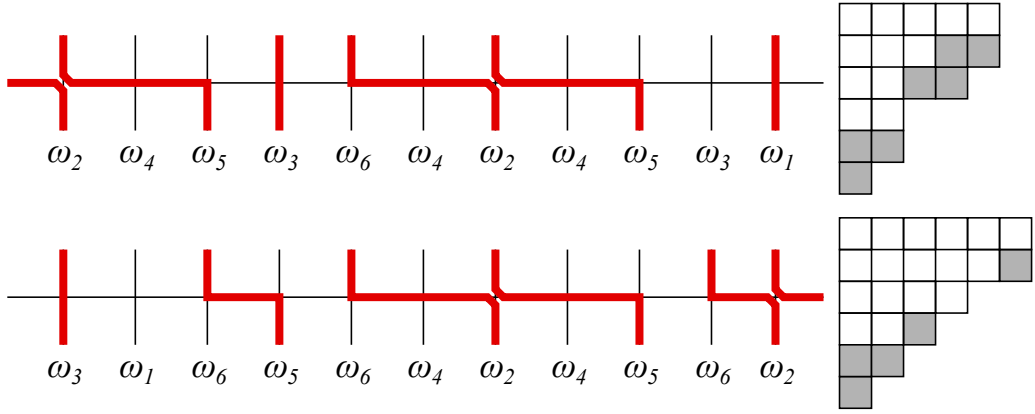} 
\caption{A configuration for the $B$ (top) and $%
C$-operator (bottom) and the associated broken rim hooks (right).} 
\end{figure}
\end{example}

\subsection{The infinite lattice and vertex operators}
We now make contact with our previous discussion of Hecke characters and vertex operators. Namely, we will show that the half-vertex operator \eqref{HeckeA} is the transfer matrix of the asymmetric six-vertex model on the infinite lattice $\Gamma=\Z\times\Z$ under the following choice of Boltzmann weights in Figure \ref{fig:6vweights}: 
\begin{equation}\label{HeckeBweights}
\omega_1=\omega_3=\omega_5=1\quad\text{ and }\quad
\omega_2=-x,\quad \omega_4=tx,\quad \omega_6=(t-1)x\;.
\end{equation}
N.B. these weights belong to the values $\Delta=0$ and $\Gamma=-tx^2$ in \eqref{quadric}. In order for the transfer matrix and partition function \eqref{Z} to be finite, we need to impose special boundary conditions on the lattice. 
Consider a single (infinite) lattice row, which we identify with $\Z$, and for two arbitrary but fixed Maya diagrams $\sigma,\sigma'$ restrict the model to those lattice row configurations $\mc{C}=\mc{C}(\sigma,\sigma')$, where the value of the upper vertical edge (N) of each vertex ${\rm v}_i$ is fixed by $\sigma_i$ and the value of the lower vertical edge (S) by $\sigma'_i$ with $i\in\Z$. In line with the definition of Maya diagrams and Rule 1 in the proof of Prop \ref{prop:AD6v}, we impose the boundary conditions that for $i\ll -1$ only vertex configurations with weight $\omega_3$ and for $i\gg 1$ only vertices of weight $\omega_1$ occur. Due to these boundary conditions, and because of the choice \eqref{HeckeBweights}, we deduce that $\prob({\rm v}_i)=1$ except for finitely many $i\in\Z$.

\begin{proposition}
The matrix elements of the `half-vertex operator' \eqref{HeckeA} are given by the six-vertex row partition function with Boltzmann weights \eqref{HeckeBweights},
\begin{equation}\label{BFA}
\langle\sigma',A(x;t)\sigma\rangle
=\sum_{\mc{C}(\sigma,\sigma')}\prod_{i\in\Z}\prob({\rm v}_i)\;.
\end{equation}
\end{proposition}
\begin{proof}
This is a direct consequence of the bijection between six-vertex lattice configurations and broken rim hooks for finite lattices in Prop \ref{prop:AD6v}, which extends in a straightforward manner to the infinite lattice due to the chosen boundary conditions for the configurations $\mc{C}(\sigma,\sigma')$ with $\sigma,\sigma'$ being Maya diagrams. Under the choice \eqref{HeckeBweights} we then have that any broken rim hook $b=\lambda(\sigma')/\mu(\sigma)$ and the corresponding six-vertex lattice configuration $\mc{C}(\cT)$ have equal weights, $\prob(b)=(t-1)\wt(b)$. 
\end{proof} 
 It is evident from \eqref{Ainv} that we can express the inverse vertex operator $A^{-1}(x;t)$ through the following `local change' of six-vertex Boltzmann weights:
\begin{corollary}
Making the choice
\begin{equation}\label{HLBweights}
\omega_1=\omega_3=\omega_5=1\quad\text{ and }\quad
\omega_2=-tx,\quad \omega_4=x,\quad \omega_6=(1-t)x
\end{equation}
for the six-vertex Boltzmann weights in Figure \ref{fig:6vweights} we obtain
\begin{equation}\label{BFAinv}
\langle\sigma',A^{-1}(x;t)\sigma\rangle
=\sum_{\mc{C}(\sigma,\sigma')}\prod_{i\in\Z}\prob({\rm v}_i)\;.
\end{equation}
\end{corollary}
\begin{proof}
The proof is analogous to showing \eqref{BFA} by using \eqref{Ainv}. We therefore omit it.
\end{proof}
N.B. it follows from \eqref{HeckeBweights} and \eqref{BFA} that the expansion coefficients $A_r(t)$ of the operator \eqref{HeckeA} are polynomial in $t$, i.e. $A_r(t)\in\End\cF_c[t]$ with $\cF_c[t]=\cF_c\otimes\C[t]$ and, thus, $A(x;t)$ restricts to an operator $\cF_c[t]\to\Lambda[t]$ for any charge $c\in\Z$. The same argument applies to the inverse $A^{-1}(x;t)$ using \eqref{HLBweights}. While the polynomial dependence on $t$ is not immediately obvious from the definitions \eqref{HeckeA} and \eqref{HLA}, it immediately follows from \eqref{HeckeSchur} and \eqref{HLSchur}; see our earlier Remark \ref{rmk:poly}.

 \subsection{Quasi-periodic boundary conditions and cylindric rim hooks}
 
 Let $q$ be an indeterminate, called the {\em twist parameter}, and define the operator 
\begin{equation}
\tau=\sum_{r=0}^n\tau_r=A+qD,
\end{equation} 
which is called the {\em row-to-row transfer matrix with quasi-periodic boundary conditions} and $\tau_r$ are the components of degree $r$ in the Boltzmann weights $\omega_2,\omega_4,\omega_6$. Comparison with Figure \ref{fig:6vweights} shows that $r$ is the number of horizontal edges having value 1 in the corresponding lattice row configuration $\mc{C}$. In order to see that each matrix element of $\tau_r$ corresponds to a single row partition function with quasi-periodic boundary conditions, note that if the left and right outer horizontal lattice edges in a row both have value 0, then we obtain a matrix element of the $A_r$-operator and if they have instead value 1, then we obtain a matrix element of the $D_r$-operator. The indeterminate $q$ is introduced to keep track of the `winding number' $d$ around the cylinder. In particular, according to \eqref{Z2mom} we have:
\begin{lemma} 
The partition function of a finite cylinder of circumference $n$ and height $\ell$ is given by
\[
Z=\sum_{\alpha_1,\ldots,\alpha_\ell=0,1}q^{\sum_i\alpha_i}\langle\lambda|T_{\alpha_\ell\alpha_\ell}\cdots T_{\alpha_1\alpha_1}|\mu\rangle=
\langle\lambda|\tau^\ell|\mu\rangle\;.
\]
\end{lemma}
\begin{proof}
This is a special case of the identity \eqref{Z2mom}, where we we fix the boundary conditions $\alpha=\beta$ and then sum over all binary strings $\alpha$.
\end{proof}
We now wish to extend our previous result \eqref{AD6v} to the case of periodic boundary conditions with the aim of introducing cylindric analogues of Hecke characters.

\begin{figure}\label{fig:cylbrimhook}
\centering
\includegraphics[width=.8\textwidth]{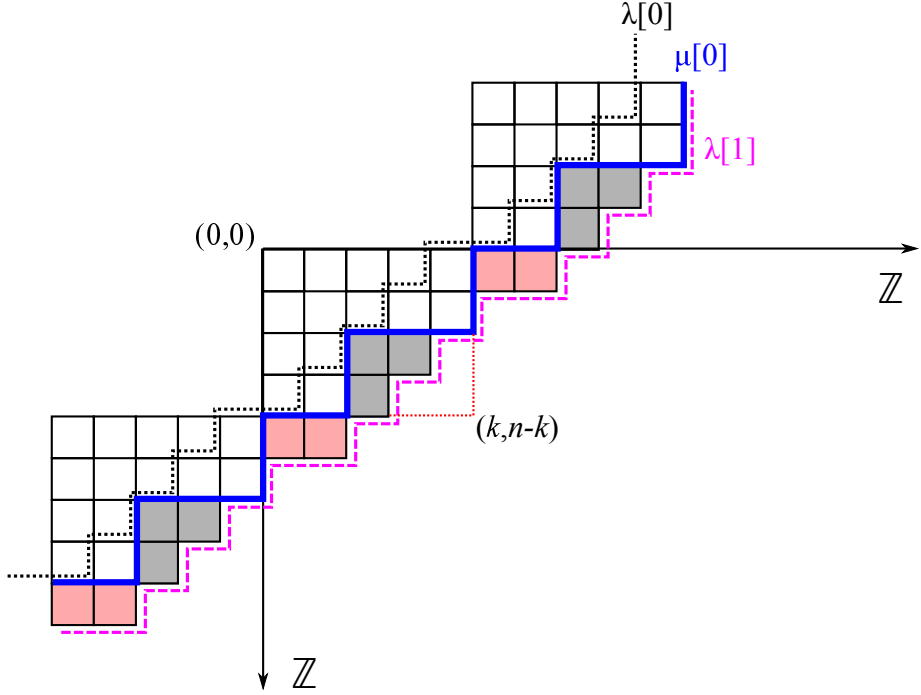} 
\caption{Example of a cylindric broken rim hook for $k=4$ and $n=9$. Shown are the cylindric loops $\mu[0]$ with $\mu=(5,5,2,2)$ (solid line) and $\lambda[0]$ (dotted line), $\lambda[1]$ (dashed line) with $\lambda=(4,3,2,1)$. The cylindric skew shape $\lambda/1/\mu$ is a broken rim hook with two connected components and corresponds to the $D$-operator configuration shown in Figure \ref{fig:ADop}.}
\end{figure}

First, we recall the notion of cylindric loops. The latter were introduced by
Gessel and Krattenthaler \cite{gessel1997} in the context of cylindric plane
partitions. We adopt here the notation used in \cite{postnikov2005}.

Fix two integers $n\ge 2$ and $0\le k\le n$. Given a partition $\lambda\in\cP^+_{k,n}$, define for every $r\in \mathbb{
Z}$ a \emph{cylindric loop} $\lambda[r]$ as the following infinite integer sequence, 
\begin{equation}
\lambda \lbrack r]=(\ldots ,\underset{r}{\lambda _{k }+r+n-k},\underset{r+1%
}{\lambda _{1}+r},\ldots ,\underset{r+k}{\lambda _{k }+r},\underset{%
r+k +1}{\lambda _{1}+r-n+k},\ldots )\;.
\end{equation}%
We interpret the latter as a map $\Z\to\Z$ subject to the condition $\lambda_{i+k}[r]=\lambda_i[r]-n+k$. The cylindric loop can therefore be visualised as a path on the cylinder $\Z\times\Z/(-k,n-k)\Z$; see Figure \ref{fig:cylbrimhook}.


A \emph{cylindric skew diagram} is the set of squares in $\Z\times\Z$
between two cylindric loops: suppose $\lambda ,\mu\in \cP^+_{k,n}$ then
we shall denote by $\lambda /d/\mu $ the set of points between the two lines 
$\lambda [d]$ and $\mu[0]$ modulo integer shifts by the vector $(-k,n-k)$, 
\begin{equation}
\lambda /d/\mu :=\{\langle i,j\rangle \in \mathbb{Z}\times \mathbb{Z}/(-k
,n-k)\mathbb{Z}~|~\mu[0]_{i}<j\le\lambda[d]_{i}\}\;.
\end{equation}%
Note that for $d=0$ we recover the familiar skew-diagram of two partitions, i.e. $\lambda /0/\mu=\lambda /\mu $. 

We extend the definition of the Boltzmann weight \eqref{rimhookweight} to cylindric broken rim hooks $b=\lambda/d/\mu$ as follows: let $h$ be a cylindric rim hook, then set $r(b)$ to be the number of rows with $1\le i\le k$ of $\lambda[d]$ which intersect $h$ and let $c(h)$ be the number of columns $j$ in which there exists a square $s=(i,j)\in h$ with $1\le i\le k$. Similarly, we extend the definition of $\bar r(b)$ to be the number of rows $1\le i\le k$ of $\lambda[d]$ which do not intersect $b$ and $\bar c(b)$ to be the number of columns $j$ in which there is no square $s=(i,j)\in b$ with $1\le i\le k$. 
Set $\#b$ to be the number of {\em distinct} cylindric rim hooks in $b$, where we call two cylindric rim hooks $h,h'$ distinct if $h$ cannot be obtained from $h'$ by a shift in the direction $(-k,n-k)$. With these conventions in place define the Boltzmann weight of a cylindric broken rim hook $b=\lambda/d/\mu$ as 
\[
\prob(\lambda/d/\mu)=\omega_1^{\bar c(b)}\omega_3^{\bar r(b)}(\omega_5\omega_6)^{\#b}\prod_{h}\omega_2^{r(h)-1}\omega_4^{c(h)-1}\;,
\]
where the product now runs over all distinct rim hooks $h\in\lambda/d/\mu$.
\begin{example}\rm 
Set $n=9$ and $k=4$. Let $\lambda=(4,3,2,1)$ and $\mu=(5,5,2,2)$. The cylindric skew shape $\lambda/1/\mu$ is shown in Figure \ref{fig:cylbrimhook}. There are two distinct rim hooks $h,h'$ of length $2$ and $3$, respectively. We find that $r(h)=1$, $c(h)=2$, $r(h')=c(h')=2$ and $\bar r(b)=\bar c(b)=1$. Thus, we obtain $\prob(\lambda/d/\mu)=\omega_1\omega_3(\omega_5\omega_6)^2\omega_2\omega_4^2$. This cylindric broken rim hook corresponds to the $D$-operator configuration shown in the lower half of Figure \ref{fig:ADop} and one verifies that the Boltzmann weights of the lattice configuration coincides with the weight of $\lambda/d/\mu$.
\end{example}

In complete analogy with the non-cylindric case we define a
{\em cylindric broken rim hook tableau} to be a sequence $\cT=(\lambda ^{(0)},\lambda
^{(1)},\ldots ,\lambda ^{(\ell )})$ of cylindric loops such that each $b_{i}=\lambda
^{(i)}/d_{i}/\lambda ^{(i-1)}$ is a broken rim hook of length $|b_i|<n$ and set $\prob(\cT)=\prod_{i=1}^\ell\prob(\lambda^{(i)}/d_i/\lambda^{(i-1)})$. We call $\alpha=(\alpha_1,\ldots,\alpha_\ell)$ with $\alpha_i=|\lambda^{(i)}/d_{i}/\lambda ^{(i-1)}|<n$ the weight of $\cT$. For an example of a cylindric broken rim hook tableau see Figure \ref{fig:cylbrimhooktab}.

N.B. we have excluded the case of broken rim hooks of length greater or equal than $n$. First note that $\tau_r=0$ for $r>n$ by definition of the transfer matrix. Furthermore, $\tau_n|V_k=qD_n|V_k=q\omega_2^k\omega_4^{n-k}\cdot 1|V_k$. Because the circumference of the cylinder is $n$,  any rim hook of length $n$ must be connected and, thus, is just a ribbon winding around the cylinder yielding a factor of $q \omega_2^k\omega_4^{n-k}$.

\begin{corollary}
Let $0\le\alpha_i<n$ for $i=1,\ldots,\ell$. Then the matrix elements of the transfer matrix for periodic boundary conditions are given by
\begin{equation}\label{tau2tab}
\langle \lambda |\tau_{\alpha_1}\tau_{\alpha_2}\cdots \tau_{\alpha_\ell }|\mu\rangle =
\sum_{d=0}^\ell q^d \sum_{\lambda /d/\mu =|\mathcal{T}|}\prob(\mathcal{T})\;,
\end{equation}%
where the sum runs  over all cylindric broken rim hook tableaux 
$$\cT=(\mu[0]=\lambda^{(0)}[0],\lambda^{(1)}[d_1],\ldots,\lambda^{(\ell)}[d_\ell]=\lambda[d])$$ 
of weight $\alpha$ and degree $d=d_1+\cdots+d_\ell$ with $0\le d_i\le 1$.
\end{corollary}
\begin{proof}
Because $\prob(\cT)$ factorises into a product, it suffices to prove the assertion for $\ell=1$. If $d=0$ then we have the case of the $A$-operator and our previous result from \eqref{AD6v} applies. Thus, we only need to focus on the case when $d=1$, the $D$-operator. Recall from \eqref{AD6v} that without the cylindric shift $D_rv_\mu$ is a linear combination of vectors $v_\lambda$ such that $\mu/\lambda$ is a broken rim hook of length $n-r$. We wish to show that this implies that $\lambda/1/\mu$ is cylindric broken rim hook of length $r$. According to its definition the path $\lambda[d]$ is obtained from $\lambda[0]$ be adding an $n$-ribbon to the latter. This $n$-ribbon must contain the boxes of $\mu/\lambda$ since the latter is a broken rim hook of length $n-r<n$. Thus, $\lambda/1/\mu$ must also be a broken rim hook that contains those boxes of the added $n$-ribbon that are not in $\mu/\lambda$ and, hence, is of length $r$. Furthermore, we have that $
\overline{\prob}(\mu/\lambda)=\prob(\lambda/d/\mu)
$, which follows from the relations $\#(\mu/\lambda)=\#(\lambda/1/\mu)$ and
\[
\bar c(\mu/\lambda)=\sum_{h\in \lambda/1/\mu}(c(h)-1),\qquad
\bar r(\mu/\lambda)=\sum_{h\in \lambda/1/\mu}(r(h)-1)\;.
\]
The latter are a direct consequence of the definition of the cylindric broken rim hook $\lambda/1/\mu$ from $\mu/\lambda$. The first one is obvious: since $\lambda/1/\mu$ is the complement of the (connected) $n$-ribbon added to $\lambda$ with respect to $\mu/\lambda$, the cylindric broken rim hook must have as many connected components as $\mu/\lambda$. To show the remaining two relations, observe that for each column (row) not intersecting $\mu/\lambda$ there must be a corresponding square in a rim hook $h\in\lambda/1/\mu$. But because $\lambda[1]$ is obtained by adding a connected $n$-ribbon to $\lambda$ for each such rim hook $h$ there is one additional square in the column (row) to the right (below). 
\end{proof}

\begin{figure}\label{fig:cylbrimhooktab}
\centering
\includegraphics[width=.95\textwidth]{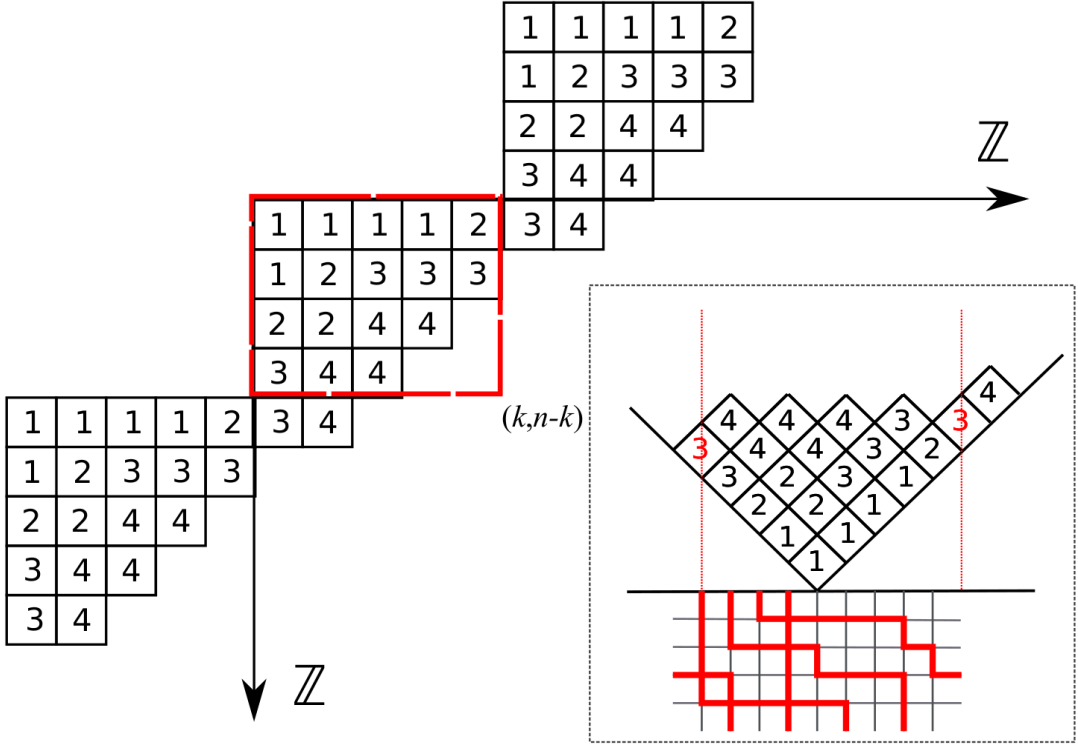} 
\caption{Example of a cylindric broken rim hook tableau for $k=4$ and $n=9$. The corresponding six-vertex configuration on the cylinder is shown in the lower right corner.}
\end{figure}

\subsection{Hecke characters and the free fermion point}
So far we have kept the discussion completely general with respect to the possible choice of Boltzmann weights $\omega$ in Figure \eqref{fig:6vweights} other than that $\omega_1$ and $\omega_3$ have inverses; see \eqref{quadric}. In order to connect with the previous discussion of Hecke characters, we now make the following special choice for the Boltzmann weights $\omega$:
\begin{equation}\label{FFBweights}
\omega_1=\omega_3=\omega_5=1,\qquad\omega_2=a x_i,\quad\omega_4=b x_i,\quad\omega_6=(a+b)x_i\;.
\end{equation}
Note that under the choice \eqref{FFBweights} the first quadric in \eqref{quadric} vanishes, $\Delta=0$, while $\Gamma=abx_i^2$.  For the symmetric six-vertex model the case $\Delta=0$ is usually referred to as {\em the free fermion point} in the physics literature. 

Setting $a=-1$, $b=t$ we recover \eqref{HeckeBweights}, while setting $a=-t$ and $b=1$ we obtain \eqref{HLBweights}. Thus, by introducing the variables $a,b$ in \eqref{FFBweights} we can treat both cases at once. We shall think of the $x_i$, the `spectral parameters', as commuting indeterminates with the label $i$ being the lattice row in which the respective vertex weights occur, that is, we choose the Boltzmann weights $\omega$ differently in each lattice row and denote the transfer matrix corresponding to the $i$th row by $\tau(x_i)=\tau(x_i;a,b)$. Analogous to the previous discussion of vertex operators, we interpret $\tau(x_i)=\sum_{r\ge 0}x_i^r\tau_r$ as `generating series' in the formal variable $x_i$ with the coefficients $\tau_r$ being endomorphisms of the vector space $\cV=\C[a,b]\otimes V^{\otimes n}$.

The next proposition shows, for the more general case of quasi-periodic boundary conditions, that under the choice \eqref{FFBweights} the asymmetric six-vertex transfer matrices factorises into two five-vertex transfer matrices of the so-called vicious and osculating walker models; see \cite{korff2014quantum} and references therein.

\begin{proposition}
Let $\tau ^{\prime }(x_{i})=\tau(x_i;0,1)$, $\tau^{\prime \prime }(x_i)=\tau(x_{i};1,0)$ denote respectively the
transfer matrices of the vicious walker and osculating walker model at the free fermion point where we set $a=0,b=1$ and $a=1,b=0$ in \eqref{FFBweights}. Then we have the factorisation%
\begin{equation}
\tau (x_{i};a ,b )=\tau ^{\prime \prime }(a x_{i})\tau^{\prime }(b x_{i})\;.  \label{factorize}
\end{equation}
In particular, setting $q=0$ we obtain that $A(x;a,b)=A'(a x)A''(b x)$, which describes a case of open boundary conditions on the finite lattice with $n$ sites.
\end{proposition}

\begin{proof}
Decompose $V\otimes V\cong W\oplus W^{\perp }$ where $W=\mathbb{C}%
w_{0}\oplus \mathbb{C}w_{1}$, $W^{\perp }=\mathbb{C}w_{2}\oplus \mathbb{C}%
w_{3}$ and the isomorphism is given by $v_{0}\otimes v_{0}\mapsto w_{0}$, $%
v_{0}\otimes v_{1}+v_{1}\otimes v_{0}\mapsto w_{1}$, $v_{0}\otimes
v_{1}-v_{1}\otimes v_{0}\mapsto w_{2}$ and $v_{1}\otimes v_{1}\mapsto w_{3}$.

Let $R^{\prime }=R^{\prime }(x_{i};a=0,b=1)$ be the vicious walker $R$-matrix and $%
R^{\prime \prime }=R^{\prime \prime }(x_{i};a=0,b=1)$ the osculating $R$-matrix.
Then one shows by a direct computation that the following block
decomposition with respect to the above isomorphism $V\otimes V\cong W\oplus
W^{\perp }$ holds true,%
\begin{equation*}
R_{13}^{\prime \prime }(a x_{i})R_{23}^{\prime }(b x_{i})=\left( 
\begin{array}{cc}
R(x_{i};a ,b ) & \ast \\ 
0 & 0%
\end{array}%
\right) \;.
\end{equation*}%
The assertion now easily follows from the definition of the transfer matrices $\tau',\tau''$ as partial traces of the monodromy matrices $T',T''$ which consist of products of $R',R''$-matrices.
\end{proof}

The transfer matrices of the vicious and osculating walker models are known to commute; see e.g. \cite{korff2014quantum}. Therefore, it follows that any two six-vertex transfer matrices with different choices of $(a,b)$ in \eqref{FFBweights} must commute as well. More generally, we can consider the Yang-Baxter algebras defined in terms of the monodromy matrices \eqref{mom} for two such independent choices of $(a,b)$ in  \eqref{FFBweights}. While the first quadric  $\Delta$ in \eqref{quadric} vanishes, the second quadric $\Gamma$ in general differs for different values of $a,b$, violating the condition \eqref{quadric=} of Prop \ref{prop:Baxter}. Nevertheless, we have the following result (possibly known to experts, but which I was unable to find in the literature):

\begin{proposition}
The monodromy matrices of the six-vertex model with weights \eqref{FFBweights} satisfy the Yang-Baxter
equation 
\begin{equation}\label{RTT}
R_{12}(x_i,x_j)T_{1}(x_i;a,b))T_{2}(x_j;a',b')=T_{2}(x_j;a',b')T_{1}(x_i;a,b)R_{12}(x_i,x_j),
\end{equation}%
where $R(x_i,x_j)=R(x_i,x_j;a,a',b,b')$ is the asymmetric six-vertex $R$-matrix with Boltzmann weights%
\begin{gather}
\omega_1=b x_i+a'x_j,\qquad\omega_2=a x_i+b' x_j,\qquad\omega_3=-a x_i+a' x_j,\\
\omega_4=b x_i-b' x_j,\qquad\omega_5=(a'+b')x_j,\qquad\omega_6=(a+b) x_i\;.
\end{gather}
\end{proposition}
\begin{proof}
It suffices to check this for a lattice with one site, as the monodromy matrix $T$ consists of a product of $R$-matrices. The computation is somewhat tedious and lengthy but straightforward and consists of checking individual matrix elements. We omit the computational details.
\end{proof}
\begin{remark}\rm 
Note that some of the $R$-matrices in \eqref{RTT} can become singular for particular choices of $a,a',b,b'$. Nevertheless, one can show using \eqref{factorize} that the corresponding transfer matrices commute, since the transfer matrices $\tau'(u)$ and $\tau''(v)$ commute for any pair of formal variables $u,v$; see e.g. \cite{korff2014quantum}.
\end{remark}

\subsection{Bethe ansatz equations and quantum cohomology}
The main purpose of this section is to describe the eigenvalues of the six-vertex transfer matrix with quasi-periodic boundary conditions and Boltzmann weights \eqref{FFBweights}. First, we use the Bethe ansatz to obtain an algebraic description of the eigenvalues as symmetric functions in the so-called Bethe roots, solutions to a set of polynomial equations called the Bethe ansatz equations. In the second half of this section we then describe the transfer matrix as a multiplication operator in a particular quotient of the ring of symmetric functions that we show to be a two-parameter extension of the small quantum cohomology ring of the Grassmannian.

Consider the decomposition $\cV=\bigoplus_{k=0}^n\cV_k$ with $\cV_k\cong \C[a,b]\otimes V_k$.
\begin{proposition}
For each $k=0,1,\ldots,n$ the restricted transfer matrix $\tau(x;a,b)|\cV_k$ has eigenvalues
\begin{equation}\label{spectransfer}
(1+(-1)^k q x^nb^n)\prod_{i=1}^k\frac{1+a x \xi_i}{1-b x \xi_i}\;,
\end{equation}
where $\xi_1,\ldots,\xi_k$ are $k$ {\em distinct} solutions of the Bethe ansatz equations, i.e. satisfy
\begin{equation}\label{freeBAE}
\xi_i^n+(-1)^kq=0,\qquad i=1,2,\ldots,k\;.
\end{equation}
In particular, the Bethe ansatz is `complete', that is all eigenvalues of the transfer matrix are obtained this way and the latter is diagonalisable.
\end{proposition}
\begin{proof}
The solution of the eigenvalue problem of the transfer matrix is a standard computation using the algebraic Bethe ansatz or the quantum inverse scattering method and has been carried out for the five-vertex transfer matrices $\tau'$, $\tau''$ in \cite{korff2014quantum}. It then follows at once from \eqref{factorize} that the common eigenbasis constructed previously for $\tau'$, $\tau''$ is also an eigenbasis of $\tau$. We therefore describe only briefly the various steps involved. 

One makes the ansatz that the eigenvectors of the transfer matrix $\tau=A+qD$ are of the form $|y\rangle=B(y_1)B(y_2)\cdots B(y_k)v_0$ with the so-called `pseudovacuum' vector $v_0$ spanning $V_0\cong\C$. Using the Yang-Baxter algebra relations \eqref{RTT} one commutes the transfer matrix  $\tau=A+qD$ past the $B$-operators and, noting that $\tau v_0=(1+q\omega_4^n)v_0$, one derives necessary conditions on the $y_i$ for $|y\rangle$ to be an eigenvector. This computation yields  the equations \eqref{freeBAE}. For the 5-vertex models $\tau'$, $\tau''$ the eigenvectors have been shown \cite{korff2014quantum} to be of the following explicit form,
\begin{equation}\label{Bethevector}
|\xi\rangle=\sum_{\lambda}s_{\lambda}(\xi_1^{-1},\ldots,\xi_k^{-1})v_\lambda
\end{equation}
where the sum runs over all partitions $\lambda\in\cP^+_{k,n}$, $s_\lambda$ is the Schur polynomial in $k$ variables and $\xi=(\xi_1,\ldots,\xi_k)$ are $k$ solutions to \eqref{freeBAE} with the $\xi_j$ being mutually distinct.

From the same computation one also infers the form \eqref{spectransfer} of the eigenvalues as the action of $\tau$ on $v_0$ can be easily computed. Finally, one verifies that the solutions to \eqref{freeBAE} give rise to $\binom{n}{k}=\dim V_k$ distinct eigenvectors, because the $B$-operators mutually commute, whence the Bethe ansatz is `complete'. 
\end{proof}

Note that neither the eigenvectors \eqref{Bethevector} nor the Bethe ansatz equations \eqref{freeBAE} depend on $x$ or $a,b$. Since the matrix entries of each $\tau_r$ are elements in $\Z[a,b]$ for $r=1,\ldots,n$, it follows that their eigenvalues are elements in $\Z[a,b]$ as well. The latter are obtained via a series expansion of \eqref{spectransfer} with respect to $x$, which must terminate after $n+1$ terms as each matrix entry $\tau(x;a,b)$ is at most of degree $n$ in $x$.

\begin{lemma} Suppose that $q^{\pm 1/n}$ exist and that under complex conjugation $\overline{q^{1/n}}=q^{-\frac{1}{n}}$. Then the solutions of the Bethe ansatz equations \eqref{freeBAE} for fixed $n$ and $k$ are given by the following discrete set in $\C[q^{\pm 1/n}]^k$,
\begin{equation}\label{BA_solns}
\Xi_{k,n}=\{\xi=(\xi_1,\ldots,\xi_k)~|~\xi_j=q^{\frac{1}{n}}e^{\frac{2\pi i}{n}(\frac{k+1}{2}+\lambda_j-j)},\;\lambda\in\cP^+_{k,n}\}\;.
\end{equation}
\end{lemma}
For a proof see e.g. \cite[Prop 10.4]{korffstroppel2010}. Note that we need both roots, $q^{\pm 1/n}$, as the eigenvalues of the transfer matrix will depend on $\xi_j$ and the eigenvectors \eqref{Bethevector} depend on $\xi_j^{-1}$. %

Since the $B$-operators mutually commute, the coefficients of the eigenvectors \eqref{Bethevector} are symmetric polynomials in the Bethe roots $\xi^{-1}_j$ and the eigenvalues of the transfer matrix symmetric polynomials in the $\xi_j$. Therefore, we are interested in describing the properties of symmetric polynomials when the latter are evaluated at solutions of \eqref{freeBAE}.

Introduce the symmetric polynomial $\Delta(y)=\prod_{1\le i,j\le k, i\neq j}(y_i-y_j)$ and consider the localisation of $\Lambda_k[q^{\pm\frac{1}{n}}]=\C[q^{\pm\frac{1}{n}}][y_1,\ldots,y_k]^{S_k}$ at $\Delta$, which we shall denote by $\Lambda_k[q^{\pm\frac{1}{n}},\Delta^{-1}]$. The latter is needed to capture the property that the Bethe roots $\xi_j$ of \eqref{freeBAE} are mutually distinct.
\begin{lemma}
Let $\mc{I}_{k,n}=\mc{I}(\Xi_{k,n})$ be the vanishing ideal of $\Xi_{k,n}$ in $\Lambda_k[q^{\pm\frac{1}{n}},\Delta^{-1}]$ and denote by $h_r$ the complete symmetric polynomials in the $y_i$. Then
\begin{equation}\label{QCideal}
\mc{I}_{k,n}=\langle h_{n+1-k},\ldots,h_{n-1},h_n+(-1)^kq\rangle
\end{equation}
and, vice versa, the set of zeroes of the ideal $\langle h_{n+1-k},\ldots,h_{n-1},h_n+(-1)^kq\rangle$ is given by $\Xi_{k,n}$. 
\end{lemma}
\begin{proof}
One first shows that $\mc{I}_{k,n}$ is radical and then uses Hilbert's Nullstellensatz. The proof follows the same lines as \cite[Proof of Theorem 6.20]{korffstroppel2010} and we therefore omit the details.
\end{proof}
The coordinate ring $R_{k,n}=\Lambda_k[q^{\pm\frac{1}{n}},\Delta^{-1}]/\mc{I}_{k,n}$ is known to be isomorphic to $qH^*(\op{Gr}_{k}(\C^n);\Z)\otimes_\Z\C[q^{\pm\frac{1}{n}}]$, where 
\begin{equation}\label{QCring}
qH^*(\op{Gr}_k(\C^n);\Z)=\Z[q,e_1,\ldots,e_k]/\mc{I}_{k,n}
\end{equation} 
is the small quantum cohomology ring of the Grassmannian $\op{Gr}_k(\C^n)$ and the $e_i$ are the elementary symmetric polynomials in the variables $y=(y_1,\ldots,y_k)$. 
\begin{remark}\label{rem:QC}\em
The presentation \eqref{QCring} of the small quantum cohomology ring $qH^*(\op{Gr}_k(\C^n);\Z)$ is due to Siebert and Tian \cite{siebert1997}. Geometrically, the polynomials $e_i$ and $h_i$ can be identified with the Chern classes of the tautological and the quotient bundle, respectively. The corresponding Chern polynomials are the eigenvalues of the transfer matrices $\tau',\tau''$ in \eqref{factorize}; see \cite{korff2014quantum}. The Schur polynomials, which can be expressed as the determinants $\{s_\lambda=\det(h_{\lambda_i-i+j})_{1\le i,j\le k}~|~\lambda_1\le n-k\}$, then represent the Schubert classes, which form a basis of the ring. Changing the base to $\C[q^{\pm\frac{1}{n}}]$ yields a semi-simple ring; see e.g. \cite{bertram1996gromov,siebert1997} and \cite[Prop 6.5]{abrams2000quantum}. In the latter works the Bethe roots are the Chern roots, which can be identified with the critical points of a Landau-Ginzburg potential when representing $qH^*(\op{Gr}_k(\C^n);\Z)$ as a Jacobi-algebra.
\end{remark}

Denote by $\op{Func}(\cP^+_{k,n},\C[q^{\pm\frac{1}{n}}])$ the set of functions $f:\cP^+_{k,n}\to\C[q^{\pm\frac{1}{n}}]$ endowed with the operations of pointwise addition and multiplication. 
\begin{proposition}\label{prop:evaluate}
There exists a ring isomorphism $\op{ev}:R_{k,n}\to \op{Func}(\cP^+_{k,n},\C[q^{\pm\frac{1}{n}}])$ which assigns to each symmetric polynomial $F(y)$ the function $f(\lambda)=F(\xi_\lambda)$, where $\xi_\lambda$ is the solution \eqref{BA_solns} fixed by $\lambda\in\cP^+_{k,n}$.
\end{proposition}

\begin{proof}
We verify the isomorphism $\op{ev}$ by identifying the idempotents in both rings. For each $\lambda\in\cP^+_{k,n}$ introduce the polynomials 
\[
F_\lambda(y)=\sum_{\nu\in\cP^+_{k,n}}\frac{s_{\nu}(\xi^{-1}_\lambda)s_{\nu}(y)}{n^k\prod_{i<j}|\xi_i(\lambda)-\xi_j(\lambda)|^{-2}}\;.
\]
Because of the identity
(see e.g. \cite{korffstroppel2010})
 \begin{equation}\label{complete}
\sum_{\nu\in\cP^+_{k,n}}s_{\nu}(\xi_\lambda)s_{\nu}(\xi_\mu^{-1})=\delta_{\lambda\mu}\;\frac{n^k}{\prod_{i<j}|\xi_i-\xi_j|^{2}}\;.
\end{equation}
the $F_\lambda$ map under $\op{ev}$ to the functions $\delta_\lambda:\cP^+_{k,n}\to\C[q^{\pm\frac{1}{n}}]$ defined by $\delta_{\lambda}(\mu)=\delta_{\lambda\mu}$. The latter are obviously the idempotents of $\op{Func}(\cP^+_{k,n},\C[q^{\pm\frac{1}{n}}])$. 

In order to see that the $F_\lambda$ are the idempotents of $R_{k,n}$, note that
\begin{eqnarray*}
F_\lambda(y) F_\mu(y)&=&\sum_{\alpha,\beta,\gamma\in\cP^+_{k,n}}
\frac{q^dC_{\alpha\beta}^{\gamma,d}
s_{\alpha}(\xi^{-1}_\lambda)s_{\beta}(\xi^{-1}_\mu)s_{\gamma}(y)}{n^{2k}\prod_{i<j}|\xi_i(\lambda)-\xi_j(\lambda)|^{-2}|\xi_i(\mu)-\xi_j(\mu)|^{-2}}
\end{eqnarray*}
Inserting the Bertram-Vafa-Intriligator formula\footnote{We use here the variant as stated in \cite[Corollary 6.2]{rietsch2001quantum}.} \cite{bertram2005two,vafa1991topological,intriligator1991fusion} for the Gromov-Witten invariants $q^dC_{\alpha\beta}^{\gamma,d}$,
\begin{equation}\label{BVI}
q^dC^{\gamma,d}_{\alpha\beta}=
\sum_{\xi\in \Xi_{k,n}} \frac{s_{\gamma}(\xi^{-1})s_{\alpha}(\xi)s_{\beta}(\xi)}{n^{k}\prod_{i<j}|\xi_i-\xi_j|^{-2}}\;,
\end{equation}
and performing the summation over $\alpha,\beta$ we obtain
\[
F_\lambda(y) F_\mu(y)=\delta_{\lambda\mu}\sum_{\gamma\in\cP^+_{k,n}}\frac{s_{\gamma}(\xi^{-1}_\lambda)s_{\gamma}(y)}{n^k\prod_{i<j}|\xi_i(\lambda)-\xi_j(\lambda)|^{-2}}=\delta_{\lambda\mu}F_\lambda(y)\;.
\]
\end{proof}
This concludes the analysis of the solutions of the Bethe ansatz equations \eqref{freeBAE}. We now turn to the description of the eigenvalues \eqref{spectransfer} of the six-vertex transfer matrix with Boltzmann weights \eqref{FFBweights} and quasi-periodic boundary conditions using \eqref{factorize}. In light of the expression \eqref{spectransfer} we consider the generating function
\begin{equation}\label{littleh}
\prod_{i=1}^k\frac{1+a x y_i}{1-b x y_i}=\sum_{r\ge 0}x^r h_r(y;a,b)
\end{equation}
in the formal variable $x$, which implicitly defines the symmetric polynomials $h_r(y;a,b)\in\Lambda_k\otimes\C[a,b]$. The latter contain the elementary $e_r(y)=h_r(y;1,0)$ and the complete symmetric polynomials $h_r(y)=h_r(y;0,1)$ as special cases.
\begin{lemma}\label{lem:bethe2h}
The Bethe ansatz equations \eqref{freeBAE} imply the identities
\begin{equation}\label{BAEideal0}
\left\{
\begin{array}{l}
h_n(y;a,b)-q(-1)^kb^{n-k}((-a)^k-b^k)=0\\
h_{r+n}(y;a,b)+q(-1)^kb^n h_r(y;a,b)=0,\quad r>0
\end{array}
\right.\;.
\end{equation}
Conversely, if the above identities hold and the inverse of the element
\begin{equation}\label{locus}
\Delta_{a,b}=\prod_{i,j=1}^k(a y_i+b y_j)
\end{equation}
exists, then each of the variables $y=(y_1,\ldots,y_k)$ must satisfy \eqref{freeBAE}.
\end{lemma}
\begin{proof}
Assume that $\xi=(\xi_1,\ldots,\xi_k)$ is a solution of \eqref{freeBAE}. Then it follows that \eqref{spectransfer} is an eigenvalue of the transfer matrix $\tau(x;a,b)$. The matrix elements of the latter are polynomial in $x$ and at most of degree $n$ and therefore it follows that \eqref{spectransfer} must be polynomial in $x$ as well. N.B. the Bethe eigenvectors \eqref{Bethevector} of the transfer matrix do not depend on $x$. Hence, the series expansion of \eqref{spectransfer} at $x=0$ must terminate after $n$ terms giving the second relation in \eqref{BAEideal0}. Noting further that
\[
\tau(x;a,b)v_\lambda
=qx^na^kb^{n-k}v_\lambda
+\ldots
\]
with the omitted terms having coefficients of degree strictly less than $n$ in $x$, one arrives at the first relation in \eqref{BAEideal0}.

To show that \eqref{BAEideal0} implies \eqref{freeBAE}, first note that the former relations imply that \eqref{spectransfer} is polynomial in $x$, as just discussed. Hence, the residues of \eqref{spectransfer} at $x=(b y_i)^{-1}$ for each $i=1,\ldots,k$ must vanish,
\[
\op{Res}_{x=1/by_i}(1+(-1)^k q x^nb^n)\prod_{i=1}^k\frac{1+a x \xi_i}{1-b x \xi_i}=
-(a+b)\frac{y_i^n+(-1)^kq}{b^{k+1}y_i^{n+1}}\prod_{j\neq i}\frac{b y_i+a y_j}{y_i-y_j}=0
\]
This gives \eqref{freeBAE}, provided that $a y_i+b y_j\neq 0$ for all $i,j=1,\ldots,k$.
\end{proof}

Define symmetric polynomials $f_i\in\Lambda_k[q^{\pm\frac{1}{n}},\Delta^{-1}]\otimes\C[a,b]=\Lambda_k[q^{\pm\frac{1}{n}},\Delta^{-1},a,b]$ with $i=0,1,\ldots,k-1$ via 
\[
\begin{array}{l}
f_0(y)=h_n(y;a,b)-q(-1)^kb^{n-k}((-a)^k-b^k),\quad\\ 
f_i(y)=h_{i+n}(y;a,b)+q(-1)^kb^n h_{i}(y;a,b),\quad i=1,\ldots,k-1
\end{array}\;
\]
and set $\mc{J}_{k,n}=\langle f_0,f_1,\ldots,f_{k-1}\rangle$.


\begin{lemma}\label{lem:6vquotient}
We have the following equality of localised ideals, 
\begin{equation}
\mc{J}_{k,n}[\Delta^{-1}_{a,b}]=(\mc{I}_{k,n}\otimes\C[a,b])[\Delta^{-1}_{a,b}]\;, 
\end{equation}
and, thus, there is a ring isomorphism
\begin{equation}\label{6vquotient}
\Lambda_k[q^{\pm\frac{1}{n}},\Delta^{-1},a,b,\Delta_{a,b}^{-1}]/\mc{J}_{k,n}\cong
R_{k,n}[a,b,\Delta_{a,b}^{-1}]
\;,
\end{equation}
where the right hand side is the localisation of the coordinate ring $R_{k,n}\otimes\C[a,b]$ %
introduced  earlier at the element \eqref{locus}.
\end{lemma}

\begin{proof}
We first show that $\mc{J}_{k,n}[\Delta^{-1}_{a,b}]\subset(\mc{I}_{k,n}\otimes\C[a,b])[\Delta^{-1}_{a,b}]$. Recall that setting the polynomials in \eqref{QCideal} to be identically zero is equivalent to the following identity
\footnote{In light of our earlier Remark \ref{rem:QC} the identity for $q=0$ is the Whitney sum formula, stating that the direct sum of the tautological and quotient bundle of the Grassmannian is trivial, and for $q=1$ it is the generalisation of this formula to the Verlinde algebra \cite{witten1993verlinde}.} in the dummy variable $x$,
\begin{equation}\label{TQ}
\left(\sum_{i=0}^k (-x)^ie_i\right)\left(\sum_{j=0}^{n-k} x^jh_j\right)=1+q(-1)^kx^n\;.
\end{equation}
To see this, note that the first bracket containing the $e_i$ equals the product $\prod_{i=1}^k(1-xy_i)$. Dividing by the latter in \eqref{TQ} we obtain 
\[
\sum_{j=0}^{n-k} x^jh_j=(1+q(-1)^kx^n)\prod_{i=1}^k\frac{1}{1-xy_i}=(1+q(-1)^kx^n)\sum_{j\ge 0}x^jh_j
\]
Comparing coefficients of $x^j$ with $j=n-k+1,\ldots,n$ then yields the relations of the ideal \eqref{QCideal}. Note that we have $k$ independent variables, so the remaining relations for other values of $j$ must be algebraically dependent.

Inserting the identity \eqref{TQ} with $x\to b x$ in \eqref{spectransfer} we arrive at
\[
(1+(-1)^k q x^nb^n)\prod_{i=1}^k\frac{1+a x y_i}{1-b x y_i}=
\left(\sum_{i=0}^k (a x)^ie_i\right)\left(\sum_{j=0}^{n-k} (b x)^jh_j\right)\;,
\]
which is obviously polynomial of maximal degree $n$ in $x$ with leading coefficient $a^kb^{n-k}e_kh_{n-k}=qa^kb^{n-k}$. Because the generating relations \eqref{BAEideal0} of $\mc{J}_{k,n}[\Delta^{-1}_{a,b}]$ are equivalent to showing that \eqref{spectransfer} is polynomial of degree $n$ in $x$ with leading coefficient $qa^kb^{n-k}$, see Lemma \ref{lem:bethe2h}, we have shown that $\mc{J}_{k,n}[\Delta^{-1}_{a,b}]\subset(\mc{I}_{k,n}\otimes\C[a,b])[\Delta^{-1}_{a,b}]$.

We now show the converse, i.e. that $(\mc{I}_{k,n}\otimes\C[a,b])[\Delta^{-1}_{a,b}]\subset\mc{J}_{k,n}[\Delta^{-1}_{a,b}]$. Consider the generating function identity
\[
\prod_{i=1}^k\frac{1+a x y_i}{1-b x y_i}=1+x(a+b)\sum_{i,j\ge 0}x^{i+j}a^jb^is_{(i|j)}(y)\;,
\]
where $s_{(i|j)}=s_{(i+1,1^j)}=h_{i+1}e_j-h_{i+2}e_{j-1}+\ldots+(-1)^jh_{i+j+1}$ is the hook Schur function. Then
\begin{equation}\label{hookschurexp}
h_{r+1}(y;a,b)=(a+b)\sum_{i=0}^ra^ib^{r-i}s_{(r-i|i)}(y)\;.
\end{equation}
The coefficient of $a^jb^i$ in the expansion \eqref{hookschurexp} is %
$
s_{(i-1|j)}+s_{(i|j-1)}=h_ie_j\;.
$ %
Thus, for $i+j=n$ we obtain from $h_n(a,b)+q(-1)^kb^n-qa^kb^{n-k}$ that $h_{n-j}e_j$ with $j=1,\ldots,k-1$ lies in the ideal $\mc{J}_{k,n}$. For $j=k$ we obtain the additional relation $h_{n-k}e_k-q=0$ in \eqref{TQ}, while for $i=n,j=0$ one arrives at $h_{n}+q(-1)^k=0$. 

We now show the ring isomorphism. Let $S=\C[a,b]$. Because the sequence
\[
0\to\mc{I}_{k,n}\to\Lambda_k[q^{\pm\frac{1}{n}},\Delta^{-1}]\to R_{k,n}\to 0
\]
is exact, we also have that
 $
0\to\mc{I}_{k,n}\otimes S\to\Lambda_k[q^{\pm\frac{1}{n}},\Delta^{-1}]\otimes S\to R_{k,n}\otimes S\to 0
$ %
is exact and after localising at the element $\Delta_{a,b}$ the assertion follows.
\end{proof}
In summary, Lemma \ref{lem:6vquotient} allows us to identify the eigenvalue problem of the six-vertex transfer matrix with Boltzmann weights \eqref{FFBweights} with multiplication in the quotient ring \eqref{6vquotient}.

\subsection{A fermionic version of the rim hook algorithm}

Having established a correspondence between the coordinate ring describing the solutions of the Bethe ansatz equations \eqref{freeBAE} and the small quantum cohomology of the Grassmannians, we can apply the following result, known as `rim hook algorithm' \cite{bertram1999}.

Recall that for $n$ fixed the $n$-core $\dot\lambda$ of a partition $\lambda$ is obtained by deleting a maximal number of (connected)  $n$-rim hooks from its Young diagram; see e.g. \cite[Ch.I.1]{macdonald1998}. The resulting $n$-core is unique. In particular, it does not depend on the choice of $n$-rim hooks removed. The $n$-weight $d(\lambda)$ of $\lambda$ is the maximal number of $n$-rim hooks which can be removed. 

\begin{lemma}[Bertram, Ciocan-Fontanine, Fulton]\label{lem:rimhook}
Let $\lambda$ be a partition and denote by $\dot\lambda$ its $n$-core. The following describes the projection $\pi_{k,n}:\Lambda\twoheadrightarrow qH^*(\op{Gr}_k(\C^n);\Z)$ with respect to the basis of Schur functions,
\begin{equation}\label{rimhookalg}
\pi_{k,n}:\;s_\lambda\mapsto \left\{
\begin{array}{ll}
(-1)^{(k-1)d(\lambda)}q^{d(\lambda)}\op{sgn}(\lambda) s_{\dot\lambda}, & \ell(\lambda)\le k\text{ and }\dot\lambda_1\le n-k\\
0, & \text{ else}
\end{array}\right.\;,
\end{equation}
where $\op{sgn}(\lambda)=(-1)^{\sum_{i=1}^d(r(h_i)-1)}$ and the sum in the exponent runs over all $n$-rim hooks $h_i$ which have been removed to obtain $\dot\lambda$.
\end{lemma}

For the proof of the above lemma we refer the reader to \cite{bertram1999}. We now state the fermionic analogue of Lemma \ref{lem:rimhook} in terms of Maya diagrams by exploiting the boson-fermion correspondence. 

The process of finding the $n$-core $\dot\lambda$ of $\lambda$ in \eqref{rimhookalg} can be nicely described in terms of the Maya diagram $\sigma(\lambda,c)$ (where $c\in\Z$ is arbitrary) as follows \cite{james1981}: form an abacus with $n$  infinite vertical runners (oriented south to north) where we place a bead on the $i$th runner at height $h\in\Z$ if $\sigma_{j}=1$ for $j=i+hn$ with $i=0,1,\ldots,n-1$. For $h\ll 0$ the runners will be filled with beads without any gaps due to the definition of the Maya diagram. The $n$-core is now obtained by moving on each runner the beads downwards as far as they can go. The resulting Maya diagram $\dot\sigma$ of charge $c$ corresponds uniquely to a partition $\dot\lambda$ under the bijection \eqref{Maya}, which is the $n$-core of $\lambda$.


%

There is one more ingredient we need from the abacus configuration associated with $\lambda$: starting from the lowest row of the abacus in which there exists a gap, number the beads consecutively from left to right and then move up to the next row. Following \cite[Section 2.7]{james1981} we call this a {\em natural numbering} of the beads. Keeping the numbering unchanged move the individual beads downwards as far as they go to obtain the $n$-core $\dot\lambda$. The resulting numbering of the beads for $\dot\lambda$ will in general differ from the natural numbering for $\dot\lambda$. Let $w_\lambda$ be the permutation which sends the natural numbering of $\dot\lambda$ to the numbering obtained from $\lambda$ as described above.

\begin{lemma}
We have the following equality of sign factors, $\op{sgn}(w_\lambda)=\op{sgn}(\lambda)$, where the factor on the right hand side has been defined in Lemma \ref{lem:rimhook}. 
\end{lemma}
\begin{proof}
The proof is a straightforward computation using the bijection \eqref{Maya}.
\end{proof}

Let $\bigwedge\C^n$ denote the {\em exterior} or {\em Grassmann algebra} and $\{\epsilon_i\}_{i=1}^n$ the standard basis of $\C^n$. Being a graded algebra $\bigwedge\C^n$ has the decomposition $\bigwedge\C^n=\bigoplus_{k=0}^n\bigwedge^k\C^n$, where we denote as usual by $\bigwedge^k\C^n$ the subspace of degree $k$ spanned by vectors of the form $w=\epsilon_{i_1}\wedge\cdots\wedge\epsilon_{i_k}$ and set $\bigwedge^0\C^n=\C$. It will be convenient to label the basis vectors 
\begin{equation}\label{altbasis}
\bigsqcup_{1\leq k\leq n}\{\epsilon_{i_1}\wedge\cdots\wedge\epsilon_{i_k}~|~1\le i_1<\ldots<i_k\le n\}
\end{equation}
 in terms of binary strings of length $n$ (finite Maya diagrams), i.e. maps $\sigma:\{1,\ldots,n\}\to\{0,1\}$ with $\sigma_i=1$ if $i=i_j$ for some $j=1,\ldots,k$ and $\sigma_i=0$ otherwise.  While we have used here the same symbol as in \eqref{Maya} it will be clear from the context when $\sigma$ describes a finite or an infinite binary string. 
 
Denote by $\cP_n^+\subset\cP^+$ the set of $n$-cores $\lambda$ which obey $\lambda_1+\lambda'_1=n$ and by $\cP^+_{k,n}\subset\cP^+$ the subset for which $\ell(\lambda)=\lambda'_1=k$. 

\begin{lemma}\label{lem:string2core} 
The set $\cP_n^+=\bigsqcup_{0\le k\le n}\cP^+_{k,n}$ is in bijection with the set of binary strings $\sigma=(\sigma_1,\ldots,\sigma_n)$ of length $n$ and, hence, with the basis vectors \eqref{altbasis} of the Grassmann algebra $\bigwedge\C^n$. In particular, the elements of $\cP^+_{k,n}$ are mapped to the basis vectors \eqref{altbasis} in $\bigwedge^k\C^n$ under this bijection.
\end{lemma}
\begin{proof}
Recall the map \eqref{Maya}, $\lambda\mapsto\sigma(\lambda,k)$ with $c=k$ and then restrict the resulting Maya diagram to the set $\{1,\ldots,n\}$. That is, $i_j=k+1+\lambda_j-j$ for $j=1,\ldots,\ell(\lambda)=k$ in \eqref{altbasis}.
\end{proof}

Recall that any binary string $\sigma =(\sigma_1,\ldots,\sigma_n)$ of length $n$ can be associated with a
basis vector $v_\lambda=v_{\sigma _{1}}\otimes \cdots \otimes v_{\sigma _{n}}\in
V^{\otimes n}$ with $\lambda(\sigma)=\dot\lambda$ being the $n$-core fixed by $\sigma$. According to Lemma \ref{lem:string2core} we can think of $\lambda (\sigma )$ as labelling a basis element \eqref{altbasis} in the exterior algebra $\bigwedge\C^n$, where the length $\ell(\lambda)=k$ of the $n$-core fixes the degree of the basis element. In what follows we identify $V^{\otimes n}\cong \bigwedge \C^n$ as vector spaces by mapping 
\begin{equation}\label{finite_wedge}
v_{\sigma _{1}}\otimes \cdots \otimes v_{\sigma _{n}}\mapsto\epsilon_{i_1}\wedge\cdots\wedge\epsilon_{i_k}\;,
\end{equation}
where $1\leq i_1<\ldots <i_k\le n$ are the positions of 1-letters in $\sigma$. By abuse of notation we shall denote both basis vectors with $v_\lambda$ where $\lambda=\lambda(\sigma)$ is the $n$-core fixed by the finite binary string $\sigma$. In particular, we have that $V_k\to\bigwedge^k\C^n$ under this isomorphism.

Define for charge $c=k$  with $0\le k\le n$ the following fermionic version of the projection map $\pi'_{k,n}:\bigwedge^{\frac{\infty}{2},k}V\twoheadrightarrow
\bigoplus_{d\ge 0}q^d\otimes\bigwedge^k\C^n$ via
\begin{equation}\label{fermi_pi}
\pi'_{k,n}:\sigma(\lambda,k)\mapsto\left\{
\begin{array}{cl}
q^{d(\lambda)}(-1)^{(k-1)d(\lambda)}\op{sgn}(w_\lambda)v_{\dot\lambda},& \dot\lambda_1\le n-k\text{ and }\ell(\lambda)\le k\\
0, &\text{else}
\end{array}
\right.\;,
\end{equation}
where, as before in Lemma \ref{lem:rimhook}, $d(\lambda)$ is the $n$-weight of $\lambda$ and the $h_i$ are the $n$-rim hooks removed from $\lambda$ to obtain the $n$-core $\dot\lambda$.
\begin{lemma}\label{lem:frimhook}
Let $\pi_{k,n}:\Lambda\twoheadrightarrow qH^*(\op{Gr}_{k}(\C^n))$ be the projection from the rim hook algorithm. Then the following diagram commutes:
\begin{equation}\label{rimhook_cd}
\begin{tikzcd}
\bigwedge^{\frac{\infty}{2},k}V\otimes_{\C}\C[a,b] \arrow[r, "\imath\otimes 1"] \arrow[d,two heads,"\;\pi'_{k,n}\otimes 1"] & 
\Lambda\otimes_{\C}\C[a,b]
\arrow[d,two heads,"\pi_{k,n}\otimes 1"]\\
(\bigoplus\limits_{d\ge 0}q^d\otimes\bigwedge^k\C^n)\otimes_\C\C[a,b]\arrow[r,"\varphi\otimes 1"]& qH^*(\op{Gr}_k(\C^n))\otimes_{\C}\C[a,b]
\end{tikzcd}
\end{equation}
where $\varphi$ is the Satake correspondence, $\varphi:q^d\otimes v_\lambda\mapsto q^d s_\lambda$.
\end{lemma}
\begin{proof}
It suffices to verify this for the bases of Maya diagrams and Schur functions in $\bigwedge^{\frac{\infty}{2},k}V$ and $\Lambda$, respectively. Recall that $\imath (\sigma(\lambda,k))=s_\lambda$ under the boson-fermion correspondence. Moreover, define $\varphi:\bigoplus\limits_{d\ge 0}q^d\otimes\bigwedge^k\C^n\to qH^*(\op{Gr}_k(\C^n))$ as stated above by $q^d\otimes v_\lambda\mapsto q^d s_\lambda$, where $v_\lambda=\epsilon_{i_1}\wedge\cdots\wedge\epsilon_{i_k}$ as defined after \eqref{finite_wedge}, $s_\lambda=\det(h_{\lambda_i-i+j})_{1\le i,j\le k}$ and $\lambda\in\cP^+_{k,n}$. The latter map is the simplest example of the Satake isomorphism extended to $qH^*(\op{Gr}_{k}(\C^n))$; see e.g. \cite{golyshev2011quantum}. The assertion now follows from the definition of the projection map $\pi'_{k,n}$.
\end{proof}



\section{Cylindric Hecke characters}
In this section we introduce the cylindric Hecke characters of Theorem \ref{thm:main}. In light of the identities \eqref{A2chi}, \eqref{Ainv2chi} and \eqref{BFA}, \eqref{BFAinv} we define the latter as matrix elements of a `re-normalised' six-vertex transfer matrix with quasi-periodic boundary conditions, where  the normalisation factor is fixed by the projection \eqref{fermi_pi} resulting from the rim hook algorithm of Lemma \ref{lem:rimhook}.

\subsection{Rim hook projection of the transfer matrix}
Define the following infinite family of operators $\{H_r(a,b)|\cV_k\}_{r\in\Z_{\ge 0}}$ on each $\cV_k\cong\bigwedge^k\C^n\otimes\C[a,b]$ with $k=0,1,\ldots,n$ by setting
\[
H_r(a,b)|\cV_k=\left\{
\begin{array}{ll}
(-1)^{(k-1)s}q^sb^{sn}\tau_{r'}(a,b), & r=r'+sn,\;1\le r'\le n-1,\;s\in\Z_{\ge 0}\\
(-1)^{(k-1)s}q^sb^{sn}(1-(-a/b)^k), & r=sn,\;s\in\Z_{\ge 0}
\end{array}
\right. \;.
\]
For $\cV\cong\bigwedge\C^n\otimes\C[a,b]$ denote by $H_r(a,b)\in\End\cV$ the operator that block-decomposes into the $H_r(a,b)|\cV_k$ for $k=0,1,\ldots,n$. In terms of formal power series in the spectral variable $x$ the relationship between the six-vertex transfer matrix and the operators $H_r(a,b)|\cV_k$ is compactly written as
\begin{equation}\label{bigH}
H(x;a,b)|\cV_k=\sum_{r\ge 0}x^rH_r(a,b)|\cV_k=\frac{\tau(x;a,b)|\cV_k}{1+(-1)^kqb^nx^n}\;,\qquad k=0,1,\ldots,n,
\end{equation}
where $(1+(-1)^kqb^nx^n)^{-1}=\sum_{r\ge 0}x^{rn}((-1)^{k-1})qb^n)^r$. Thus, we have in particular that
\[
H_r(a,b)=\tau_r(a,b),\qquad 0\le r<n\;.
\]
As an immediate consequence of the Bethe ansatz for the six-vertex transfer matrix we then obtain the following:
\begin{corollary}
The operators $H_r(a,b)$ all commute with each other and $H_r(a,b)|\cV_k$ has eigenvalues $h_r(\xi;a,b)$ where $\xi$ is a solution of the Bethe ansatz equation \eqref{freeBAE}. Moreover,
\[
H(x;a,b)H(x;-b,-a)=H(x;a,b)H(-x;b,a)=1\;.
\]
\end{corollary}

Our motivation to introduce \eqref{bigH} is the following result, which states that this `renormalised' six-vertex transfer matrix is the projection of the half-vertex operators \eqref{HeckeA} and \eqref{HLA} under the rim hook algorithm.

\begin{lemma}
Let $(\pi'_{k,n}\otimes1):\bigwedge^{\frac{\infty}{2},k}V\otimes\C[t]\twoheadrightarrow \bigwedge^k\C^n\otimes\C[t]$ be the projection induced by the rim hook algorithm via \eqref{rimhook_cd} when evaluating either at $a=-1,b=t$ or at $a=-t,b=1$. Then in the former case we have
\begin{equation}\label{A2H}
(\pi'_{k,n}\otimes 1)\circ A(x;t)=H(x;-1,t)\circ(\pi'_{k,n}\otimes 1)
\end{equation}
and in the latter case
\begin{equation}\label{Ainv2H}
(\pi'_{k,n}\otimes 1)\circ A^{-1}(x;t)=H(x;-t,1)\circ(\pi'_{k,n}\otimes 1)\;,
\end{equation}
where $A(x;t)$ and $A^{-1}(x;t)$ are the half-vertex operators \eqref{HeckeA} and \eqref{HLA}, respectively.
\end{lemma}
\begin{proof}
Set $a=-1$ and $b=t$. Employing the boson-fermion correspondence the action of $A_r$ corresponds to multiplying with $h_r[(t-1)y]$ in $\Lambda[t]=\Lambda\otimes\C[t]$, while under the Satake correspondence $H_r(-1,t)$ corresponds to multiplying with $h_r[(t-1)y]$ in $qH^*(\op{Gr}_k(\C^n))\otimes\C[t]$. The assertion then follows from the commutativity of the diagram \eqref{rimhook_cd}. The proof of the second relation is analogous.
\end{proof}

\subsection{A cylindric Murnaghan-Nakayama rule}
Let $\lambda,\mu\in\cP^+_{k,n}$ and $\nu\in\cP^+$. Then we define cylindric skew Hecke characters 
$\chi_{t}^{\lambda/d/\mu}$ via the following matrix elements
\begin{equation}\label{cylchi}
q^d\chi_{t}^{\lambda/d/\mu}(\nu)=\frac{
\langle\lambda'|H_{\nu_1}H_{\nu_2}\cdots H_{\nu_\ell}|\mu' \rangle}{ (t-1)^{\ell(\nu)}}\;,\quad 
d=\frac{|\mu|+|\nu|-|\lambda|}{n}\;,
\end{equation}
where $H_r=H_r(t,-1)$ are the operators defined in \eqref{bigH} with $a=t$, $b=-1$ and $d\in\Z_{\ge 0}$ (otherwise the matrix element vanishes). Note that we have used the conjugate partitions $\lambda',\mu'\in\cP^+_{n-k,n}$ in the definition \eqref{cylchi}, the reason for this will become apparent below when relating the cylindric characters to cylindric Schur functions; see Corollary \ref{cor:cylchi2cylschur}.  The definition \eqref{cylchi} should be seen as a generalisation of the identity \eqref{Ainv2chi}. Setting $\mu=\0$ we obtain the cylindric non-skew Hecke characters from Theorem \ref{thm:main}, $\chi_t^{\lambda[d]}=\chi_t^{\lambda/d/\0}$, which only depend on the cylindric loop $\lambda[d]$. These loops naturally occur in connection with the shifted level $n$ action of the extended affine symmetric group $\hat S_k$ on the weight lattice; see e.g. \cite{korff2018positive}.%

\begin{lemma}[cylindric Murnaghan-Nakayama rule]\label{lem:cylMNrule}
The cylindric Hecke characters \eqref{cylchi} satisfy the following relations
\begin{align*}
\text{\rm (i)}\quad &
\chi_{t}^{\lambda/d/\mu}(\nu_1,\ldots,\nu_i+n,\ldots,\nu_\ell)=
(-1)^{k-1}\;\chi_{t}^{\lambda/d-1/\mu}(\nu),\quad i=1,\ldots,\ell\\
\text{\rm (ii)}\quad &
\chi_{t}^{\lambda/d/\mu}(\nu,m)=
\left\{
\begin{array}{ll}
\sum\limits_{\rho\in\cP^+_{k,n}}\chi_{t}^{\lambda/d/\rho}(\nu)\chi_{t}^{\rho/\mu}(m)+
\sum\limits_{\rho\in\cP^+_{k,n}}\chi_{t}^{\lambda/d-1/\rho}(\nu)\chi_{t}^{\rho/1/\mu}(m),& 0<m< n\\
(-1)^k\frac{t^{n-k}-1}{t-1}\,\chi_{t}^{\lambda/d-1/\mu}(\nu),& m=n
\end{array}
\right.\;,
\end{align*}
where
\[
\chi_{t}^{\rho/\mu}(m)=\left\{
\begin{array}{ll}
(t-1)^{\#(\rho/\mu)-1}\prod_{h\in\rho/\mu}(-1)^{r(h)-1}t^{c(h)-1}, & \mu\subset\rho\\
0, &\text{ else}
\end{array}\right.
\]
is the non-cylindric Hecke character and
\[
\chi_{t}^{\rho/1/\mu}(m)=
\left\{\begin{array}{ll}
(t-1)^{\#(\rho/1/\mu)-1}\prod\limits_{h\in\rho/1/\mu}(-1)^{r(h)-1}t^{c(h)-1},& \mu[0]\subset\rho[1]\\
0,&\text{ else}
\end{array}
\right.
\]
is the cylindric part. The latter vanishes unless $|\mu|+m-|\rho|=n$. The two properties (i) and (ii) suffice to compute cylindric Hecke characters recursively.
\end{lemma}

\begin{proof}
The first property follows directly from the definition of the cylindric Hecke characters and \eqref{bigH},
 $
\langle\lambda'|H_{\nu_1}\cdots H_{\nu_i+n}\cdots H_{\nu_\ell}|\mu'\rangle =
(-1)^{k-1}q\langle\lambda'|H_{\nu_1}\cdots H_{\nu_i}\cdots H_{\nu_\ell}|\mu'\rangle
$. 
To prove the second property we first observe that %
$
\langle\lambda'|H_{\nu}H_r|\mu'\rangle =\sum_{\rho\in\cP^+_{k,n}}\langle\lambda'|H_{\nu}|\rho'\rangle\langle\rho'| H_{r}|\mu'\rangle 
$. 
Because of the first property we may assume that $0<r\le n$. But then $H_r=\tau_r$ for $r<n$ and $H_n=(-1)^{k}q(t^{n-k}-1)$ according to its definition \eqref{bigH}.
\end{proof}

Note that for $d=0$ property (i) is void and the second sum in property (ii) becomes zero, thus yielding the known recurrence relation for (non-cylindric) skew Hecke characters $\chi_t^{\lambda/\mu}$ \cite{ram1991}.

\begin{remark}\rm
While we have been working over $\C(t)$ we point out that the symmetric group characters can be recovered in the limit $t\to1$ by noting the identities \cite[Lem 4.8]{wan2015}
\[
\lim_{t\to 1}\frac{h_\mu[(t-1)Y]}{(t-1)^{\ell(\mu)}}=p_\mu[Y]\quad\text{ and }\quad
\lim_{t\to 1}(t-1)^{\ell(\mu)}m_\mu[(t-1)Y]=\frac{p_\mu[Y]}{z_\mu}\;,
\]
where $z_\lambda=\prod_{i>0}m_i(\lambda)!i^{m_i(\lambda)}$. Denote by $\chi^{\lambda/d/\mu}(\nu)=\lim_{t\to 1}\chi^{\lambda/d/\mu}(\nu)$ the cylindric version of a skew character for the symmetric group. The cylindric Murnaghan-Nakayama rule (ii) then simplifies to
\[
\chi^{\lambda/d/\mu}(\nu,m)=
\left\{
\begin{array}{ll}
\sum\limits_{\rho\in\cP^+_{k,n}}(-1)^{r(\rho/\mu)-1}\chi^{\lambda/d/\rho}(\nu)+
\sum\limits_{\rho\in\cP^+_{k,n}}(-1)^{r(\rho/1/\mu)-1}\chi^{\lambda/d-1/\rho}(\nu),& 0<m< n\\
(-1)^k(n-k)\,\chi^{\lambda/d-1/\mu}(\nu),& m=n
\end{array}
\right.
\]
where the first sum runs over all partitions $\rho\in\cP^+_{k,n}$ such that $\rho/\mu$ is an unbroken rim hook of length $m$ and the second sum runs over all partitions such that $\rho/1/\mu$ is an unbroken cylindric rim hook of length $m$. Thus, the combinatorics to compute the characters considerably simplifies when $t\to 1$, while the co-product expansion \eqref{chi2GW} remains valid; c.f. the discussion in \cite{korff2018positive}. However, the connection with the asymmetric six-vertex model is lost: only lattice configurations where each lattice row has at most one vertex with Boltzmann weight $\omega_6$ (and thus also a vertex of weight $\omega_5$ according to Rule 1) survive in this limit and, thus, only a small subset of six-vertex lattice configurations remains.
\end{remark}

\subsection{Bethe ansatz formula}
As an alternative to the cylindric Murnaghan-Nakayama rule of Lemma \ref{lem:cylMNrule} we obtain from the Bethe ansatz the following expression for cylindric Hecke characters in terms of the Bethe roots \eqref{BA_solns}.
\begin{lemma}\label{lem:Bethe2chi}
We have the following alternative expression for cylindric Hecke characters,
\begin{equation}\label{Bethe2chi}
q^d \chi_{t}^{\lambda/d/\mu}(\nu)=
\sum_{\xi\in\Xi_{n-k,n}}\frac{h_\nu(\xi;t,-1)}{(t-1)^{\ell(\nu)}}\;\frac{s_{\lambda'}(\xi^{-1})s_{\mu'}(\xi)}{n^k\prod_{i<j}|\xi_i-\xi_j|^{-2}}\;,
\end{equation}
where the sum runs over all solutions $\xi$ of the Bethe ansatz equations \eqref{freeBAE} with $k$ replaced by $n-k$. It follows that $\chi_{t}^{\lambda/d/\mu}(\nu)=0$ unless $|\lambda|+dn=|\mu|+|\nu|$.
\end{lemma}

\begin{proof}
A direct computation using the Bethe ansatz eigenbasis \eqref{Bethevector} of the operators $H_r=H_r(t,-1)$,
\begin{eqnarray*}
\langle\lambda'|H_{\nu}|\mu'\rangle =
\sum_{\xi}\langle\lambda'|H_{\nu}|\xi\rangle\langle\xi|\mu'\rangle
= \sum_{\xi} h_{\nu}(\xi;t,-1)\;\frac{s_{\lambda'}(\xi^{-1})s_{\mu'}(\xi)}{n^{n-k}\prod_{i<j}|\xi_i-\xi_j|^{-2}}\;,
\end{eqnarray*}
where the sums run over all the solutions $\xi$ of \eqref{freeBAE}. Here we have made use of the relation \eqref{complete}. To deduce the vanishing of the cylindric Hecke character for $|\lambda|+dn\neq |\mu|+|\nu|$ we observe that the symmetric polynomials in \eqref{Bethe2chi} are all homogeneous of degree $|\nu|$, $|\lambda|$, $|\mu|$, respectively. It then follows from the explicit dependence of the Bethe roots $\xi$ on the twist parameter $q^{\frac{1}{n}}$ given in \eqref{BA_solns}, that the $q$-dependence of the right hand side of \eqref{Bethe2chi} is given by an overall factor $q^{\frac{|\nu|+|\mu|-|\lambda|}{n}}$. Recall that $\bar q^{\frac{1}{n}}=q^{-\frac{1}{n}}$. To arrive at the assertion we note that $d\in\Z_{\ge 0}$ on the left hand side of \eqref{Bethe2chi} according to the definition \eqref{cylchi} and because the matrix elements of $H(x;t,-1)$ are formal power series in $q$.
\end{proof}
The expression \eqref{Bethe2chi} is stating that an alternative definition of the cylindric Hecke characters can be given via the (localised) quotient ring from Lemma \ref{lem:6vquotient} when specialising to $a=t,b=-1$ (and replacing $k$ with $n-k$).
\begin{corollary}
Recall the definition of the quotient ring $R_{k,n}[t,-1,\Delta_{t,-1}^{-1}]$ from Lemma \ref{lem:6vquotient}. Then we have for any $\mu\in\cP^+$ and $\lambda,\nu\in\cP^+_{k,n}$ the following multiplication formula in this ring (c.f. \eqref{skewtFrobenius}),
\begin{equation}\label{cylskewtFrobenius}
h_\mu[(t-1)Y]\,s_{\nu'}[Y]=\sum_{\lambda\in\cP^+_{k,n}}q^d\chi^{\lambda/d/\nu}_t(\mu)
(t-1)^{\ell(\mu)}s_{\lambda'}[Y]\;,
\end{equation}
where $d n=|\mu|+|\nu|-|\lambda|$.
\end{corollary}
\begin{proof}
Employing the ring isomorphism from Proposition \ref{prop:evaluate} and Lemma \ref{lem:6vquotient} this is a direct consequence of the previous formula \eqref{Bethe2chi} and \eqref{skewtFrobenius}.
\end{proof}

Since the localised quotient ring is a $t$-deformation of the quantum cohomology ring, see Lemma \ref{lem:6vquotient}, the following is now not a surprise.

\begin{corollary}
We have the following decomposition of skew cylindric Hecke characters into non-skew ones,
\begin{equation}\label{cylchi2GW}
\chi_{t}^{\lambda/d/\mu}=\sum_{d'=0}^d\;\;\sum_{\nu\in\cP^+_{k,n}}C^{\lambda,d-d'}_{\mu\nu} \chi_{t}^{\nu/d'/\0}\;,
\end{equation}
where $C^{\lambda,d}_{\mu\nu}$ are the 3-point genus 0 Gromov-Witten invariants.
\end{corollary}
\begin{proof}
The proof employs once more the Bertram-Vafa-Intriligator formula \eqref{BVI}. %
Inserting the latter in the right hand side of the asserted equality we find for any partition $\rho$ that
\begin{multline*}
q^d\sum_{\nu\in\cP^+_{k,n}}C^{\lambda',d-d'}_{\mu'\nu'} \chi_{t}^{\nu/d'/\0}(\rho) =
\sum_{\xi}\sum_{\xi'}\frac{s_{\lambda'}(\xi^{-1})s_{\mu'}(\xi)h_{\rho}(\xi';t,-1)}{n^{2(n-k)}\prod_{i<j}|\xi_i-\xi_j|^{-2}|\xi'_i-\xi'_j|^{-2}}
\sum_{\nu\in\cP^+_{k,n}}s_{\nu'}(\xi)s_{\nu'}((\xi')^{-1})\\
=\sum_{\xi}\frac{s_{\lambda'}(\xi^{-1})s_{\mu'}(\xi)h_{\rho}(\xi;t,-1)}{n^{n-k}\prod_{i<j}|\xi_i-\xi_j|^{-2}}
= q^d \chi_{t}^{\lambda/d/\mu}(\rho),
\end{multline*}
where $d'n=|\rho|-|\nu|$ and we have again made use of \eqref{complete} with $k$ replaced by $n-k$. The asserted expansion then follows by recalling that $C^{\lambda',d}_{\mu'\nu'}=C^{\lambda,d}_{\mu\nu}$ as the rings $qH^*(\op{Gr}_k(\C^n))$ and $qH^*(\op{Gr}_{n-k}(\C^n))$ are canonically isomorphic.
\end{proof}

\subsection{The coalgebra of cylindric Hecke characters}

First we establish that the cylindric Hecke characters \eqref{cylchi} are virtual characters, i.e. integer linear combinations of the irreducible characters in $\cR(t)$.

\begin{proposition}[rim hook expansion for cylindric Hecke characters]
We have the following expansion of cylindric skew Hecke characters into (non-cylindric) skew Hecke characters,
\begin{equation}\label{virtual_chi}
\chi^{\lambda/d/\mu}_{t}=\sum_{\nu}\varepsilon(\nu'/\lambda')\chi^{\nu/\mu}_{t},
\end{equation}
where the sum runs over all partitions $\nu\in\cP^+$ of $n$-weight $d$ with $\dot\nu=\lambda$ and which satisfy in addition that $\nu_1=\ell(\nu')\le k$. The coefficient $\varepsilon(\nu'/\lambda')\in\{0,\pm 1\}$ is the same as the coefficient in \eqref{rimhookalg} and \eqref{fermi_pi}. 
\end{proposition}
Setting $\mu=\0$ in \eqref{virtual_chi} we obtain the expansion of cylindric non-skew characters $\chi_t^{\lambda[d]}=\chi_t^{\lambda/d/\0}$ into irreducible Hecke characters $\chi_t^{\nu}$.
\begin{proof}
Let $\mu\in\cP^+_{k,n}$. Since $\mu$ is an $n$-core, we have
\begin{multline*}
H_\nu\pi'_{n-k,n}(\sigma(\mu',n-k))= H_\nu v_{\mu'} = \sum_{\lambda'\in\cP^+_{n-k,n}}\chi^{\lambda/d/\mu}_{t}(\nu)(t-1)^{\ell(\nu)}v_{\lambda'}\\
= (-1)^{|\nu|}\pi'_{n-k,n}(A^{-1}_\nu(t)\sigma(\mu',k))
=\sum_{\rho\in\cP^+} \chi^{\rho'/\mu'}_{t}(\nu)(t-1)^{\ell(\nu)}\pi'_{n-k,n}(\sigma(\rho,n-k))\;.
\end{multline*}
Here we have used \eqref{Ainv2H} to arrive at the second line, noting that $H(-x;-t,1)=H(x;t,-1)$. Comparing coefficients in both expansions after applying \eqref{fermi_pi} in the last line yields the desired identity.
\end{proof}

\begin{lemma}
The cylindric (non-skew) Hecke characters $\{\chi_{t}^{\lambda[d]}~|~d\in\Z_{\ge 0},\;\lambda\in\cP^+_{k,n}\}$ are linearly independent.
\end{lemma}

\begin{proof}
According to \eqref{virtual_chi} each pair $\chi^{\lambda[d]}_{t}=\chi^{\lambda/d/\0}_{t}$, $\chi^{\mu[d']}_{t}=\chi^{\mu/d'/\0}_{t}$ with $(\lambda,d)\neq(\mu,d')$ has linear expansions into mutually disjoint sets of irreducible Hecke characters. 
\end{proof}

We have at last assembled all the ingredients to prove Theorem \ref{thm:main}.

\begin{proof}[Proof of Theorem \ref{thm:main}]
Let $\lambda\in\cP^+_{k,n}$ and $d\ge 0$. Then according to \eqref{virtual_chi} we have that $\chi^{\lambda[d]}_{t}=\chi_t^{\lambda/d/\0}\in\cR^m(t)$ with $m=|\lambda|+dn$. Let $m=m'+m''$ be any decomposition and $\alpha\vdash m'$, $\beta\vdash m''$ arbitrary but fixed. Then
\begin{eqnarray*}
q^d\chi_t^{\lambda/d/\0}(\alpha,\beta) 
&=&\sum_{\mu\in\cP^+_{k,n}}\frac{\langle\lambda'|H_\alpha|\mu'\rangle}{(t-1)^{\ell(\alpha)}}
\frac{\langle\mu'| H_\beta|\0\rangle}{(t-1)^{\ell(\beta)}}
= \sum_{\mu\in\cP^+_{k,n}} q^{d'}\chi_t^{\lambda/d'/\mu}(\alpha)q^{d''}\chi_t^{\mu/d''/\0}(\beta),
\end{eqnarray*}
where $d'=(|\mu|+m'-|\lambda|)/n$ and $d''=(m''-|\mu|)/n$ according to Lemma \ref{lem:Bethe2chi}. Inserting the expansion \eqref{cylchi2GW} for the cylindric skew character $\chi_t^{\lambda/d'/\mu}$ yields the desired expansion formula \eqref{chi2GW}.
\end{proof}

\subsection{Cylindric Schur functions}
We complete this section by relating the cylindric Hecke characters to cylindric Schur functions and, thus, motivating the definition \eqref{cylchi}. Cylindric Schur functions have been considered by several authors previously in the literature; see e.g. \cite{gessel1997, postnikov2005, mcnamara2006,lam2006} and \cite{korff2018positive}. We recall their definition. 

Let $\lambda,\mu\in\cP^+_{k,n}$ and $d\ge0$ such that $\lambda/d/\mu$ is a well-defined cylindric skew diagram. A map $\cT:\lambda/d/\mu\to\N$ is called a {\em cylindric tableau} of shape $\lambda/d/\mu$ if it satisfies $\cT(i,j)=\cT(i+k,j-n+k)$ and
\[
\left\{
\begin{array}{l}
\cT(i+1,j)>\cT(i,j)\quad\text{(strictly increasing in columns)}\\
\cT(i,j+1)\ge \cT(i,j)\quad\text{(weakly increasing in rows)}
\end{array}\right.\;.
\]
Similar to the non-cylindric case, we refer to the vector $\theta(\cT)=(\theta_1,\theta_2,\ldots)$ with
\[
\theta_m=\#\{(i,j)\in\lambda/d/\mu~|~\cT(i,j)=m\text{ and }0\le j\le n-k\}
\]
as the {\em weight} of $\cT$. The {\em cylindric skew Schur function} $s_{\lambda/d/\mu}\in\Lambda$ is then defined as the weighted sum over all cylindric tableaux of shape $\lambda/d/\mu$,
\begin{equation}\label{cylSchur}
s_{\lambda/d/\mu}[Y]=\sum_{\cT}y^{\theta(\cT)}\;.
\end{equation}
The following result is due to McNamara \cite{mcnamara2006}.
\begin{theorem}[McNamara]
For any cylindric skew shape $\lambda/d/\mu$ with $\lambda,\mu\in\cP^+_{k,n}$, $d\ge 0$, one has the following expansion of cylindric skew Schur functions into skew Schur functions,
\begin{equation}\label{cylSchur2Schur}
s_{\lambda/d/\mu}=\sum_{\nu}\varepsilon(\nu/\lambda)s_{\nu/\mu}
\end{equation}
where the sum runs over all partitions $\nu$ of $n$-weight $d$ and $\dot\nu=\lambda$ and which satisfy in addition that $\nu_1\le n-k$. The coefficient $\varepsilon(\nu'/\lambda')$ is the same as in \eqref{virtual_chi}.
\end{theorem}

Consider the map $\jmath_{k,n}:\bigwedge^k\C^n(t)\otimes\C[\!\![q]\!\!]\to\Lambda(t)\otimes\C[\!\![q]\!\!]$ defined via
\begin{equation}\label{cylBF}
v\otimes q^m\mapsto \langle v,H(x_1;a,b)H(x_2;a,b)\cdots v_\0\rangle
=\sum_{\mu}q^m\langle v,H_\mu v_\0\rangle m_\mu[X]\,,
\end{equation}
where the matrix element is a formal power series in $q$ with coefficients in $\C(t)$.

\begin{proposition}
Let $\lambda\in\cP^+_{k,n}$. Then we have the following identity in terms of formal power series in $q$,
\begin{equation}\label{6v2cylSchur}
\langle\lambda'|H(x_1;t,-1)H(x_2;t,-1)\cdots|\0\rangle=\sum_{d\ge 0}q^d s_{\lambda/d/\0}[(t-1)X]\;,
\end{equation}
where $s_{\lambda/d/\0}$ is the cylindric Schur function defined in  \eqref{cylSchur}
\end{proposition}
\begin{proof}
Making use of the completeness of the Bethe ansatz we can work in the eigenbasis \eqref{Bethevector} of the operators $H_r(a,b)$ with $a=t,b=-1$ and apply the Cauchy identity \cite{macdonald1998}
\[
\prod_{i\ge 1}\prod_{j=1}^{n-k}\frac{1+t x_i y_j}{1+x_i y_j}=
\sum_{\mu\in\cP^+}(-1)^{|\mu|}s_{\mu'}[(1-t)X]s_{\mu'}(y_1,\ldots,y_{n-k})
\]
to arrive at
\[
\langle \lambda'|H(x_1;t,-1)H(x_2;t,-1)\cdots |\0\rangle
=\sum_{\mu}\langle\lambda'|S_{\mu'}|\0\rangle (-1)^{|\mu|}s_{\mu'}[(1-t)X]\;,
\]
where $S_{\mu'}$ is the unique operator with eigenvalues $s_{\mu'}(\xi)$ and $\xi$ is a solution to the Bethe ansatz equations \eqref{freeBAE}. Thus, we can employ the rim hook algorithm \eqref{fermi_pi} to compute the matrix element
\[
\langle\lambda'|S_{\mu'}|\0\rangle=\left\{
\begin{array}{ll}
q^{d}(-1)^{d(n-k)+\sum_{i=1}^dr(h_i)}, & \dot\mu=\lambda\text{ and }\ell(\mu')=\mu_1\le n-k\\
0,&\text{else}
\end{array}
\right.\;,
\]
where $d=d(\mu')$ is the $n$-weight of $\mu'$ and the $h_i$ denote again the $n$-rim hooks removed to obtain the $n$-core $\dot\mu'$. As before we denote the sign factor resulting from the rim hook algorithm by $\varepsilon(\mu'/\lambda')$. Then
\[
\langle\lambda'|H(x_1;t,-1)H(x_2;t,-1)\cdots|\0\rangle
=\sum_{\mu_1\le k,\dot \mu=\lambda'}q^{d(\mu)}\varepsilon(\mu'/\lambda') (-1)^{|\mu|}s_{\mu'}[(1-t)X]\;,
\]
Noting that $(-1)^{|\mu|}s_{\mu'}[(1-t)X]=s_\mu[(t-1)X]$, the assertion now follows from the expansion \eqref{cylSchur2Schur} of the cylindric Schur functions into Schur functions.
\end{proof}

\begin{corollary}\label{cor:cylchi2cylschur}
Let $\lambda\in\cP^+_{k,n}$ and $d\ge 0$. Then %
$\ch_t(\chi^{\lambda/d/\0}_t)=s_{\lambda/d/\0}$, %
where $\ch_t$ is the quantum characteristic map \eqref{chart}.
\end{corollary}
\begin{proof}
Using the following expansion into monomial symmetric functions,
\begin{equation*}
\langle\lambda|H(x_1;t,-1)H(x_2;t,-1)\cdots|\0\rangle=
\sum_{d\ge 0}q^d\sum_{\mu}\chi_t^{\lambda'/d/\0}(\mu)(t-1)^{\ell(\mu)}m_\mu[X],
\end{equation*}
the result follows from \eqref{6v2cylSchur} upon making the (plethystic) variable substitution $X=Y/(t-1)$.

\end{proof}
Since the quantum characteristic map $\ch_t:\cR(t)\to\Lambda(t)$ is a Hopf algebra isomorphism, it follows from the last corollary that the cylindric (non-skew) Schur functions $s_{\lambda/d/\0}$ span a sub-coalgebra in $\Lambda(t)$ whose structure constants are also the Gomov-Witten invariants; c.f. \cite{korff2018positive}.

\end{document}